\documentclass[a4paper,11pt,leqno]{amsart}
\usepackage[margin=2cm,bindingoffset=0cm,nofoot]{geometry} 

\usepackage{amsmath, amssymb, amsthm, stmaryrd, mathrsfs}
\usepackage{euscript}
\usepackage{verbatim}
\usepackage{graphicx}
\usepackage{enumerate}
\numberwithin{equation}{section}

\theoremstyle{plain}

\newtheorem{theorem}{Theorem}[section]
\newtheorem{prop}[theorem]{Proposition}
\newtheorem{cor}[theorem]{Corollary}
\newtheorem{lemma}[theorem]{Lemma}

\theoremstyle{definition}

\theoremstyle{remark}
\newtheorem{rk}[theorem]{Remark}
\newtheorem{rks}[theorem]{Remarks}

\DeclareMathOperator{\supp}{supp}

\newcommand{\Mbar}{\overline{M}}

\newcommand{\bR}{\mathbb{R}}
\newcommand{\bB}{\mathbb{B}}

\newcommand{\bN}{\mathbb{N}}
\newcommand{\bH}{\mathbb{H}}

\newcommand{\cC}{\mathcal{C}}

\newcommand{\gbar}{\overline{g}}
\newcommand{\ghat}{\widehat{g}}

\renewcommand{\hbar}{\overline{h}}

\newcommand{\gammatil}{\widetilde{\gamma}}
\newcommand{\Qtil}{\widetilde{Q}}
\newcommand{\Ttil}{\widetilde{T}}
\newcommand{\gtil}{\tilde{g}}

\def\into{\hookrightarrow}
\newcommand{\phibar}{\overline{\phi}}

\newcommand{\phitil}{\widetilde{\phi}}
\newcommand{\psitil}{\widetilde{\psi}}
\newcommand{\sigmatil}{\widetilde{\sigma}}
\newcommand{\util}{\widetilde{u}}
\newcommand{\thetatil}{\widetilde{\theta}}
\newcommand{\gbrev}{\breve{g}}
\newcommand{\Ebrev}{\breve{E}}
\newcommand{\Pbrev}{\breve{P}}
\newcommand{\wbrev}{\breve{w}}
\newcommand{\phat}{\hat{p}}
\newcommand{\ScriPlus}{\mathscr{I}^+}

\DeclareMathOperator{\tr}{tr}
\DeclareMathOperator{\divg}{div}

\DeclareMathOperator{\vol}{Vol}

\DeclareMathOperator{\id}{Id}
\DeclareMathOperator{\Ker}{Ker}
\DeclareMathOperator{\im}{Im}

\newcommand{\lie}{\mathcal{L}}
\newcommand{\rlie}{\mathring{\mathcal{L}}}

\newcommand{\hg}{\widehat{g}}

\newcommand{\riem}{\mathrm{R}}

\newcommand{\riemuddd}[4]{\riem^{#1}_{\phantom{#1} #2 #3 #4}}

\newcommand{\ric}{\mathrm{Ric}}

\newcommand{\ricdd}[2]{\ric_{#1 #2}}
\newcommand{\ricud}[2]{\ric^{#1}_{\phantom{#1} #2}}

\newcommand{\ricuu}[2]{\ric^{#1 #2}}

\newcommand{\scal}{\mathrm{Scal}}

\newcommand{\hscal}{\widehat{\scal}}

\newcommand{\grad}[1]{\nabla_{#1}}

\begin{document}
\title[Non CMC solutions on an AH manifold]{A large class of non constant mean curvature solutions of the Einstein constraint equations on an asymptotically hyperbolic manifold}
\author[Romain Gicquaud and Anna Sakovich]{Romain Gicquaud and Anna Sakovich}
\date{\today}
\address{Current Address: \newline
Laboratoire de Math\'ematiques et de Physique Th\'eorique \newline
UFR Sciences et Technologie\newline
Facult\'e Fran\c cois Rabelais\newline
Parc de Grandmont \newline
37300 Tours, \newline \newline
Institutionen f\"or Matematik\newline
Kungliga Tekniska H\"ogskolan\newline
100 44 Stockholm\newline
Sweden}

\begin{abstract} We construct solutions of the constraint equation with non constant mean curvature on an asymptotically hyperbolic manifold by the conformal method. Our approach consists in decreasing a certain exponent appearing in the equations, constructing solutions of these sub-critical equations and then in letting the exponent tend to its true value. We prove that the solutions of the sub-critical equations remain bounded which yields solutions of the constraint equation unless a certain limit equation admits a non-trivial solution. Finally, we give conditions which ensure that the limit equation admits no non-trivial solution.
\end{abstract}

\maketitle
\tableofcontents

\section{Introduction}\label{secIntro}
An important issue in General Relativity is the study of the Cauchy problem. Given a globally hyperbolic space-time $(\mathcal{M}, g)$ satisfying the vacuum Einstein equations \[\ricdd{\mu}{\nu} - \frac{\scal}{2} g_{\mu\nu} = 0,\] and a spacelike hypersurface $M \subset \mathcal{M}$, two natural quantities can be defined: the induced metric on $M$ (still denoted by $g$), which can be understood as the magnetic part\footnote{More exactly, these are the Christoffel symbols of the induced metric that can be considered as the magnetic part, the metric $g$ is the analogue of the potential vector.} of the gravitational field at time ``t = 0'', and the second fundamental form of the embedding $M \subset \mathcal{M}$, which corresponds to the electric part. General Relativity is a constrained Hamiltonian system: the induced metric $g$ and $K$ cannot be arbitrary, they are linked by the following two equations:

\begin{align}
\label{eqHamiltonianConstraint}
\scal_g + \left(\tr_g K \right)^2 - \left| K \right|^2_g & = 0 \qquad \text{(Hamiltonian constraint),}\\
\label{eqMomentumConstraint}
\nabla^i K_{ij} - \nabla_j \left( \tr_g K \right) & = 0 \qquad \text{(Momentum constraint).}
\end{align}

These equations can be seen either as a consequence of the Gauss and Codazzi equations (see e.g. \cite{BartnikIsenberg}), or as a consequence of the Hamiltonian analysis of the Einstein-Hilbert action (described e.g. in \cite{Carlip}). If the initial data $(g,K)$ on $M$ satisfies (\ref{eqHamiltonianConstraint})-(\ref{eqMomentumConstraint}), then the Cauchy problem is well posed. This is the content of the celebrated Choquet-Bruhat and Choquet-Bruhat--Geroch theorems \cite{CB1, CB2}, which assert that in this case there exists an essentially unique space-time $(\mathcal{M}, g)$ which is the ``time evolution'' of $(M, g, K)$. We refer the reader to \cite{FriedrichRendall} for a recent review of this problem, see also \cite{Hawking}, \cite{Wald} and \cite{ChoquetBruhat}. As a consequence, much work has been done in constructing and classifying triples $(M, g, K)$ which satisfy the constraint equations \eqref{eqHamiltonianConstraint} and \eqref{eqMomentumConstraint}. Most of the results are summarized in \cite{BartnikIsenberg}.\\

The principal method used to construct solutions of the constraint equations is the conformal method, also known as the Choquet-Bruhuat--Lichnerowicz--York method, which will be described in Section \ref{secConfMeth}. We refer the reader to \cite{BartnikIsenberg} and \cite{ChoquetBruhat} for other constructions of solutions of the constraint equations. This method has proved to be very successful for the construction of constant mean curvature (CMC) solutions of the constraint equations, i.e. such that $\tr K$ is constant on $M$, because in this case the problem is reduced to solving an essentially uncoupled system of elliptic partial differential equations. However, less is known when the mean curvature is non constant. Most of the results deal with the case when $d (\tr K)$ is small.\\

In three recent articles \cite{HNT1, HNT2} and \cite{MaxwellNonCMC}, a new construction of non-CMC solutions of the constraint equations on compact manifolds has been introduced. It is based on a generalization of the monotony method to systems of PDE and it still requires another smallness condition (smallness of the TT-tensor $\sigma_0$, see Section \ref{secConfMeth}). The ``global super-solution'' used in this construction is close to zero so the method applies straightforwardly neither to the asymptotically Euclidean nor to the asymptotically hyperbolic cases since on these manifolds the conformal factor $\phi$ has to satisfy $\phi \to 1$ at infinity. Another drawback of constructing a conformal factor that is close to zero is that it creates ``small'' universes. Nevertheless, the aforesaid method will play an important role in this work since we will adapt it to construct solutions of the sub-critical constraint equations.\\

In this paper we will be interested in asymptotically isotropic initial data, that is to say solutions $(M, g, K)$ of the constraint equations such that $(M, g)$ is a conformally compact manifold and $K \pm g \to 0$ at infinity. As a consequence of the Hamiltonian constraint, this implies that the sectional curvature of $g$ tends to $-1$ at infinity. Conformally compact manifolds whose sectional curvature tends to $-1$ at infinity are called asymptotically hyperbolic manifolds. The construction of asymptotically isotropic CMC solutions by the conformal method has been developed in \cite{AnderssonChrusciel}, \cite{Gicquaud} and \cite{Sakovich}, while the case $d (\tr K)$ small is studied in \cite{IsenbergPark}. These initial data are of particular interest in the numerical study of gravitational radiation which is well-defined only at future null infinity ($\ScriPlus$), see e.g. the recent article \cite{BuchmanPfeifferBardeen}.\\

The approach which we follow is similar to the one introduced in \cite{DahlGicquaudHumbert}. It is inspired by the solution of the Yamabe problem (see e.g. \cite{LeeParker}): decrease a certain exponent that appears in the equations so that they become sub-critical and study how the sub-critical solutions behave when the exponent tends to its true value. Note however that, contrary to the Yamabe problem, the constraint equations do not come from a variational principle, thus there is no ``good'' interpretation of their criticality. Our main result (Theorem \ref{thmMainTheorem}) is an analogue of \cite{DahlGicquaudHumbert}: if a certain limit equation \eqref{eqLimit} admits no non-trivial solution, then the set of solutions of the equations of the conformal method \eqref{eqLichnerowicz} and \eqref{eqVector} is non empty and compact. We also give conditions ensuring that the limit equation admits no non-trivial solution (Propositions \ref{propLnNearCMC} and \ref{propLinftyNearCMC}). In particular, if the manifold on which one wants to solve the equations of the conformal method is Einstein, Proposition \ref{propLinftyNearCMC} provides a large upper bound for the $L^\infty$-norm of the variation of the (prescribed) mean curvature under which the limit equation admits no non-trivial solution (see Remark \ref{rkEinstein}).\\

The outline of this article is as follows. In Section \ref{secPreliminaries} we introduce the main notions to be used throughout the paper. In particular, Section \ref{secConfMeth} describes the conformal method which we use for constructing solutions of the constraint equations. Section \ref{secAH} defines the class of asymptotically hyperbolic manifolds together with weighted H\"older and Sobolev spaces and a new class of weighted local Sobolev spaces. With all those preliminaries at hand, we can then state the main result of this article, Theorem \ref{thmMainTheorem}, whose proof occupies the next three sections. In Section \ref{secSubCritSolutions} we prove that the sub-critical equations admit solutions (Proposition \ref{propSubCrit}) and study their properties (Propositions \ref{propSubCritBehav} and \ref{propSubCritBehav2}). We also show that the method we have developed yields a new proof of existence and uniqueness of ``near CMC'' solutions of the equations of the conformal method (Theorem \ref{thmNearCMC}). The second step in the proof of Theorem \ref{thmMainTheorem} (Section \ref{secBlowUp}) consists in letting the regularization parameter tend to zero and in studying the behavior of the solutions. The main result of this section is Corollary \ref{corFundamentalCorollary}. In Section \ref{secHigherRegularity}, we show that allowing for more regularity of the conformal data, one obtains solutions with higher regularity. Finally, in Section \ref{secLimit}, we give conditions under which the limit equation admits no non-zero solution. Appendix \ref{secFredholm} contains the proof of the Fredholm theorem for the weighted local Sobolev spaces, while in  Appendix \ref{secNoL2CKV} we present the proofs of $L^2$-estimate at infinity for the vector Laplacian and of the fact that there are no $L^2$ conformal Killing vector fields on asymptotically hyperbolic manifolds. As a corollary, we prove the positivity of a certain Sobolev constant which appears in the study of the limit equation.\\

\noindent\textit{Acknowledgments.} We are grateful to Mattias Dahl, Emmanuel Humbert and Laurent V\'eron for useful discussions and advices. We are also grateful to Harald Pfeiffer for pointing us the reference \cite{BuchmanPfeifferBardeen}. We also thank Erwann Delay and Eric Bahuaud for their interest in this work. Finally, we thank Piotr Chru{\'s}ciel and both referees for their useful remarks which have improved the presentation of this paper.

\section{Preliminaries}\label{secPreliminaries}
\subsection{The conformal method}\label{secConfMeth}
A natural way to understand the constraint equations \eqref{eqHamiltonianConstraint} and \eqref{eqMomentumConstraint} is to consider the Hamiltonian constraint \eqref{eqHamiltonianConstraint} as a scalar equation for the metric and the momentum constraint \eqref{eqMomentumConstraint} as a vectorial equation for the second fundamental form $K$. As a consequence, to construct solutions $(M, \hg, \widehat{K})$ of the system \eqref{eqHamiltonianConstraint}-\eqref{eqMomentumConstraint}, we will look for $\hg$ in the conformal class of a given metric $g$, i.e. in the form $\hg = \phi^\kappa g$, where $\kappa = \frac{4}{n-2}$. In order to understand the structure of solutions of the momentum constraint \eqref{eqMomentumConstraint}, we decompose $\widehat{K}$ as $\widehat{K} = \tau \ghat + \widehat{\sigma}$, where $\tau = \frac{1}{n} \mathrm{tr}_{\hg} \widehat{K}$ is the mean curvature of the hypersurface $M \subset \mathcal{M}$ and $\widehat{\sigma}$ is a symmetric traceless 2-tensor. The equation \eqref{eqMomentumConstraint} then becomes

\[\ghat^{ik}\widehat{\nabla}_k \widehat{\sigma}_{ij} - (n-1) \widehat{\nabla}_j\tau = 0.\]

This equation still involves $\phi$ in the term $\ghat^{ik}\widehat{\nabla}_k \widehat{\sigma}_{ij}$. Setting $\widehat{\sigma} = \phi^{-2} \sigma$, one obtains the following equation to be solved for $\sigma$:

\[\phi^{-\kappa-2} g^{ik}\nabla_k \sigma_{ij} - (n-1) \nabla_j\tau = 0.\]

To solve this equation, one has to freeze some degrees of freedom of $\sigma$. We decompose $\sigma$ as a sum $\sigma = \sigma_0 + \sigma_1$ of a particular solution $\sigma_1$ and a solution  $\sigma_0$ of the homogeneous problem \[\nabla^i \sigma_{0ij} = 0.\] Note that a 2-tensor which is symmetric, traceless and divergence-free is called a \textbf{TT-tensor}. A construction of TT-tensors will be given in Corollary \ref{corTTTensors}. As for $\sigma_1$, it can be chosen as the traceless part of the Lie derivative of the metric in the direction of the dual of some 1-form $\psi$:

\[\sigma_{1ij} =  \rlie_{\psi^\sharp} g_{ij} = \grad{i} \psi_j + \grad{j} \psi_i - \frac{2}{n} \nabla^k \psi_k g_{ij}.\]

We will simplify the notation and write $\lie \psi = \rlie_{\psi^\sharp} g$. The momentum constraint \eqref{eqMomentumConstraint} now reads

\[g^{ik}\nabla_k \left(\lie \psi\right)_{ij} = (n-1) \nabla_j\tau.\]

Remark that the decomposition of the set of $L^2$ symmetric traceless 2-tensors into $L^2$ TT-tensors and tensors of the form $\lie \psi$ is an orthogonal decomposition. We denote by $\Delta_L$ the operator appearing on the left-hand side of the above equation: \[\Delta_L \psi_j = g^{ik}\nabla_k \left(\lie \psi\right)_{ij} = \Delta \psi_j + \nabla_i \nabla_j \psi^i - \frac{2}{n} \nabla_j \nabla_k \psi^k.\] It can be checked that this operator is elliptic. The equation \[\Delta_L \psi = (n-1) \phi^{\kappa+2} d\tau\] is called the \textbf{vector equation}.\\

We now write the Hamiltonian constraint \eqref{eqHamiltonianConstraint} in terms of $\phi$ and $\sigma$. By the conformal transformation law of scalar curvature, we get

\begin{align*}
0 & = \hscal - \left|\widehat{K}\right|^2_{\hg} + \left(\tr_{\hg} \widehat{K}\right)^2\\
	& = \phi^{-\kappa-1}\left(-\frac{4(n-1)}{n-2} \Delta \phi + \scal~\phi\right) - n \tau^2 - \left|\widehat{\sigma}\right|^2_{\hg} + n^2 \tau^2\\
	& = \phi^{-\kappa-1}\left(-\frac{4(n-1)}{n-2} \Delta \phi + \scal~\phi\right) + n (n-1) \tau^2 - \left|\sigma\right|_g^2 \phi^{-4-2\kappa}.
\end{align*}

Multiplying by $\phi^{\kappa+1}$, we finally obtain:

\[- \frac{4(n-1)}{n-2}\Delta\phi + \scal~\phi + n(n-1)\tau^2 \phi^{\kappa+1} - \left|\sigma\right|_g^2 \phi^{-3-\kappa} = 0.\]

This equation is known as the \textbf{Lichnerowicz equation}. Remark that if $\sigma = 0$, this equation corresponds to the prescribed scalar curvature equation with $\hscal = - n(n-1) \tau^2$.\\

Let us now summarize the conformal method:

\begin{itemize}
\item Fix a Riemannian manifold $(M, g)$, a function $\tau : M \to \bR$ and a symmetric, traceless, divergence-free 2-tensor $\sigma_0$.\\
\item Solve the following system in the unknowns $\psi$ and $\phi$:

\begin{align}
\label{eqLichnerowicz}
-\frac{4(n-1)}{n-2} \Delta \phi + \scal~\phi + n(n-1) \tau^2 \phi^{\kappa+1} & = \left|\sigma\right|^2_g \phi^{-\kappa-3},\\
\label{eqVector}
\Delta_L \psi & = (n-1) \phi^{\kappa+2} d \tau,
\end{align}
where $\phi : M \to \bR_+^*$ is a (strictly) positive function and $\sigma = \sigma_0 + \lie \psi$, where $\psi \in \Gamma(M, T^*M)$ is a 1-form.\\

\item Then $\hg = \phi^\kappa g$ and $\widehat{K} = \tau \hg + \phi^{-2} \sigma$ solve the constraint equations \eqref{eqHamiltonianConstraint}-\eqref{eqMomentumConstraint}.
\end{itemize}

When $\tau$ is a constant function (or, equivalently, when the hypersurface $M \subset \mathcal{M}$ has constant mean curvature), the vector equation does not involve $\phi$. In this case the solution to the constraint equations can be obtained by first solving the vector equation \eqref{eqVector}, and then the Lichnerowiz equation \eqref{eqLichnerowicz}. This particular case is described in \cite{AnderssonChrusciel}, see also \cite{Gicquaud} and \cite{Sakovich} for further results.

\subsection{Asymptotically hyperbolic manifolds and weighted spaces}\label{secAH}
In this section, we recall some of the main definitions and theorems of \cite{LeeFredholm} and also introduce a new class of function spaces defined on asymptotically hyperbolic manifolds.

\subsubsection{Definition}\label{defAH}
Let $\Mbar$ be a smooth compact manifold with boundary $\partial M$. We denote by $M$ the interior of $\Mbar$. A \textbf{defining function} for $\partial M$ is a smooth function $\rho : \Mbar \to [0;~\infty)$ such that $\rho^{-1}(0) = \partial M$ and $d\rho \neq 0$ along $\partial M$. A Riemannian metric $g$ on $M$ is called $C^{l, \beta}$-\textbf{conformally compact} if $\rho^2 g$ extends to a $C^{l, \beta}$ metric $\gbar$ on $\Mbar$. A simple calculation proves that if $g$ is $C^{l, \beta}$-conformally compact with $l + \beta \geq 2$ then the sectional curvature of $g$ satisfies \[\sec_g = -|d\rho|^2_{\gbar} + O(\rho)\] in a neighborhood of $\partial M$. As a consequence, $g$ is said to be \textbf{asymptotically hyperbolic} if $g$ is conformally compact and is such that $|d\rho|^2_{\gbar} = 1$ along $\partial M$, that is to say, if $\sec_g \to -1$ at infinity.\\

In what follows, we will denote by $M_\mu$ and $K_\mu$ the sets  $M_\mu = \{p \in M \,|\, \rho(p) < \mu\}$ and $K_\mu = M \setminus M_\mu = \{p \in M \,|\, \rho(p) \geq \mu\}$, where $\mu>0$. We remark that $K_\mu$ is a compact subset of $M$.

\subsubsection{Weighted function spaces and Fredholm theorems}
Let $l\geq 2$ be an integer and assume that $0\leq \beta<1$. In this section, we consider a $C^{l, \beta}$-asymptotically hyperbolic manifold $(M, g)$, together with an $O(n)$-subbundle $E$ of some tensor bundle $\left(T^*M\right)^{\otimes p} \otimes \left(T_*M\right)^{\otimes q}$  invariant under parallel translation (such a bundle will be called a geometric bundle). Below we introduce three types of weighted spaces of sections $u \in \Gamma(M, E)$, the first being

\begin{itemize}
\item\textbf{Weighted Sobolev spaces:} Let $0 \leq k \leq l$ be an integer, let $1 \leq p \leq \infty$ be a real number, and let $\delta \in \bR$. The weighted Sobolev space $W^{k, p}_\delta(M, E)$ is the set of sections $u \in \Gamma(M, E)$ such that $u \in W^{k, p}_{loc}$ and such that the norm
\[\left\| u \right\|_{W^{k, p}_\delta(M, E)} = \sum_{i=0}^k \left( \int_M \left| \rho^{-\delta} \nabla^{(i)} u \right|^p_g d\mu_g\right)^{\frac{1}{p}}\] is finite.
\end{itemize}

Before introducing the other two function spaces, we need to recall the definition of special charts on $M$ known as M\"obius charts (see \cite[Chapter 2]{LeeFredholm} for more details). Denote by $\mathbb{H}^n$ the hyperbolic space seen as the half-space $\{(x_1, \ldots, x_n) \in \bR^n \, \vert \, x_1 > 0 \}$ in $\mathbb{R}^n$ endowed with the metric $\gbrev = \frac{1}{x_1^2} g_{eucl}$. We first select a finite number of smooth coordinate charts $(\Omega, \Phi)$, where $\Phi = (\rho, \theta^1, \cdots, \theta^{n-1})$, on $\Mbar$ such that their domains of definition $\Omega$ cover $\partial M$. We complete this system of charts by adding a finite number of charts whose domains of definition are precompact in $M$ (i.e. do not intersect $\partial M$). Let $B_r$ be the hyperbolic ball centered at $(1,0, \cdots, 0)$ of radius $r$ in $\mathbb{H}^n$. If $p_0 \in M$ is in the reciprocal image by one of the chosen charts $(\Omega, \Phi)$, $p_0 = \Phi^{-1}(\rho_0, \theta^1_0\cdots, \theta^{n-1}_0)$, we define a M\"obius chart $\Phi^r_{p_0} : \left(\Phi_{p_0}\right)^{-1}\left(B_r\right) \to B_r$ in a neighborhood of $p_0$ by

\[\Phi_{p_0}\left(p\right) = \left(\frac{\rho(p)}{\rho_0}, \frac{\theta^1(p)-\theta^1_0}{\rho_0}, \cdots, \frac{\theta^{n-1}(p)-\theta^{n-1}_0}{\rho_0} \right).\]
Let us now define the following two classes of function spaces:

\begin{itemize}
\item\textbf{Weighted local Sobolev spaces:} Let $0 \leq k \leq l$ be an integer, let $1 \leq p \leq \infty$ be a real number, and let $\delta \in \bR$. The weighted local Sobolev space $X^{k, p}_\delta(M, E)$ is the set of sections $u \in \Gamma(M, E)$ such that $u \in W^{k, p}_{loc}$ and such that the norm
\[\left\| u \right\|_{X^{k, p}_\delta(M, E)} = \sup_{p_0 \in M} \rho^{-\delta}(p_0) \left\| \left(\left(\Phi^1_{p_0}\right)^{-1}\right)^* u\right\|_{W^{k, p}(B_1)}\] is finite.

\item\textbf{Weighted H\"older spaces:} Let an integer $k \geq 0$ and $0 \leq \alpha < 1$ be such that $k + \alpha \leq l + \beta$, and let $\delta \in \bR$. The weighted H\"older space $C^{k, \alpha}_\delta(M, E)$ is the set of sections $u \in \Gamma(M, E)$ such that $u \in C^{k, \alpha}_{loc}$ and such that the norm
\[\left\|u\right\|_{C^{k, \alpha}_\delta(M, E)} = \sup_{p_0 \in M} \rho^{-\delta}(p_0) \left\| \left(\left(\Phi^1_{p_0}\right)^{-1}\right)^* u\right\|_{C^{k, \alpha}(B_1)}\]
is finite.
\end{itemize}

We now state some properties of the weighted local Sobolev spaces that will be useful for us. These spaces stand in some sense halfway between weighted Sobolev spaces and H\"older spaces. The proof of some results will be omitted being straightforward variants of their counterparts for H\"older or Sobolev spaces (see \cite{LeeFredholm}).

\begin{lemma}[A first density result]\label{lmDensity}
Let $(M, g)$ be a $C^{l, \beta}$-asymptotically hyperbolic manifold with $l + \beta \geq 2$ and let $E \to M$ be a geometric tensor bundle. Assume that  $k \in \bN$, $k \leq l$, $p \in (1; \infty)$, and $\delta \in \bR$. Then $C^{l, \beta}_{loc}(M, E) \cap X^{k, p}_\delta(M, E)$ is dense in $X^{k, p}_\delta(M, E)$.
\end{lemma}

\begin{proof} We remark first that multiplication by $\rho^\delta$ defines an isomorphism from $X^{k, p}_0(M, E)$ to $X^{k, p}_\delta(M, E)$. By \cite[Lemma 2.2]{LeeFredholm}, $M$ can be covered by a countable number of M\"obius charts $B_i = \left(\Phi_{x_i}\right)^{-1} (B_1)$ centered at $x_i$ and of ``radius one'', and having the following property: if $\tilde{B}_i = \left(\Phi_{x_i}\right)^{-1} (B_2)$ then there exists a number $N \geq 1$ such that for each $i$ the set $\{j : \tilde{B}_i \bigcap \tilde{B}_j \neq \emptyset\}$ contains at most $N$ elements. Let $\psi: \bH^n \to \bR$ be a smooth cut-off function, such that $0 \leq \psi \leq 1$, $\psi = 1$ on $B_1$, and $\psi = 0$ outside $B_2$. For each index $i$ we set $\psi_i = \psi \circ \Phi_{x_i}$ on $\tilde{B}_i$, and $\psi_i = 0$ outside $\tilde{B}_i$. Define also \[\phi_i = \frac{\psi_i}{\sum_j \psi_j}.\] Since the $\psi_i$ are uniformly bounded in $C^{l, \beta}_0(M, \bR)$, and are such that $1 \leq \sum_j \psi_j \leq N$, we see that the functions $\phi_i$ are also uniformly bounded in $C^{l, \beta}_0(M, \bR)$. Now let us consider $u \in X^{k, p}_0(M, E)$ and $\epsilon > 0$. For each $i$ select $v_i \in C^{l, \beta}(\tilde{B}_i, E)$ such that $\left\|u - v_i\right\|_{W^{k, p}(\tilde{B}_i, E)} < \epsilon$ and define \[\tilde{u} = \sum_i \phi_i v_i.\] By construction $\tilde{u} \in C^{l, \beta}_{loc}$. We obtain the following estimate (here the constant $C$ can vary from line to line but is independent of $u$ and $\epsilon$):
\begin{align*}
\left\|u - \tilde{u}\right\|_{X^{k, p}_0(M, E)}
	& \leq C \sup_i \left\|u - \tilde{u}\right\|_{W^{k, p}(B_i, E)}\\
	& \leq C \sup_i \sum_j \left\|\phi_j (u-v_j)\right\|_{W^{k, p}(B_i, E)}\\
	& \leq C \sup_i \sum_{j | B_i \cap \tilde{B}_j \neq \emptyset} \left\|\phi_j\right\|_{C^k(\tilde{B}_j, \bR)} \left\|u-v_j\right\|_{W^{k, p}(\tilde{B}_j, E)}\\
	& \leq C \epsilon.
\end{align*}
This concludes the proof of the lemma.
\end{proof}

In fact we have not been able to prove density of $C^{l, \beta}_\delta(M, E)$ in $X^{k, p}_\delta(M, E)$. The problem is that for this we need to smooth $u \in X^{k, p}_\delta(M, E)$ an infinite number of times (once in each M\"obius chart), while we cannot control the $C^{l, \beta}$-norm of each smoothing well enough to ensure that when we glue them together, the result is in $C^{l, \beta}_\delta(M, E)$. This issue is linked to the nature of $X^{k, p}_\delta$-tensors: they locally look like Sobolev tensors but globally behave as H\"older tensors. Having this idea in mind together with the fact that the closure of the set of smooth compactly supported functions in $C^0(\bR^n)$ is the set of continuous functions tending to zero at infinity, we introduce the spaces $X^{k, p}_{\delta^+}(M, E)$ as follows. We select a smooth cut-off function $\chi : \bR_+ \to \bR$ such that $\chi = 1$ on $\left[0; \frac{1}{2}\right]$ and $\chi(r) = 0$ for any $r \geq 1$ and define

\begin{equation}\label{eqDefXkpDeltaPlus}
X^{k, p}_{\delta^+}(M, E) = \left\{u \in X^{k, p}_\delta(M, E) \,:\, \left\| \chi\left(\frac{\rho}{\mu}\right) u \right\|_{X^{k, p}_\delta(M, E)} \to 0\text{ as } \mu \to 0^+\right\}.
\end{equation}
This space is naturally endowed with the $X^{k, p}_\delta$-norm.

\begin{prop}[A second density result]\label{propDensity}
Let $(M, g)$ be a $C^{l, \beta}$-asymptotically hyperbolic manifold with $l + \beta \geq 2$ and let $E \to M$ be a geometric tensor bundle. Assume that  $k \in \bN$, $k \leq l$, $p \in (1; \infty)$, and $\delta \in \bR$. Then $X^{k, p}_{\delta^+}(M, E)$ is the closure of the set $C^{l, \beta}_c(M, E)$ of compactly supported sections of $E$ in $X^{k, p}_\delta(M, E)$.
\end{prop}

\begin{proof}
We first prove that $X^{k, p}_{\delta^+}(M, E)$ is closed as a subspace of $X^{k, p}_\delta(M, E)$. Let $\{u_i\}_i$ be an arbitrary sequence of elements of $X^{k, p}_{\delta^+}(M, E)$ converging to $u \in X^{k, p}_\delta(M, E)$. Choose an arbitrary $\epsilon > 0$. We remark that the functions $\chi\left(\frac{\rho}{\mu}\right)$ are uniformly bounded in $C^k(M, \bR)$. Thus there exists a constant $C > 0$ such that 

\begin{align*}
\left\| \chi\left(\frac{\rho}{\mu}\right) u \right\|_{X^{k, p}_\delta(M, E)}
	& \leq \left\| \chi\left(\frac{\rho}{\mu}\right) u_i \right\|_{X^{k, p}_\delta(M, E)} + \left\| \chi\left(\frac{\rho}{\mu}\right) (u - u_i) \right\|_{X^{k, p}_\delta(M, E)}\\
	& \leq \left\| \chi\left(\frac{\rho}{\mu}\right) u_i \right\|_{X^{k, p}_\delta(M, E)} + C \left\| u - u_i \right\|_{X^{k, p}_\delta(M, E)}.
\end{align*}
Let $i$ be such that $\left\| u - u_i \right\|_{X^{k, p}_\delta(M, E)} < \frac{\epsilon}{2 C}$ and let $\mu > 0$ be such that $\left\| \chi\left(\frac{\rho}{\mu}\right) u_i \right\|_{X^{k, p}_\delta(M, E)} < \frac{\epsilon}{2}$, then \[\left\| \chi\left(\frac{\rho}{\mu}\right) u \right\|_{X^{k, p}_\delta(M, E)} < \epsilon.\] Since $\epsilon$ was chosen arbitrarily, this proves that $u \in X^{k, p}_{\delta^+}(M, E)$. We conclude that $X^{k, p}_{\delta^+}(M, E)$ is a closed subspace of $X^{k, p}_\delta(M, E)$.\\

It is obvious that $C^{l, \beta}_c(M, E)$ is a subspace of $X^{k, p}_{\delta^+}(M, E)$. Next we show that for arbitrary $u \in X^{k, p}_{\delta^+}(M, E)$ and $\epsilon > 0$, there exists $\util \in C^{l, \beta}_c(M, E)$ such that $\left\|u - \util\right\|_{X^{k, p}_\delta(M, E)} < \epsilon$. Select $\mu > 0$ large enough so that $\left\| \chi\left(\frac{\rho}{\mu}\right) u \right\|_{X^{k, p}_\delta(M, E)} < \frac{\epsilon}{2}$. Then the section $\left(1-\chi\left(\frac{\rho}{\mu}\right)\right) u$ has its support in $K_{\mu/2}$ and is $\frac{\epsilon}{2}$-close to $u$ in $X_\delta^{k,p}(M,E)$-norm. Finally, remark that for tensors which have their supports contained in a (fixed) compact subset of $M$, the standard $W^{k, p}(M, E)$-norm and $X^{k, p}_\delta(M, E)$-norm are equivalent. Hence, from standard density results, there exists a section $\util \in C^{l, \beta}_c(M, E)$ such that \[\left\|\left(1-\chi\left(\frac{\rho}{\mu}\right)\right) u - \util\right\|_{X^{k, p}_\delta(M, E)} < \frac{\epsilon}{2}.\] By the triangle inequality we have $\left\|u - \util\right\|_{X^{k, p}_\delta(M, E)} < \epsilon$. Consequently, $C^{l, \beta}_c(M, E)$ is dense in $X^{k, p}_{\delta^+}(M, E)$.
\end{proof}

\begin{prop}[Weighted Sobolev embedding and Rellich Theorem]\label{propRellich} Let $(M, g)$ be a $C^{l, \beta}$-asymptotically hyperbolic manifold with $l + \beta \geq 2$ and $E \to M$ a geometric tensor bundle. 
\begin{itemize}
\item {\sc (Sobolev Spaces)} If $p, q \in (1; \infty)$, $k, j \in \bN$, $\alpha \in (0; 1)$ are such that $k + \alpha \leq l + \beta$, and if $\delta, \delta' \in \bR$, we have continuous inclusions
\[
\begin{aligned}
W^{k, p}_\delta(M, E) &\hookrightarrow C^{j, \alpha}_{\delta'} (M, E),&& \text{if $k - \frac{n}{p} \geq j + \alpha$ and $\delta' \leq \delta$,}\\
W^{k, p}_\delta(M, E) &\hookrightarrow W^{j, q}_{\delta'} (M, E),&& \text{if $k - \frac{n}{p} \geq j - \frac{n}{q}$ and $\delta' \leq \delta$.}
\end{aligned}
\]
\item {\sc (Local Sobolev Spaces)} If $p, q \in (1; \infty)$, $k, j \in \bN$, $\alpha \in (0; 1)$ are such that $k + \alpha \leq l + \beta$, and $\delta, \delta' \in \bR$, we have continuous inclusions
\[
\begin{aligned}
X^{k, p}_\delta(M, E) &\hookrightarrow C^{j, \alpha}_{\delta'} (M, E),&& \text{if $k - \frac{n}{p} \geq j + \alpha$ and $\delta' \leq \delta$,}\\
X^{k, p}_\delta(M, E) &\hookrightarrow X^{j, q}_{\delta'} (M, E),&& \text{if $k - \frac{n}{p} \geq j - \frac{n}{q}$ and $\delta' \leq \delta$.}
\end{aligned}
\]
\end{itemize}
Furthermore each of the embeddings is compact if both inequalities are strict.
\end{prop}

\begin{lemma}[$X^{k, p}_\delta$-spaces as Banach algebras]\label{lmBanachAlgebra}
Let $(M, g)$ be a $C^{l, \beta}$-asymptotically hyperbolic manifold with $l + \beta \geq 2$, and let $E_1 \to M$ and $E_2 \to M$ be two geometric tensor bundles. Let $k \in \bN$ be an integer, $p \in (1; \infty)$ and $\delta_1, \delta_2 \in \bR$. If $k p > n$, then the map
\[
\begin{array}{rccc}
\otimes : & X^{k, p}_{\delta_1}(M, E_1) \times X^{k, p}_{\delta_2}(M, E_2) & \to & X^{k, p}_{\delta_1+\delta_2}(M, E_1 \otimes E_2)\\
					& (u, v) & \mapsto &  u \otimes v
\end{array}
\]
is a continuous bilinear map. In particular, $X^{k, p}_0(M, \bR)$ is a unital Banach algebra and $X^{k, p}_\delta(M, \bR)$ are non-unital Banach algebras for any $\delta > 0$.
\end{lemma}

\begin{proof}
Let $p_0 \in M$ be arbitrary. It follows from \cite[Theorem 5.23]{Adams} that the tensor map $C^\infty(B_1, E_1) \times C^\infty(B_1, E_2) \to C^\infty(B_1, E_1 \otimes E_2)$ extends to a continuous bilinear map $\otimes : W^{k, p}(B_1, E_1) \times W^{k, p}(B_1, E_2) \to W^{k, p}(B_1, E_1 \otimes E_2)$ (here we have blurred the distinction between $E_1$, $E_2$ and their pull-back by the M\"obius chart). Hence there exists a constant $C > 0$ such that for any $p_0 \in M$
\begin{align*}
\left\| \left(\left(\Phi^1_{p_0}\right)^{-1}\right)^* (u \otimes v)\right\|_{W^{k, p}(B_1, E_1 \otimes E_2)}
	& \leq C \left\| \left(\left(\Phi^1_{p_0}\right)^{-1}\right)^* u\right\|_{W^{k, p}(B_1, E_1)} \left\| \left(\left(\Phi^1_{p_0}\right)^{-1}\right)^* v\right\|_{W^{k, p}(B_1, E_2)}\\
	& \leq C \rho^{\delta_1}(p_0) \left\|u\right\|_{X^{k, p}_{\delta_1}(M, E_1)} \rho^{\delta_2}(p_0) \left\|v\right\|_{X^{k, p}_{\delta_2}(M, E_2)}\\
	& \leq C \rho^{\delta_1+\delta_2}(p_0) \left\|u\right\|_{X^{k, p}_{\delta_1}(M, E_1)} \left\|v\right\|_{X^{k, p}_{\delta_2}(M, E_2)}.\\
\end{align*}
Taking the supremum over all $p_0 \in M$, we get that $u \otimes v$ belongs to $X^{k, p}_{\delta_1+\delta_2}(M, E_1 \otimes E_2)$ and that
\[\left\|u \otimes v\right\|_{X^{k, p}_{\delta_1+\delta_2}(M, E_1 \otimes E_2)} \leq C \left\|u\right\|_{X^{k, p}_{\delta_1}(M, E_1)} \left\|v\right\|_{X^{k, p}_{\delta_2}(M, E_2)}.\]
\end{proof}

We now state two weighted analogs of the Young inequality:

\begin{lemma}\label{lmYoung1}
Let $(M, g)$ be an asymptotically hyperbolic manifold and let $u, v : M \to \bR$ be two measurable functions. Let $p, q, r \in [1; \infty)$ be such that $\frac{1}{p}+\frac{1}{q} = \frac{1}{r}$ and assume that $u \in L^p_{\delta_1}(M, \bR)$, $v \in L^q_{\delta_2}(M, \bR)$ for $\delta_1, \delta_2 \in \bR$. Then $u v \in L^r_{\delta_1 + \delta_2}(M, \bR)$ and \[\left\| u v \right\|_{L^r_{\delta_1 + \delta_2}} \leq \left\| u \right\|_{L^p_{\delta_1}} \left\| v \right\|_{L^q_{\delta_2}}.\]
\end{lemma}

\begin{proof}
Indeed, let $\alpha = \frac{p}{r}$  and $\beta = \frac{q}{r}$ so that $\frac{1}{\alpha} + \frac{1}{\beta} = 1$. Then by the H\"older inequality we have
\begin{align*}
\left(\int_M \left|u v\right|^r \rho^{-(\delta_1 + \delta_2) r} d\mu_g\right)^{\frac{1}{r}}
	& \leq \left(\int_M \left|u\right|^{\alpha r} \rho^{-\alpha \delta_1 r} d\mu_g\right)^{\frac{1}{\alpha r}}
				 \left(\int_M \left|v\right|^{\beta r}  \rho^{-\beta  \delta_2 r} d\mu_g\right)^{\frac{1}{\beta r}}\\
	&  =   \left(\int_M \left|u\right|^p \rho^{-p \delta_1} d\mu_g\right)^{\frac{1}{p}} \left(\int_M \left|v\right|^q \rho^{-q \delta_2} d\mu_g\right)^{\frac{1}{q}}\\
	&  =   \left\|u\right\|_{L^p_{\delta_1}} \left\|v\right\|_{L^q_{\delta_2}}.
\end{align*}
\end{proof}

\begin{lemma}\label{lmYoung2}
Let $(M, g)$ be an asymptotically hyperbolic manifold and let $u, v : M \to \bR$ be two measurable functions. Let $p, q, r \in [1; \infty)$ be such that $\frac{1}{p}+\frac{1}{q} = \frac{1}{r}$ and assume that $u \in X^{0,p}_{\delta_1}(M, \bR)$, $v \in L^q_{\delta_2}(M, \bR)$ for  $\delta_1, \delta_2 \in \bR$. Then $u v \in L^r_{\delta'}(M, \bR)$ for any $\delta'$ such that $\delta' + \frac{n-1}{r} < \delta_1 + \delta_2 + \frac{n-1}{q}$, and \[\left\| u v \right\|_{L^r_{\delta'}} \leq C \left\| u \right\|_{X^{0, p}_{\delta_1}} \left\| v \right\|_{L^q_{\delta_2}}\] for some constant $C > 0$ independent of $u$ and $v$.
\end{lemma}

\begin{proof}
As in the proof of Lemma \ref{lmDensity}, we cover $M$ by a uniformly locally finite set of M\"obius balls $B_i$. Set $\rho_i = \rho(x_i)$ and remark that on each $B_i$ we have $c^{-1} \rho_i \leq \rho \leq c \rho_i$ for some constant $c > 0$ independent of $i$. From this, letting $C$ denote a positive constant that can vary from line to line, but which is independent of $u$ and $v$, we deduce that
\begin{align*}
\int_M \left|u v\right|^r \rho^{-r \delta'} d\mu_g
	& \leq C \sum_i \rho_i^{-r \delta'} \int_{B_i} \left|u v\right|^r d\mu_g\\
	& \leq C \sum_i \rho_i^{-r \delta'} \left(\int_{B_i} u^p d\mu_g\right)^{\frac{r}{p}} \left(\int_{B_i} v^q d\mu_g\right)^{\frac{r}{q}}\qquad\text{(by the H\"older inequality)}\\
	& \leq C \sum_i \rho_i^{-r \delta'} \rho_i^{r \delta_1} \left\|u\right\|_{X^{0, p}_{\delta_1}}^r \left(\int_{B_i} v^q d\mu_g\right)^{\frac{r}{q}}\\
	& \leq C \left\|u\right\|_{X^{0, p}_{\delta_1}}^r \sum_i \rho_i^{r (\delta_1 + \delta_2 - \delta')} \left(\rho_i^{-q \delta_2} \int_{B_i} v^q d\mu_g\right)^{\frac{r}{q}}\\
	& \leq C \left\|u\right\|_{X^{0, p}_{\delta_1}}^r \left(\sum_i \left(\rho_i^{-q \delta_2} \int_{B_i} v^q d\mu_g\right)^{\alpha \frac{r}{q}} \right)^{\frac{1}{\alpha}}
				\left( \sum_i \rho_i^{\beta r (\delta_1 + \delta_2 - \delta')}\right)^{\frac{1}{\beta}},
\end{align*}
where $\alpha > 1$ and $\beta$ is such that $\frac{1}{\alpha} + \frac{1}{\beta} = 1$. Setting $\alpha = \frac{q}{r}$, we get
\begin{align*}
\int_M \left|u v\right|^r \rho^{-r \delta'} d\mu_g
	& \leq C \left\|u\right\|_{X^{0, p}_{\delta_1}}^r \left(\sum_i \rho_i^{-q \delta_2} \int_{B_i} v^q d\mu_g \right)^{\frac{r}{q}}
				 \left( \sum_i \rho_i^{\beta r (\delta_1 + \delta_2 - \delta')}\right)^{\frac{1}{\beta}}\\
	& \leq C \left\|u\right\|_{X^{0, p}_{\delta_1}}^r \left\|v\right\|_{L^q_{\delta_2}}^r \left( \sum_i \rho_i^{\beta r (\delta_1 + \delta_2 - \delta')}\right)^{\frac{1}{\beta}}.
\end{align*}

To conclude the proof, it only remains to show that the last sum converges under the assumption of the lemma. Indeed,
\[\sum_i \rho_i^{\beta r (\delta_1+\delta_2-\delta')} \leq C \int_M \rho^{\beta r (\delta_1+\delta_2-\delta')} d\mu_g \leq C \int_{\Mbar} \rho^{\beta r (\delta_1+\delta_2-\delta') - n} d\mu_{\gbar}.\]
The last integral converges provided that $\beta r (\delta_1+\delta_2-\delta') - n > -1$. We recall that \[1 = \frac{1}{\alpha} + \frac{1}{\beta} = \frac{r}{q} + \frac{1}{\beta},\] hence $\beta r = p$. Thus the condition becomes \[\delta_1+\delta_2-\delta' > \frac{n-1}{p},\] which is equivalent to the assumption of the lemma, since $\frac{n-1}{p} = \frac{n-1}{r} - \frac{n-1}{q}$.
\end{proof}

\begin{lemma}[Embedding of $X^{k, p}_\delta$ into $W^{k,q}_{\delta'}$]\label{lmEmbeddingXIntoL}
Let $(M, g)$ be an asymptotically hyperbolic manifold and let $E\to M$ be a geometric tensor bundle. Let $p, q \in [1; \infty)$ be such that $q \leq p$, $\delta, \delta' \in \bR$ and assume that $u \in X^{k,p}_{\delta}(M, E)$ for some integer $k$ such that $0\leq k\leq l$. Then $u \in W^{k,q}_{\delta'}(M, E)$ for any $\delta'$ such that $\delta' + \frac{n-1}{q} < \delta$, and \[\left\| u \right\|_{W^{k,q}_{\delta'}} \leq C \left\|u\right\|_{X^{k, p}_{\delta}}\] for some constant $C > 0$ independent of $u$.
\end{lemma}

\begin{proof}
We will treat the case $q < p$, the equality case being simpler. Since the proof goes as for the previous lemma, we will use the same notation.

First let $k=0$. Remark that for each $i$ by the H\"older inequality we have
\[\int_{B_i} |u|_g^q d\mu_g \leq \left(\int_{B_i} 1^\alpha d\mu_g \right)^\frac{1}{\alpha} \left(\int_{B_i} |u|_g^{\beta q} d\mu_g \right)^{\frac{1}{\beta}} \leq C \left(\int_{B_i} |u|_g^p d\mu_g \right)^{\frac{q}{p}},\] since the volume of $B_i$ is uniformly bounded. Thus
\begin{align*}
\left\| u \right\|_{L^q_{\delta'}}
	& \leq \sum_i \rho_i^{-q\delta'} \int_{B_i} |u|_g^q d\mu_g\\
	& \leq C \sum_i \rho_i^{-q\delta'} \left(\int_{B_i} |u|_g^p d\mu_g \right)^{\frac{q}{p}}\\
	& \leq C \sum_i \rho_i^{-q\delta'} \rho^{q \delta} \left\|u\right\|^q_{X^{0, p}_\delta}\\
	& \leq C \left( \sum_i \rho_i^{q(\delta - \delta')} \right) \left\|u\right\|^q_{X^{0, p}_\delta}.
\end{align*}

Arguing as in the proof of the previous lemma, we conclude that the sum converges under the assumptions we have made.\\

Now assume that $k\geq 1$. In this case the statement follows from the fact that $\nabla^i u\in X_\delta^{0,p}(M,E \otimes (T^*M)^i)\subset L_{\delta'}^p(M,E\otimes (T^*M)^i)$ for all integer $i$ such that $0\leq i \leq k$.
\end{proof}

\subsection{Statement of the main result}

Having defined the setup of the article, we can now state our main result:

\begin{theorem}\label{thmMainTheorem}
Let $(M, g)$ be a $C^{l, \beta}$-asymptotically hyperbolic manifold with constant scalar curvature, where $l \in \bN$, $l \geq 2$ and $\beta \in [0; 1)$. Let $\sigma_0 \in X^{1, p}_\delta$, where $p \in (n; \infty)$, and $\delta \in (0; n)$, and assume that $\tau : M \to \bR$ is a positive function such that $\tau - 1\in X^{1, p}_\delta$. If the limit equation
\[\Delta_L \psi = \lambda \sqrt{\frac{n-1}{n}} \left| \lie \psi \right|_g \frac{d \tau}{\tau}\]
admits no non-zero solution $\psi \in W^{1, 2}$ for any $\lambda \in (0; 1]$, then the set of solutions $(\phi, \psi) \in X^{2, p}_+ \times X^{2, p}_\delta$ of the constraint equations \eqref{eqLichnerowicz}-\eqref{eqVector} is non-empty and compact. Moreover, any such function $\phi$ satisfies $\phi - 1 \in X^{2, p}_\delta$.\\

Furthermore,
\begin{itemize}
\item if $\sigma_0 \in X^{k-1, p}_\delta$ and $\tau - 1 \in X^{k-1, p}_\delta$ for some $1\leq k \leq l$, then any solution $(\phi, \psi) \in X^{2, p}_+ \times X^{2, p}_\delta$ belongs to $\left(1 + X^{k, p}_\delta\right) \times X^{k, p}_\delta$, and the set $\{(\phi-1, \psi)\} \subset X^{k, p}_\delta \times X^{k, p}_\delta$ is compact.
\item if $\sigma_0 \in C^{k-1, \alpha}_\delta$ and $\tau - 1 \in C^{k-1, \alpha}_\delta$ for some $k \geq 1$ and $\alpha \in (0; 1)$ such that $k + \alpha \leq l + \beta$, then any solution $(\phi, \psi) \in X^{2, p}_+ \times X^{2, p}_\delta$ belongs to $\left(1 + C^{k, \alpha}_\delta\right) \times C^{k, \alpha}_\delta$, and the set $\{(\phi-1, \psi)\} \subset C^{k, \alpha}_\delta \times C^{k, \alpha}_\delta$ is compact.
\end{itemize}
\end{theorem}

\begin{rks}~
\begin{enumerate}
 \item The assumption on scalar curvature is only needed to give the optimal decay at infinity of the solutions. It could be weakened to $\scal+n(n-1) \in C^{k-2, \alpha}_\delta$.
 \item Compactness of the set of solutions holds when assuming that the limit equation \eqref{eqLimit} only admits no non-zero solution for $\lambda = 1$. 
\end{enumerate}
\end{rks}

The idea of the proof is similar to the one in \cite{DahlGicquaudHumbert}. Let us discuss the strategy.\\

The main difficulty of dealing with the coupled system comes from the fact that there is a competition between the two leading terms in the Lichnerowicz equation \eqref{eqLichnerowicz}, namely, between $n(n-1)\tau^2 \phi^{\kappa+1}$ and $\left|\lie\psi\right|_g^2 \phi^{-\kappa-3}$. Intuitively, this can be illustrated as follows. Assume that we can define a quantity $\lambda$ corresponding to the magnitude of $\phi$ (later we will define it rigorously as $\lambda = \gamma_\beta(\phi, \psi)^{\frac{1}{2\kappa+4}}$ with $\gamma_\beta(\phi, \psi)$ being the energy to be defined in Section \ref{secBlowUp}). Then, from the vector equation \eqref{eqVector}, one deduces that $\psi$ (and hence $\lie \psi$) has order $\lambda^{\kappa+2}$. Consequently, $\left|\lie\psi\right|_g^2 \phi^{-\kappa-3} \simeq \lambda^{2(\kappa+2) - \kappa - 3}$ has order $\lambda^{\kappa+1}$, which is the same as that of $n(n-1)\tau^2 \phi^{\kappa+1}$, while other terms of the Lichnerowicz equation have lower order. This remark shows that it may be impossible to get an a priori estimate for the solutions of the conformally formulated constraint equations  \eqref{eqLichnerowicz}-\eqref{eqVector}.\\

To remedy this situation, we will favor the term $n(n-1)\tau^2 \phi^{\kappa+1}$ in the Lichnerowicz equation by slightly decreasing the exponent of $\phi$ in the vector equation, thus introducing the subcritical equations \eqref{eqSubCritConstraint}. For these equations, it is possible to obtain an a priori bound for the solutions, and then solve them by a fixed point method, see Section \ref{secSubCritSolutions}.\\

The next step consists in letting the regularization parameter $\epsilon$ tend to zero. Using a simple compactness argument, one can show that if the solutions of the subcritical equations are uniformly bounded then there is a solution of the constraint equations. In order to measure the magnitude of the solutions, we define the notion of energy, associated to a pair $(\phi, \psi)$ of solutions of the subcritical equations, by formula \eqref{eqDefEnergy}. Then, in Proposition \ref{propBoundPhi}, we show that the energy defined in such a way indeed measures the magnitude of $(\phi, \psi)$. As a consequence, if the energy of the solutions of the subcritical equations remains bounded as $\epsilon$ tends to zero, then there exists a solution of the constraint equations.\\

How can we ensure that the energy remains bounded? Assume by contradiction that the energy is unbounded. Then, returning to our heuristic analysis, we see that all the non-dominant terms of the Lichnerowicz equation are beaten by $n(n-1)\tau^2 \phi^{\kappa+1}$ and $\left|\lie\psi\right|_g^2 \phi^{-\kappa-3}$, which means that they disappear in the limit, leaving the direct relation \[n(n-1)\tau^2 \phi^{\kappa+1}  =\left|\lie\psi\right|_g^2 \phi^{-\kappa-3}\] between $\phi$ and $\psi$. Plugging this relation into the vector equation, we obtain a non-zero solution of the limit equation \eqref{eqLimit}: \[\Delta_L \psi = \lambda \sqrt{\frac{n-1}{n}} \left| \lie \psi \right|_g \frac{d \tau}{\tau},\] where the coefficient $\lambda \in (0; 1]$ appears because the subcritical equations do not behave perfectly well under the rescaling of the unknowns. This is the content of Proposition \ref{propUnboundedEnergy}.\\

Consequently, if the limit equation does not admit non-zero solutions for any $\lambda \in (0; 1]$, then the energy remains bounded as $\epsilon$ tends to zero. This completes the proof of Theorem \ref{thmMainTheorem}.

\section{Sub-critical equations}\label{secSubCritSolutions}
In this section, we study the following system:

\begin{equation}
\label{eqSubCritConstraint}
\left\lbrace
\begin{aligned}
-\frac{4(n-1)}{n-2} \Delta \phi + \scal~\phi + n(n-1) \tau^2 \phi^{\kappa+1} & = \left|\sigma\right|^2_g \phi^{-\kappa-3} \quad \text{(Lichnerowicz equation),}\\
\Delta_L \psi & = (n-1) \phi^{\kappa+2 - \epsilon} d \tau \quad \text{($\epsilon$-Vector equation),}
\end{aligned}
\right.
\end{equation}
where $0 < \epsilon < 1$ is a constant. The approach we follow is similar to the one introduced in \cite{HNT1, HNT2} and \cite{MaxwellNonCMC}. The main ingredient is the following variant of the Schauder fixed point theorem (see \cite[Corollary 11.2]{GilbargTrudinger}):

\begin{prop}[Modified Schauder fixed point theorem]\label{propSchauder}
Let $X$ be a Banach space, let $C \subset X$ be a closed convex subset, and let $F : C \to X$ be a continuous function such that $F(C) \subset C$, and $F(C)$ is precompact in $X$. Then $F$ admits (at least) one fixed point in $C$.
\end{prop}

This proposition is a simple consequence of the classical Schauder theorem applied to the closure of the convex hull of $F(C)$ (which is a subset of $C$ since $C$ is closed and convex). In our case, $F$ will be the composition of the map $\phi \mapsto \psi_\phi$, where $\psi_\phi$ is the unique solution $\psi$ of the $\epsilon$-vector equation, with the map $\psi \mapsto \sigma = \sigma_0 + \lie \psi \mapsto \Phi_\sigma$, where $\Phi_\sigma$ is the unique solution of the Lichnerowicz equation. These two equations will be studied in the next two subsections \ref{secVector} and \ref{secLichnerowicz}. Compactness of the image will be a consequence of the weighted Rellich Theorem (Proposition \ref{propRellich}), the required loss of decay at infinity being allowed by the assumption $d\tau \in X^{0, p}_\delta$ (see Theorem \ref{thmMainTheorem}). $X$ will be the space $L^\infty(M, \bR)$ of bounded functions over $M$, and $C$ will be the subset of functions that are between two well chosen barrier functions. The construction of solutions of the sub-critical equations together with their main properties are described in Subsection \ref{secSubCrit}. Finally, in Subsection \ref{secRemarkAlmostCMC}, we show how the techniques that we have developed in this section lead to an improvement of the results of \cite{IsenbergPark} in the ``near CMC'' case.

\subsection{The vector equation}\label{secVector}
We first formulate the main isomorphism theorem for the vector Laplacian on an asymptotically hyperbolic manifold. 

\begin{prop}[Isomorphism theorem for the vector Laplacian]\label{propIsomVectLaplacian} Let $(M, g)$ be a $C^{l, \beta}$-asymptotically hyperbolic manifold with $l + \beta \geq 2$. Then the vector Laplacian is an isomorphism between the following spaces:
\[
\left\lbrace
\begin{array}{rll}
W^{m+2, p}_{\delta}			 & \to W^{m, p}_{\delta} & \text{if $m+2\leq l+\beta$ and $\left|\delta+\frac{n-1}{p}-\frac{n-1}{2}\right| < \frac{n+1}{2}$};\\
X^{m+2, p}_{\delta}			 & \to X^{m, p}_{\delta} & \text{if $m+2\leq l+\beta$ and $\delta \in (-1, n)$};\\
C^{m+2, \alpha}_\delta & \to C^{m, \alpha}_\delta & \text{if $m+2+\alpha \leq l+\beta$ and $\delta \in (-1, n)$}.\\
\end{array}
\right.
\]
\end{prop}

The proof of this proposition is carried out in Appendix \ref{secNoL2CKV}. Its first corollary is the existence of TT-tensors:

\begin{cor}[Construction of TT-tensors]\label{corTTTensors}
Let $(M, g)$ be a $C^{l, \beta}$-asymptotically hyperbolic manifold with $l+\beta \geq 2$ and $\sigma_0$ a symmetric traceless 2-tensor. Then,
\begin{enumerate}
\item if $\sigma_0 \in W^{k-1, p}_\delta$ with $2 \leq k \leq l$, $1 < p < \infty$, and $\left|\delta + \frac{n-1}{p} - \frac{n-1}{2} \right| < \frac{n+1}{2}$, then there exists a unique 1-form $\psi \in W^{k, p}_\delta$ such that $\sigma = \sigma_0 + \lie \psi \in W^{k-1, p}_\delta$ is transverse.
\item if $\sigma_0 \in X^{k-1, p}_\delta$ with $2 \leq k \leq l$, $1 < p < \infty$, and $\delta \in (-1, n)$, then there exists a unique 1-form $\psi \in X^{k, p}_\delta$ such that $\sigma = \sigma_0 + \lie \psi \in X^{k-1, p}_\delta$ is transverse.
\item if $\sigma_0 \in \mathcal{C}^{k-1, \alpha}_\delta$ with $2 \leq k + \alpha \leq l$, $0 < \alpha < 1$, and $\delta \in (-1, n)$, then there exists a unique 1-form $\psi \in \mathcal{C}^{k, \alpha}_\delta$ such that $\sigma = \sigma_0 + \lie \psi \in X^{k-1, \alpha}_\delta$ is transverse.
\end{enumerate}
\end{cor}

The second corollary of this proposition is the continuous dependence of $\psi$ with respect to $\phi$ in the vector equation:

\begin{prop}[Solution of the vector equation]\label{propVectEq}
Let $(M, g)$ be a $C^{l, \beta}$-asymptotically hyperbolic manifold with $l+\beta \geq 2$ and $\tau : M \to \bR$ a function such that $\tau \to 1$ at infinity. Define $\psi_\phi$ as the solution of \[\Delta_L \psi = (n-1) \phi^{\kappa+2-\epsilon} d \tau.\] Then the map $\phi \mapsto \psi_\phi$ is well defined and locally Lipschitz continuous when seen as a map between the following spaces:
\[
\begin{aligned}
X^{k-2, \infty}_+ \to W^{k, p}_\delta, & \text{ if $\tau - 1 \in W^{k-1, p}_\delta$ with $2 \leq k \leq l$, $1 < p < \infty$, and $\left|\delta + \frac{n-1}{p} - \frac{n-1}{2} \right| < \frac{n+1}{2}$;}\\
X^{k-2, \infty}_+ \to X^{k, p}_\delta, & \text{ if $\tau - 1 \in X^{k-1, p}_\delta$ with $2 \leq k \leq l$, $1 < p < \infty$, and $\left|\delta - \frac{n-1}{2} \right| < \frac{n+1}{2}$;}\\
C^{k-2, \alpha}_+ \to C^{k, \alpha}_\delta, & \text{ if $\tau - 1 \in C^{k-1, \alpha}_\delta$ with $2 \leq k + \alpha \leq l$, $0 < \alpha < 1$, and $\left|\delta - \frac{n-1}{2} \right| < \frac{n+1}{2}$.}
\end{aligned}
\]
Here the index $+$ is understood as follows:
\[X^{k-2, \infty}_+ = \left\{\phi \in X^{k-2, \infty}_0 : \exists\, \epsilon > 0 \text{ such that } \phi \geq \epsilon \text{ a.e.} \right\}\] with a similar definition for the $C^{k-2, \alpha}_+$ spaces. Note that $X^{k-2, \infty}_+$ (resp. $C^{k-2, \alpha}_+$) is an open subset in $X^{k-2, \infty}_0$ (resp. $C^{k-2, \alpha}_0$).
\end{prop}

\subsection{The Lichnerowicz equation}\label{secLichnerowicz}
In this section we continue our construction of the mapping $F$ by studying the Lichnerowicz equation \[-\frac{4(n-1)}{n-2} \Delta \phi + \scal~\phi + n(n-1) \tau^2 \phi^{\kappa+1} - \left|\sigma\right|_g^2 \phi^{-\kappa-3} = 0,\] where $\sigma = \sigma_0 + \lie \psi \in L^\infty$. The construction is based on the monotony method.

\begin{prop}\label{propMonotony} Let $(M, g)$ be a $C^{l, \beta}$-asymptotically hyperbolic manifold with $l+\beta \geq 2$.
Assume that there exist two (strictly) positive functions $\phi_\pm \in X^{2, p}_0$, where $p > n$,  such that
\[
\left\lbrace
\begin{aligned}
-\frac{4(n-1)}{n-2} \Delta \phi_+ + \scal~\phi_+ + n(n-1) \tau^2 \phi_+^{\kappa+1} - \left|\sigma\right|_g^2 \phi_+^{-\kappa-3} & \geq 0,\\
-\frac{4(n-1)}{n-2} \Delta \phi_- + \scal~\phi_- + n(n-1) \tau^2 \phi_-^{\kappa+1} - \left|\sigma\right|_g^2 \phi_-^{-\kappa-3} & \leq 0,
\end{aligned}
\right.
\]
and such that $\phi_- \leq \phi_+$. Then there exists a solution $\phi \in X^{2, p}_0$ of the Lichnerowicz equation such that $\phi_- \leq \phi \leq \phi_+$.
\end{prop}

The proof of this proposition is standard, see e.g. \cite[Proposition 2.1]{Gicquaud} or \cite[Proposition 5.1]{Sakovich} for similar statements.\\

We will assume that $\tau > 0$ everywhere on $M$, $\tau \to 1$ at infinity and $\tau \in X^{1, p}_0$, where $p > n$. By \cite[Theorem 1.1]{Gicquaud}, there exists a unique function $\phi_- > 0$, $\phi_- \in X^{2, p}_0$ such that

\begin{equation}
\label{eqPrescScal}
\left\lbrace
\begin{aligned}
-\frac{4 (n-1)}{n-2} \Delta \phi_- + \scal~\phi_- + n(n-1) \tau^2 \phi_-^{\kappa+1} & = 0,\\
\phi_- & \to 1 \quad\text{at}~\partial M.\\
\end{aligned}
\right.
\end{equation}
Note that $\phi_-$ is the solution of the prescribed scalar curvature equation $\hscal = -n(n-1) \tau^2$, where $\hscal$ is the scalar curvature of $(M,\ghat)$ with $\ghat = \phi_-^\kappa g$. The existence of $\phi_-$ will be used in the proof of the following statement.

\begin{prop}[Solution of the Lichnerowicz equation]\label{propLichEq}
Let $(M, g)$ be a $C^{l, \beta}$-asymptotically hyperbolic manifold with $l+\beta \geq 2$ and let $\sigma$ be a symmetric traceless 2-tensor. Assume that $\sigma \in L^\infty$. Then there exists a unique solution $\phi$  of the Lichnerowicz equation \eqref{eqLichnerowicz} such that $\phi \in X^{2, p}_+$, where $p > n$ \footnote{The condition that $\phi$ is uniformly bounded from below ensures that $\phi$, seen as a conformal factor, does not collapse the boundary at infinity.}, and there exist constants $a>0$ and $b > 0$, independent of $\sigma$ and such that $\phi \leq a \left(\left\|\sigma\right\|_{L^\infty}\right)^{\frac{1}{\kappa+2}} + b$. Furthermore, the map $L^\infty \ni \sigma \mapsto \phi \in L^\infty$ is Lipschitz-continuous on any subset of the form $\{\left\|\sigma\right\|_{L^\infty} \leq C\}$.
\end{prop}

The proof of this proposition relies on the following a priori estimate:

\begin{lemma}[An a priori estimate for $X^{2, p}_0$ functions]\label{lmAPrioriX2p} Let $(M, g)$ be a $C^{l, \beta}$-asymptotically hyperbolic manifold with $l+\beta \geq 2$. Suppose that functions $f \in L^\infty$ and $u \in W^{2, p}_{loc} \bigcap L^\infty$, where $p > n$, satisfy \[ - \Delta u + \langle b, \nabla u\rangle_g + c u = f,\] where $b$ is a bounded 1-form such that $|b| \leq A$, and $c \in L^\infty$ is bounded from below by a constant $c_0 > 0$, i.e. $c \geq c_0 > 0$. Then \[\left\|u \right\|_{L^\infty} \leq \frac{1}{c_0} \left\|f\right\|_{L^\infty}.\]
\end{lemma}

\begin{proof}
Consider $\util = u - \frac{1}{c_0} \left\|f\right\|_{L^\infty}$. It is obvious that $\util$ satisfies \[- \Delta \util + \langle b, \nabla \util\rangle_g + c \util = f - \frac{c}{c_0} \left\|f\right\|_{L^\infty} \leq 0.\] \\

First, let us show that $\util \leq 0$. Assume by contradiction that $\lambda = \sup_M \util > 0$. Denote \[F_i = \lambda + \frac{1}{i} - \util,\] then $F_i$ satisfies \[-\Delta F_i + \langle b, \nabla F_i\rangle_g + c F_i \geq c \left(\lambda + \frac{1}{i}\right).\] As in the Cheng and Yau maximum principle (see e.g. \cite[Theorem 3.5]{GrahamLee}), select a sequence $p_i \in M$ such that $\util(p_i) \geq \lambda - \frac{1}{i}$, and for each $i$ select a M\"obius chart centered at each $p_i$: $\Phi^1_{p_i}(p) = \left(\theta^1_i, \ldots, \theta^n_i\right)$ centered at $p_i$. Let $g_i(p) = 1 - \left[\left(\theta^1_i(p)\right)^2 + \cdots + \left(\theta^n_i(p)\right)^2\right]$. Due to the fact that $(M, g)$ is $C^2$-asymptotically hyperbolic, both $\left|\nabla g_i \right|_g$ and $\left|\Delta g_i\right|$ are bounded independently of $i$ (see \cite[proof of Theorem 3.5]{GrahamLee} for details). Let $\frac{1}{i} \leq \epsilon_i \leq \frac{3}{i}$ be a sequence to be chosen later and define \[h_i = \max \left\{0, \epsilon_i - \frac{F_i}{g_i}\right\}.\] Since $F_i \geq \frac{1}{i}$, it is easy to see that $h_i$ has compact support. We extend $h_i$ by zero outside the M\"obius chart. Multiplying the differential inequality for $F_i$ by $h_i$ and integrating by parts we find that
\begin{align*}
\left(\lambda + \frac{1}{i}\right) \int_M c~h_i d\mu_g
	& \leq \int_M \left(\left\langle \nabla h_i, \nabla F_i\right\rangle_g + h_i \left\langle b, \nabla F_i\right\rangle_g + c~h_i F_i\right) d\mu_g\\
	&  =   \int_M \left(\left\langle \nabla h_i, \nabla \left(g_i\frac{F_i}{g_i}\right)\right\rangle_g + h_i \left\langle b, \nabla \left(g_i\frac{F_i}{g_i}\right)\right\rangle_g + c~h_i F_i\right) d\mu_g\\
	&  =   \int_M \left(- g_i \left|\nabla h_i\right|_g^2 + \frac{F_i}{g_i} \left\langle \nabla h_i, \nabla g_i \right\rangle_g - h_i g_i \left\langle b, \nabla h_i \right\rangle_g + h_i \frac{F_i}{g_i} \left\langle b, \nabla g_i\right\rangle_g + c~h_i F_i\right) d\mu_g\\
	&  =   \int_M \left(- g_i \left|\nabla h_i\right|_g^2 + (\epsilon_i - h_i) \left\langle \nabla h_i, \nabla g_i \right\rangle_g - h_i g_i \left\langle b, \nabla h_i \right\rangle_g\right.\\
	&  \left. + h_i \frac{F_i}{g_i} \left\langle b, \nabla g_i\right\rangle_g + c~h_i F_i\right) d\mu_g\\
	&  =   \int_M \left(- g_i\left|\nabla h_i\right|_g^2-\epsilon_i h_i \Delta g_i-\left\langle\nabla\frac{h_i^2}{2}, \nabla g_i\right\rangle_g - h_i g_i \left\langle b, \nabla h_i \right\rangle_g\right.\\
	&  \left. + h_i \frac{F_i}{g_i} \left\langle b, \nabla g_i\right\rangle_g + c~h_i F_i\right) d\mu_g\\
	&  =   \int_M \left[- g_i\left|\nabla h_i\right|_g^2 - h_i g_i \left\langle b, \nabla h_i \right\rangle_g + h_i \frac{F_i}{g_i} \left\langle b, \nabla g_i\right\rangle_g + \left(\frac{h_i^2}{2} - \epsilon_i h_i\right) \Delta g_i + c~h_i F_i\right]d\mu_g.
\end{align*}
We now estimate the last three terms in the following way. Since $0 \leq h_i \leq \epsilon_i$, there exists a constant $C > 0$ independent of $i$ such that \[\int_M \left(\frac{h_i^2}{2} - \epsilon_i h_i\right) \Delta g_i d\mu_g \leq C \epsilon_i \int_M h_i d\mu_g.\] On $\supp h_i$ we have $F_i \leq \epsilon_i g_i \leq \epsilon_i$, hence the last term also satisfies \[\int_M c~h_i F_i d\mu_g \leq C \epsilon_i \int_M h_i d\mu_g\] for some other constant $C > 0$ (in what follows, $C$ can vary from line to line but remains independent of $i$). The third term can be estimated similarly. As a consequence, we obtain the following inequality:
\[c_0 \left(\lambda + \frac{1}{i}\right) \int_M h_i d\mu_g \leq \int_M \left(- g_i\left|\nabla h_i\right|_g^2 - h_i g_i \left\langle b, \nabla h_i \right\rangle_g \right)+ C \epsilon_i \int_M h_i d\mu_g.\]

Now remark that, since $F_i \geq \frac{1}{i}$, $\frac{1}{i \epsilon_i} \leq g_i \leq 1$ on $\supp h_i \supset \supp |\nabla h_i|$. From this we deduce that 
\[c_0 \left(\lambda + \frac{1}{i}\right) \int_M h_i d\mu_g \leq - \frac{1}{i \epsilon_i} \int_M \left|\nabla h_i\right|_g^2 d\mu_g + A \int_M h_i |\nabla h_i| d\mu_g + C \epsilon_i \int_M h_i d\mu_g.\]

When $i$ is large enough, this implies that \[\frac{1}{3} \int_M \left|\nabla h_i\right|_g^2 d\mu_g \leq A \int_M h_i |\nabla h_i| d\mu_g.\] Select such an $i$. The proof now goes as in \cite[Theorem 8.1]{GilbargTrudinger}. It easy to see that $\bar{\epsilon}=\inf \frac{F_i}{g_i}$ satisfies $\frac{1}{i}\leq \bar{\epsilon}< \frac{3}{i}$. Note that we have left a large freedom in the choice of the sequence $\epsilon_i$ and that all we have done so far does not rely on this sequence. We now assume that $\epsilon_i> \bar{\epsilon}$ so that that $\supp h_i$ is a non empty set. We introduce  $\Gamma_i(\epsilon_i) = \supp |\nabla h_i| \subset \supp h_i$, and note that \[\left\|\nabla h_i\right\|_{L^2(M)} \leq 3 A \left\| h_i\right\|_{L^2(\Gamma_i(\epsilon_i))}\] by the Cauchy-Schwarz inequality. Hence
\begin{align*}
\left\|h_i \right\|_{L^{\frac{2n}{n-2}}(M)} & \leq C \left\| \nabla h_i \right\|_{L^2(M)} \qquad\text{(Sobolev embedding)}\\
  & \leq C \left\| h_i\right\|_{L^2(\Gamma_i(\epsilon_i))}\\
  & \leq C \left\|h_i \right\|_{L^{\frac{2n}{n-2}}(\Gamma_i(\epsilon_i))} \vol(\Gamma_i(\epsilon_i))^{\frac{1}{n}} \qquad\text{(H\"older inequality)}.\\
\end{align*}
Since $h_i \neq 0$ by assumption, this implies that $\vol(\Gamma_i(\epsilon_i))$ is bounded from below by a constant independent of $\epsilon_i$. We now let $\epsilon_i \rightarrow \bar{\epsilon}$, and deduce from the above inequality that $\frac{F_i}{g_i}$ attains its maximum on a set of positive measure. But $\nabla \frac{F_i}{g_i} = 0$ on this set, which contradicts the lower bound on the volume of $\Gamma_i(\epsilon_i)$. We conclude that $\util \leq 0$.\\

Finally, note that if we replace $u$ by $-v$ in the equation, the first part of the proof yields $v- \frac{1}{c_0} \left\|f\right\|_{L^\infty}\leq 0$. Hence $u \geq -\frac{1}{c_0} \left\|f\right\|_{L^\infty}$, which completes the proof.
\end{proof}

We now return to the proof of the proposition:\\

\begin{proof}[Proof of Proposition \ref{propLichEq}]~\\

\noindent $\bullet$ \emph{Existence:} Select $\phi_+ = \Lambda$, a constant. Then $\phi_+$ is a super-solution on $M$ of the Lichnerowicz equation if and only if
\[\scal~\Lambda + n(n-1) \tau^2 \Lambda^{\kappa+1} - \left|\sigma\right|_g^2 \Lambda^{-\kappa-3} \geq 0\quad\text{on $M$},\] that is to say, if \[n(n-1) \tau^2 \Lambda^{2 \kappa+ 4} + \scal~\Lambda^{\kappa + 4} \geq \left|\sigma\right|_g^2\quad\text{on $M$.}\] Let $s_0 = \min_M \scal$ and $\tau_0 = \min_M \tau$. $\tau_0 > 0$ by assumption (since $\tau \neq 0$ on $M$ and $\tau \to 1$ at $\partial M$). The preceding inequality will be true if \[n(n-1) \tau_0^2 \Lambda^{2 \kappa + 4} + s_0 \Lambda^{\kappa + 4} \geq \left\|\sigma\right\|_{L^\infty}^2\quad\text{on $M$.}\] If $\Lambda \geq \left(\frac{-2 s_0}{n(n-1) \tau_0^2}\right)^{\frac{1}{\kappa}} = b$ then $n(n-1) \tau_0^2 \Lambda^{2 \kappa+4} + s_0 \Lambda^{\kappa + 4} \geq n(n-1) \frac{\tau_0^2}{2} \Lambda^{2 \kappa + 4}$. We conclude that $\phi_+ = \Lambda$ is a super-solution if \[\Lambda \geq \max\left\{b, \left( \frac{2 \left\|\sigma \right\|_{L^\infty}^2}{n(n-1) \tau_0^2}\right)^{\frac{1}{2\kappa+4}}\right\}.\] Then, by the monotony method (Proposition \ref{propMonotony}), there exists a solution $\phi$ such that $\phi \geq \phi_-$, where $\phi_-$ is the solution of \eqref{eqPrescScal}, and $\left\|\phi\right\|_{L^\infty} \leq \Lambda \leq a \left( \left\| \sigma \right\|_{L^\infty} \right)^{\frac{1}{\kappa+2}} + b$, where $a = \left(\frac{2}{n(n-1) \tau_0^2}\right)^{\frac{1}{2\kappa+4}}$.\\

\noindent $\bullet$ \emph{Uniqueness:} Set $u=\log \phi$. Then the Lichnerowicz equation \eqref{eqLichnerowicz} becomes 
\[- \frac{4(n-1)}{n-2} \left(\Delta u + \left| d u \right|^2_g\right) + \scal + n(n-1) \tau^2 e^{\kappa u} - \left|\sigma\right|_g^2 e^{-(4+\kappa) u} = 0.\]

Let $\phi_1$ and $\phi_2$ be positive solutions of \eqref{eqLichnerowicz}. Subtracting the equations for $u_1=\log \phi_1$ and $u_2=\log \phi_2$, we obtain

\begin{align*}
0 & = - \frac{4(n-1)}{n-2} \left(\Delta \left(u_1 - u_2\right) + \left| d u_1 \right|^2_g - \left| d u_2 \right|^2_g\right) + n(n-1) \tau^2 \left(e^{\kappa u_1} - e^{\kappa u_2}\right)\\
	& - \left|\sigma\right|_g^2 \left(e^{-(4+\kappa) u_1} - e^{-(4+\kappa) u_2}\right).
\end{align*}

Define $u_\lambda = \lambda u_2 + (1 - \lambda) u_1$. Then the previous equation becomes

\begin{align*}
0 & = - \frac{4(n-1)}{n-2} \left(\Delta \left(u_1 - u_2\right) + \left\langle d (u_1 + u_2), d (u_1 - u_2) \right\rangle_g\right)\\
	&   + \left(n(n-1)\kappa\tau^2\int_0^1 e^{\kappa u_\lambda} d\lambda+(\kappa+4)\left|\sigma\right|_g^2 \int_0^1 e^{-(\kappa+4) u_\lambda} d\lambda \right) (u_1 -u_2).
\end{align*}

Since $\phi_1$ and $\phi_2$ are uniformly bounded from below, i.e. $\phi_1, \phi_2 \geq \epsilon > 0$, and since $\tau$ is also bounded from below, there exists a certain constant $c_0 > 0$ such that \[n(n-1)\kappa\tau^2\int_0^1 e^{\kappa u_\lambda} d\lambda+(\kappa+4)\left|\sigma\right|_g^2 \int_0^1 e^{-(\kappa+4) u_\lambda} d \lambda \geq c_0.\] We are now in a position to apply Lemma \ref{lmAPrioriX2p} with $f = 0$ to conclude that $u_1 = u_2$. This proves uniqueness.\\

\noindent $\bullet$ \emph{Continuity:} The continuous dependency of $\phi \in L^\infty$ with respect to $\sigma \in L^\infty$ follows the same line. Let $\sigma_1, \sigma_2 \in L^\infty$ and denote by $\phi_1$ and $\phi_2$ the corresponding solutions of the Lichnerowicz equation \eqref{eqLichnerowicz} and by $u_1$ and $u_2$ their logarithms. We obtain the following equation:

\begin{align*}
0 & = - \frac{4(n-1)}{n-2} \left(\Delta \left(u_1 - u_2\right) + \left| d u_1 \right|^2_g - \left| d u_2 \right|^2_g \right)\\
	& + n(n-1) \tau^2 \left(e^{\kappa u_1} - e^{\kappa u_2}\right) - \left|\sigma_1\right|_g^2 e^{-(4+\kappa) u_1} + \left|\sigma_2\right|_g^2 e^{-(4+\kappa) u_2}\\
	& = - \frac{4(n-1)}{n-2} \left(\Delta \left(u_1 - u_2\right) + \left\langle d (u_1 + u_2), d (u_1 - u_2) \right\rangle_g\right)\\
	& + n(n-1) \tau^2 \left(e^{\kappa u_1} - e^{\kappa u_2}\right) - \left|\sigma_1\right|_g^2 \left(e^{-(4+\kappa) u_1} - e^{-(4+\kappa) u_2}\right)
		- \left(\left|\sigma_1\right|_g^2 - \left|\sigma_2\right|_g^2\right) e^{-(4+\kappa) u_2}\\
	& = - \frac{4(n-1)}{n-2} \left(\Delta \left(u_1 - u_2\right) + \left\langle d (u_1 + u_2), d (u_1 - u_2) \right\rangle_g\right)\\
	& + \left(n(n-1)\kappa \tau^2 \int_0^1 e^{\kappa u_\lambda} d\lambda + (\kappa+4) \left|\sigma_1\right|_g^2 \int_0^1 e^{-(\kappa+4) u_\lambda} d\lambda \right) (u_1-u_2)\\
	& + \left(\left|\sigma_1\right|_g^2 - \left|\sigma_2\right|_g^2\right) e^{-(4+\kappa) u_2}.
\end{align*}
From the existence part of the proof, we infer that $\phi_1, \phi_2 \geq \phi_-$. As a consequence, we have \[n(n-1)\kappa \tau^2 \int_0^1 e^{\kappa u_\lambda} d\lambda + (\kappa+4) \left|\sigma_1\right|_g^2 \int_0^1 e^{-(\kappa+4) u_\lambda} d\lambda \geq n(n-1)\kappa \tau^2 \phi_-^\kappa \geq c_0 > 0.\] Thus Lemma \ref{lmAPrioriX2p} implies that
\[\left\|u_1 - u_2\right\|_{L^\infty}
\leq \frac{1}{c_0} \left\|\left(\left|\sigma_1\right|_g^2 - \left|\sigma_2\right|_g^2\right) e^{-(4+\kappa) u_2}\right\|_{L^\infty}
\leq \frac{1}{c_0} \left\|\left(\left|\sigma_1\right|_g^2 - \left|\sigma_2\right|_g^2\right) \phi_-^{-(4+\kappa)}\right\|_{L^\infty}.\]
$\phi_-$ being bounded from below, we obtain that there exists a constant $A > 0$ such that \[\left\|u_1 - u_2\right\|_{L^\infty} \leq A \left\|\sigma_1 + \sigma_2\right\|_{L^\infty} \left\|\sigma_1 - \sigma_2\right\|_{L^\infty}.\] This proves the proposition.
\end{proof}

\begin{rk}
 In the proof of Proposition \ref{propLichEq}, we used the fact that $\tau$ is bounded from below by some positive constant. In the general case, when $\tau$ is allowed to have zeros or to change sign, it can be shown that the existence of a solution to the Lichnerowicz equation is equivalent to the existence of a solution to the corresponding prescribed scalar curvature equation \eqref{eqPrescScal}. We give a proof which slightly differs from \cite[Proposition 4.2]{MaxwellRoughCompact}:

 If $\phi$ is a bounded positive solution to the Lichnerowicz equation \eqref{eqLichnerowicz} then $\phi$ is a natural super-solution for the prescribed scalar curvature equation \eqref{eqPrescScal}. A sub-solution to Equation \eqref{eqPrescScal} can be constructed as in \cite[Proof of Theorem 3.1]{Gicquaud}. By Proposition \ref{propMonotony}, this yields a bounded positive solution to Equation \eqref{eqPrescScal}.

 To carry out the proof in the reverse direction, assume that Equation \eqref{eqPrescScal} admits a positive solution $\phi_-$. Then $\phi_-$ is a sub-solution of the Lichnerowicz equation. To construct a super-solution, we introduce the metric $\gtil = \phi_-^\kappa g$. The Lichnerowicz equation becomes
 \[-\frac{4(n-1)}{n-2} \widetilde{\Delta} \phitil + n(n-1) \tau^2 \left(\phitil^{\kappa+1}-\phitil\right) = \left|\sigmatil\right|_{\gtil}^2 \phitil^{-\kappa-3},\] where $\widetilde{\Delta}$ is the Laplacian associated to the metric $\gtil$, $\sigmatil = \phi_-^{-2} \sigma$ and $\phitil = \phi_-^{-1} \phi$. Let $u$ be the unique bounded solution of \[-\frac{4(n-1)}{n-2} \widetilde{\Delta} u + n(n-1) \kappa \tau^2 u = \left|\sigmatil\right|_{\gtil}^2,\] then it is easy to check that $\phi_+ = \phi_- (1+u)$ is a super-solution of the Lichnerowicz equation. Indeed, we remark that from the maximum principle $u \geq 0$ and that, by convexity, $(1+u)^{\kappa+1} - (1+u) \geq \kappa u$. Finally, from Proposition \ref{propMonotony} it follows that the Lichnerowicz equation admits a solution. 

Note that in the argument above we have been quite sloppy regarding the existence and uniqueness of $u$. The problem is that $\gtil$ might not be asymptotically hyperbolic in the sense of Definition \ref{defAH}, since the low regularity assumptions of Proposition \ref{propLichEq} do not guarantee enough smoothness of the conformal factor $\phi_-$. One possible way to get around this difficulty it is to rewrite the equation that $u$ satisfies in terms of the metric $g$ and to notice that $\phi_- \to 1$, $|d\phi_-|_g\to 0$ at infinity (see \cite[Proof of Theorem 3.1]{Gicquaud}). We leave the details as an exercise to the reader and refer to Section \ref{secHigherRegularity} for the method to prove isomorphism theorems for non-geometric operators.
\end{rk}

\subsection{Solutions of the sub-critical equations}\label{secSubCrit}
In this section we construct solutions of the sub-critical constraint equations \eqref{eqSubCritConstraint}. Recall that the proof is based on Proposition \ref{propSchauder}. Select $\tau : M \to \bR$ such that $\tau - 1 \in X^{1, p}_\delta$ for some $\delta \in \bR$ satisfying $0 < \delta < n$, and $p > n$, and such that $\tau > 0$ on $M$. Let also $\sigma_0 \in X^{1, p}_\delta$ be a TT-tensor (see Corollary \ref{corTTTensors}). Define $C_\Lambda \subset L^\infty$ as $C_\Lambda = \{ \phi \in L^\infty | \phi_- \leq \phi \leq \Lambda\}$, where $\phi_-$ is the solution of the prescribed scalar curvature equation \eqref{eqPrescScal} described in the previous subsection, and $\Lambda > 1$ is a constant to be chosen later. We define a map $F_\epsilon : C_\Lambda \to L^\infty$ as follows. First, select $\phi \in C_\Lambda$. By Proposition \ref{propVectEq}, there exists a unique $\psi \in X^{2, p}_\delta(M, T^*M)$ which solves the $\epsilon$-vector equation. To $\psi$ we can associate $\sigma = \sigma_0 + \lie \psi \in X^{1, p}_\delta(M, S_2^0(M))$, which is then used to construct a solution $\phitil \in X^{2, p}_0(M, \bR) \subset L^\infty(M, \bR)$ of the Lichnerowicz equation \eqref{eqSubCritConstraint}. Finally, we define $F_\epsilon(\phi) = \phitil$. Our construction can be summarized by the following diagram:

\[
\begin{array}{cccccccccc}
F_\epsilon: & C_\Lambda & \to & X^{2, p}_\delta(M, T^*M) & \to & X^{1, p}_\delta(M, S_2^0(M)) & \to & L^\infty(M, S_2^0(M)) & \to & L^\infty(M).\\
	    & \phi & \mapsto	& \psi & \mapsto	& \sigma & \mapsto	& \sigma & \mapsto	& \phitil
\end{array}
\]

The mapping $\sigma \in X^{1, p}_\delta \to \sigma \in L^\infty$ is compact by the Rellich theorem (Proposition \ref{propRellich}). From this it is easy to notice that the image $F_\epsilon \left(C_\Lambda\right)$ is precompact in $L^\infty$. Propositions \ref{propVectEq} and \ref{propLichEq} also imply that $F_\epsilon$ is Lipschitz, and hence continuous. It only remains to check that if $\Lambda > 1$ is well chosen then $F_\epsilon \left(C_\Lambda\right) \subset C_\Lambda$. Indeed, if $\phi \in C_\Lambda$, then \[\left\|\psi\right\|_{X^{2, p}_\delta} \leq C \left\|\phi^{\kappa+2-\epsilon} d\tau \right\|_{X^{0, p}_\delta} \leq C \left\| \phi \right\|^{\kappa+2-\epsilon}_{L^\infty}  \left\| d\tau \right\|_{X^{0, p}_\delta} \leq C \Lambda^{\kappa+2-\epsilon} \left\| d\tau \right\|_{X^{0, p}_\delta}\] (here the constant $C$ can vary from line to line but is independent of $\phi$ and $\epsilon$). Therefore
\begin{align*}
\left\| \sigma \right\|_{L^\infty}
	& \leq \left\| \sigma_0 \right\|_{L^\infty} + \left\| \lie \psi \right\|_{L^\infty}\\
	& \leq \left\| \sigma_0 \right\|_{L^\infty} + C \left\| \lie \psi \right\|_{X^{1, p}_\delta} \qquad\text{(Sobolev injection)}\\
	& \leq \left\| \sigma_0 \right\|_{L^\infty} + C \left\| \psi \right\|_{X^{2, p}_\delta}\\
	& \leq \left\| \sigma_0 \right\|_{L^\infty} + C \Lambda^{\kappa+2-\epsilon} \left\| d\tau \right\|_{X^{0, p}_\delta}.
\end{align*}
From Proposition \ref{propLichEq}, we deduce that
\begin{align*}
\sup_M F_\epsilon (\phi)
	&  =   \sup_M \phitil\\
	& \leq a \left(\left\| \sigma \right\|_{L^\infty}\right)^{\frac{1}{\kappa+2}} + b \qquad\text{(Proposition \ref{propLichEq})}\\
	& \leq a \left(\left\| \sigma_0 \right\|_{L^\infty} + C \Lambda^{\kappa+2-\epsilon} \left\| d\tau \right\|_{X^{0, p}_\delta} \right)^{\frac{1}{\kappa+2}} + b.
\end{align*}

As a consequence, $\phitil \leq \Lambda$ provided that $a \left( \left\| \sigma_0 \right\|_{L^\infty} + C \Lambda^{\kappa+2-\epsilon} \left\| d\tau \right\|_{X^{0, p}_\delta} \right)^{\frac{1}{\kappa+2}} + b \leq \Lambda$ which is always true if $\Lambda > 1$ is large enough. Since, by construction, we have $\phitil \geq \phi_-$, we conclude that $C_\Lambda$ is stable. Applying Proposition \ref{propSchauder}, we obtain the following result:

\begin{prop}[Solutions of the sub-critical equations]\label{propSubCrit}
Let $(M, g)$ be a $C^{l, \beta}$-asymptotically hyperbolic manifold with $l+\beta \geq 2$. Suppose that $\tau : M \to \bR$ is a positive function such that $\tau \in 1 + X^{1, p}_\delta$, where $p > n$ and $0<\delta<n$, and let $\sigma_0 \in X^{1, p}_\delta$ be a TT-tensor. Then for any $0 < \epsilon < 1$, there exists a solution $(\phi, \psi)$, such that $\phi\in X_0^{2,p}$ and $\psi\in X_\delta^{2,p}$, of the $\epsilon$-sub-critical constraint equations \eqref{eqSubCritConstraint}.
\end{prop}

Let us recall the notations $M_\mu=\rho^{-1}(0;\mu)$ and $K_\mu=M\setminus M_\mu$. We will prove the following proposition regarding the behavior at infinity of solutions of the constraint equations:

\begin{prop}[Global behavior of the solutions]\label{propSubCritBehav} Assume further that $(M, g)$ has constant scalar curvature, $\scal = -n(n-1)$, that $\tau-1 \in X^{1, p}_\delta$, and that $\sigma_0 \in X^{1, p}_\delta$ for a certain $\delta$ such that $0 < \delta < n$ and $p > n$.
\begin{itemize}
\item There exists a constant $\eta \in (0; 1)$ such that any solution $(\phi, \psi)$ of the $\epsilon$-sub-critical constraint equations \eqref{eqSubCritConstraint}, with $\phi$ uniformly bounded from below, satisfies  $\phi_- \leq \phi \leq \eta^{-\frac{1}{\epsilon}}$.\\
\item There exists $\mu > 0$ such that any pair $(\phi, \psi)$ which is a solution of Equation \eqref{eqSubCritConstraint} with $\phi$ uniformly bounded from below satisfies $\phi \leq 1 + k \left(\frac{\rho}{\mu}\right)^{\delta'}$ on $M_\mu$, where $k = \max \left\{ \left\| \phi \right\|_{L^\infty} - 1, k_0 \right\}$, $k_0 > 0$ is some fixed constant (independent of $\phi$, $\psi$ and $\epsilon$), and $\delta'$ is such that $0 < \delta' \leq \frac{\delta}{\kappa+2}$.
\end{itemize}
\end{prop}

\begin{proof}Before giving the proof of this proposition, we prove the following inequality which will be needed at some point:

\begin{equation}
\label{eqInegConvex}
\forall v \geq 0, t \in [0; 1],\quad\left(1+v\right)^{\kappa+1} \geq 1 + (\kappa+1) t v + (1 - t) v^{\kappa+1}.
\end{equation}

\begin{itemize}
\item if $\kappa \geq 1$, one has \[\kappa (\kappa+1) \left(1+v\right)^{\kappa-1} \geq \kappa (\kappa+1) v^{\kappa-1}.\] Integrating this inequality with respect to $v$ twice, we see that \[\left(1+v\right)^{\kappa+1} \geq 1 + (\kappa+1) v + v^{\kappa+1}.\] Inequality \eqref{eqInegConvex} follows.\\
\item if $\kappa < 1$, then, due to the fact that the function $x \mapsto x^\kappa$ is concave, it is easy to show that
\[(\kappa+1) \left(1+v\right)^{\kappa} \geq (\kappa+1) \left(t + (1-t) v\right)^{\kappa} \geq (\kappa+1) \left[ t + (1-t) v^\kappa\right].\] Integrating this inequality with respect to $v$, one obtains \eqref{eqInegConvex}.
\end{itemize}

As the proof of Proposition \ref{propLichEq} shows, $\phi \geq \phi_-$, which establishes the first part of the first inequality. We now prove the second part. Set $v_+ = k \left(\frac{\rho}{\mu}\right)^{\delta'}$ and $\phi_+ = 1 + v_+$. We shall now find a condition which ensures that $\phi_+$ is a super-solution of the Lichnerowicz equation \eqref{eqLichnerowicz} with $\sigma = \sigma_0 + \lie \psi$ on $M\setminus K_\mu$. That is, $\phi_+$ should satisfy \[-\frac{4(n-1)}{n-2} \Delta \phi_+ - n(n-1)~\phi_+ + n(n-1) \tau^2 \phi_+^{\kappa+1} - \left|\sigma\right|_g^2 \phi_+^{-\kappa-3} \geq 0 \qquad\text{on $M_\mu$},\] where $\sigma = \sigma_0 + \lie \psi$.\\

By the inequality \eqref{eqInegConvex}, $\phi_+$ is a super-solution provided that \[-\frac{4(n-1)}{n-2} \Delta \phi_+ - n(n-1)~\phi_+ + n(n-1) \tau^2 \left(1 + (\kappa+1) t v_+ + (1 - t) v_+^{\kappa+1} \right) - \left|\sigma\right|_g^2 \phi_+^{-\kappa-3} \geq 0,\] for some $t \in [0; 1]$ to be determined. Decompose this inequality as follows:
\begin{align*}
0 & \leq \underbrace{-\frac{4(n-1)}{n-2} \Delta \phi_+ - n(n-1)~\phi_+ + n(n-1) \tau^2 \left(1 + (\kappa+1) t v_+\right)}_{(1)}\\
	&		 + \underbrace{n(n-1) \tau^2 (1 - t) v_+^{\kappa+1} - \left|\sigma\right|_g^2 \phi_+^{-\kappa-3}}_{(2)}.
\end{align*}

We prove that each term is positive provided that $\mu$ is small enough. A simple calculation shows that
\begin{align*}
\Delta \rho^{\delta'}
	& = \rho^2 \left( \overline{\Delta} \rho^{\delta'} - (n-2) \left\langle \frac{\overline{\nabla} \rho}{\rho}, \overline{\nabla} \rho^{\delta'} \right\rangle_{\gbar} \right)\\
	& = \delta'(\delta'-n+1) \rho^{\delta'} |\overline{\nabla}\rho|^2_{\gbar} + \delta' \rho^{\delta'+1} \overline{\Delta} \rho\\
	& = \delta'(\delta'-n+1) \rho^{\delta'} + \rho^{\delta'} O (\rho).\\
\end{align*}
We compute:
\begin{align*}
(1)	& =    -\frac{4(n-1)}{n-2} \Delta \phi_+ - n(n-1)~\phi_+ + n(n-1) \tau^2 \left(1 + (\kappa+1) t v_+\right)\\
		& =		 -\frac{4(n-1)}{n-2} \Delta v_+ - n(n-1)~\left(1 + v_+\right) + n(n-1) \tau^2 \left(1 + (\kappa+1) t v_+\right)\\
		& =		 -\frac{4(n-1)}{n-2} \Delta v_+ + n(n-1)~\left(\tau^2 (\kappa+1) t - 1\right) v_+ + n(n-1) \left(\tau^2 - 1\right)\\
		& =		 \left[-\frac{4(n-1)}{n-2} \delta' (\delta' - n + 1) + n(n-1)~\left(\tau^2(\kappa+1) t - 1\right) + o(1)\right] v_+ + n(n-1) \left(\tau^2 - 1\right)\\
		& =		 \left[\frac{4(n-1)}{n-2} (\delta' + 1) (n - \delta') + n(n-1) (\kappa+1) (\tau^2 t-1) + o(1)\right] v_+ + n(n-1) \left(\tau^2 - 1\right).\\
\end{align*}
Remark that the $o(1)$ term does not depend on $v_+$. If $\mu > 0$ is small enough and $t$ is close enough to $1$, then the bracketed term can be assumed to be greater than a given small positive constant on $M_\mu = \rho^{-1}(0; \mu)$. As a consequence, if $k$ is greater than a certain $k_0$, we can assume that $(1)$ is non-negative on $M \setminus K_\mu$.\\

We now turn our attention to the second term $(2)$. Using the decay properties of $\sigma_0$ and $d \tau$, it is easy to argue that there exist constants $\alpha, \beta > 0$ independent of the pair $(\phi, \psi)$ such that \[\left|\sigma\right|_g^2 \leq \left(\alpha + \beta~k^{2 \kappa + 4 - 2 \epsilon}\right) \rho^{2\delta} \leq \left(\alpha + \beta~k^{2 \kappa + 4}\right) \rho^{2\delta},\] if we assume that $k\geq 1$. Obviously, $(2)$ will be positive provided that \[n(n-1) \tau^2 (1-t)v_+^{2 \kappa + 4} \geq \left(\alpha + \beta~k^{2 \kappa + 4}\right) \rho^{2\delta}.\] To show that this inequality holds, it is enough to note that since $\delta \geq (\kappa + 2) \delta'$, we have \[n(n-1) \tau^2 (1-t) \frac{\rho^{(2\kappa+4)\delta'}}{\mu^{(2\kappa+4)\delta'}} \geq \left(\alpha + \beta~k^{2 \kappa + 4}\right) \rho^{2\delta}\qquad \text{on $M_\mu,$}\] provided that $\mu > 0$ is small enough.\\

We have proved that if $\mu > 0$ is small enough (but independent of $(\phi, \psi)$ !), then $\phi_+$ is a super-solution of the Lichnerowicz equation on $M_\mu$. Since $\phi_+ \geq \phi$ on $\partial(M\setminus K_\mu) = \rho^{-1}(\mu)$, the proof of Proposition \ref{propLichEq} shows that $\phi \leq \phi_+$ on $M \setminus K_\mu$.\\

We finally prove the first part of the proposition, namely, the existence of a uniform constant $\lambda$, such that $\phi_-\leq \phi\leq \lambda^{\frac{1}{\epsilon}}$. Recall that there exist constants $\alpha, \beta > 0$ such that $|\sigma|_g^2 \leq \alpha + \beta \left\|\phi\right\|_{L^\infty}^{2\kappa+4-2\epsilon}$. $\Lambda > 0$ will then be a super-solution for the Lichnerowicz equation if \[-n(n-1) \Lambda + n(n-1) \tau_0^2 \Lambda^{\kappa+1} \geq \left(\alpha + \beta \left\|\phi\right\|_{L^\infty}^{2\kappa+4-2\epsilon}\right) \Lambda^{-\kappa-3},\] where $\tau_0 = \inf_M \tau > 0$, that is to say if \[-n(n-1) \Lambda^{\kappa+4} + n(n-1) \tau_0^2 \Lambda^{2\kappa+4} \geq \alpha + \beta \left\|\phi\right\|_{L^\infty}^{2\kappa+4-2\epsilon}.\] If $\Lambda$ is large enough (say $\Lambda \geq \Lambda_0 > 1$), this inequality is true provided \[\frac{n(n-1)}{2} \tau_0^2 \Lambda^{2\kappa+4} \geq \beta \left\|\phi\right\|_{L^\infty}^{2\kappa+4-2\epsilon},\] that is, if $\Lambda \geq C \left\| \phi \right\|_{L^\infty}^{1 - \frac{\epsilon}{\kappa+2}}$ for some constant $C > 0$ independent of $\epsilon$, $\psi$, and $\phi$. Select now $\Lambda = \max \left\{ \Lambda_0, C \left\| \phi \right\|_{L^\infty}^{1 - \frac{\epsilon}{\kappa+2}}\right\}$, which is obviously a super-solution of the Lichnerowicz equation. However, since $\phi \to 1$ at infinity, there exists a compact subset $K'$ such that on $M \setminus K'$, $\phi \leq \Lambda_0 \leq \Lambda$. By a technique similar to the proof of Proposition \ref{propLichEq} (applied on a compact subset), we conclude that $\phi \leq \Lambda$ on $M$. In particular, if $\left\|\phi\right\|_{L^\infty}$ is larger than $\Lambda_0$, then $\left\|\phi\right\|_{L^\infty} \leq C \left\| \phi \right\|_{L^\infty}^{1 - \frac{\epsilon}{\kappa+2}}$. This proves the proposition.
\end{proof}

\begin{rks}
\begin{enumerate}
\item This proposition essentially asserts that when $\epsilon \to 0$, the blow-up phenomenon occurs for all solutions at a certain rate and in a compact subset instead of shifting towards the infinity.
\item The assumption that $(M, g)$ has constant scalar curvature is imposed to remove any spurious dependence of the behavior at infinity of the functions $\phi$ with respect to the scalar curvature. This assumption causes no loss of generality, see \cite{AnderssonChruscielFriedrich}.
\end{enumerate}
\end{rks}

In the next proposition, we give a more precise description of the decay for the solutions of the Lichnerowicz equation.

\begin{prop}[Individual behavior of the solutions]\label{propSubCritBehav2} Under the assumptions of Proposition \ref{propSubCritBehav}, there exists a constant $\mu > 0$ such that
\[\phi_- \leq \phi \leq 1 + \left\|\phi\right\|_{L^\infty}^{2 \kappa+4} \left(\frac{\rho}{\mu}\right)^\delta \quad\text{ on } M_\mu.\]
\end{prop}

Remark that this proposition regards in some sense the individual behavior at infinity of the solutions, meaning that the upper bound for the function $\phi-1$ is proportional to $\left\|\phi\right\|_{L^\infty}^{2\kappa+4}$. See also \cite[Theorem 3.3]{Gicquaud}. Note also that the proposition remains true when $\epsilon = 0$, i.e. for solutions of the (critical) constraint equations.

\begin{proof} From the proof of Proposition \ref{propLichEq}, we know that $\phi \geq \phi_-$. Therefore we only need to prove the second inequality. Arguing as in the proof of the previous proposition, one has \[\left|\sigma\right|^2_g \leq \left(\alpha + \beta~\left\|\phi\right\|_{L^\infty}^{2 \kappa + 4}\right) \rho^{2\delta}.\] Since $\phi \geq \phi_-$ and $\phi_- \to 1$ at infinity, it follows that $\left\|\phi\right\|_{L^\infty} \geq 1$. Define $v_+ = k \left(\frac{\rho}{\mu}\right)^\delta$ and $\phi_+ = 1 + v_+$. It suffices to prove that if $k = \left\|\phi\right\|_{L^\infty}^{2\kappa+4}$ and $\mu > 0$ is small enough, then $\phi_+$ is a super-solution of the Lichnerowicz equation \eqref{eqLichnerowicz} on $M_\mu$. Indeed, in this case we have $\phi_+ \geq \phi$ on $\partial M_\mu = \rho^{-1}(\mu)$ and the proof of Proposition \ref{propLichEq} ensures that $\phi \leq \phi_+$ on $M_\mu$.

Since $\phi_+ \geq 1$, $\phi_+$ will be a super-solution provided that
\[-\frac{4(n-1)}{n-2} \Delta \phi_+ - n(n-1) \phi_+ + n(n-1) \tau^2 \phi_+^{\kappa+1} \geq \left|\sigma\right|^2_g.\] By convexity, one has $\phi_+^{\kappa+1} = (1+v_+)^{\kappa+1} \geq 1 + (\kappa+1) v_+$. Hence it will be enough to show that \[-\frac{4(n-1)}{n-2} \Delta v_+ - n(n-1) v_+ + n(n-1) (\kappa+1) \tau^2 v_+ + n(n-1) (\tau^2 - 1) \geq \left|\sigma\right|^2_g.\] Computing as in the previous proposition, we see that this inequality can be rewritten as \[\left[\frac{4(n-1)}{n-2} (\delta+1)(n-\delta) + o(1)\right] v_+ + n(n-1) (\tau^2 - 1) \geq \left|\sigma\right|^2_g,\] where the $o(1)$-term depends only on $\tau$ and $\rho$. Since $\delta \in (0; n)$, we can choose $\mu > 0$ small enough so that $\frac{4(n-1)}{n-2} (\delta+1)(n-\delta) + o(1) \geq \epsilon > 0$ for some constant $\epsilon$ on $M_\mu$ (remark that $\epsilon$ 	and $\mu$ can be chosen independent of $(\phi, \psi)$). We finally deduce that if \[\epsilon v_+ + n(n-1) (\tau^2-1) \geq \left(\alpha + \beta~\left\|\phi\right\|_{L^\infty}^{2 \kappa + 4}\right) \rho^{2\delta},\] then $\phi_+$ is a super-solution on $M_\mu$. Diminishing $\mu$ if necessary, it is easy to see that this inequality can be satisfied. This completes the proof of the proposition.
\end{proof}

\subsection{A remark on almost CMC solutions}\label{secRemarkAlmostCMC}

In this section we show how the technique that we have developed gives an extension of the result of \cite{IsenbergPark} by enlarging the range of decay at infinity for $\tau$ and $\sigma_0$. The main difference between the theorem below and Theorem \ref{thmMainTheorem} is that in the almost CMC case we can prove that the solution is unique.

\begin{theorem}[Existence and uniqueness of solutions for the almost CMC constraint equations]\label{thmNearCMC} Let $(M, g)$ be a $C^{l, \beta}$-asymptotically hyperbolic manifold with $l + \beta \geq 2$. Suppose that $p \in (n; \infty)$ and $\delta \in (0; n)$ are such that $\left|\delta + \frac{n-1}{p} - \frac{n-1}{2}\right| < \frac{n+1}{2}$ and let $\sigma_0 \in X^{2, p}_\delta$ be a TT-tensor. Let $\tau : M \to \bR$ be a positive function such that $\tau - 1 \in X^{1, p}_\delta$. Then there exists a positive constant $\lambda$ depending on $(M, g)$, $p$, $\delta$, $\left\|\sigma_0\right\|_{L^\infty}$ and $\min_M \tau$ (which is strictly positive by assumption) such that if \[\left\|d\tau\right\|_{X^{0, p}_\delta} < \lambda,\]  then the system \eqref{eqLichnerowicz}-\eqref{eqVector} admits a unique solution $(\phi, \psi)$, such that $\phi \in X^{2, p}_0(M, \bR)$, $\phi > 0$, $\phi \to 1$ at infinity, and $\psi \in X^{2, p}_\delta(M, T^*M)$. 
\end{theorem}

\begin{proof}
The proof goes as in Proposition \ref{propSubCrit} except that we work directly with the (critical) constraint equations \eqref{eqLichnerowicz}-\eqref{eqVector}. In particular, first we define a map $F$ from $C_\Lambda = \left\{ \phi \in L^\infty | \phi_- \leq \phi \leq \Lambda \right\}$ to $L^\infty$ and seek a condition that will ensure that $F$ maps $C_\Lambda$ into itself. The same calculation as in Subsection \ref{secSubCrit} shows that $C_\Lambda$ is stable provided that \[a \left( \left\| \sigma_0 \right\|_{L^\infty} + C \Lambda^{\kappa+2} \left\| d\tau \right\|_{X^{0, p}_\delta} \right)^{\frac{1}{\kappa+2}} + b \leq \Lambda,\] where the constants $a, b$ and $C$ are the same as before. As a consequence, if $a \left( C \left\| d\tau \right\|_{X^{0, p}_\delta} \right)^{\frac{1}{\kappa+2}} < 1$, then, for any $\Lambda > 0$ large enough, $F$ maps $C_\Lambda$ into itself.\\

Next we show that, if $(\phi, \psi)$ is a solution of the constraint equations, then the function $\phi$ belongs to $C_\Lambda$. Indeed, remark that the function \[f: m \mapsto m - a \left( \left\| \sigma_0 \right\|_{L^\infty} + C m^{\kappa+2} \left\| d\tau \right\|_{X^{0, p}_\delta} \right)^{\frac{1}{\kappa+2}} - b\] is increasing. Let $m = \max_M \phi$. Repeating the calculation at the beginning of Subsection \ref{secSubCrit}, we see that the following inequality holds: \[m \leq a \left( \left\| \sigma_0 \right\|_{L^\infty} + C m^{\kappa+2} \left\| d\tau \right\|_{X^{0, p}_\delta} \right)^{\frac{1}{\kappa+2}} + b.\] Hence $f(m) \leq 0$. Since $f(\Lambda) \geq 0$ by assumption, we obtain that $m \leq \Lambda$, hence  $\phi \in C_\Lambda$.\\

Now it remains to prove existence and uniqueness of a solution $(\phi, \psi)$ such that $\phi \in C_\Lambda$. We will do it by showing that the function $F: C_\Lambda \to C_\Lambda$ is a contracting map provided that $\left\|d\tau\right\|_{X^{0, p}_\delta}$ is small enough. First note that for any $\phi \in C_\Lambda$, the associated $\sigma$ is uniformly bounded in the $L^\infty$-norm by some constant $K$ depending only on $\left\|\sigma_0\right\|_{L^\infty}$, $\Lambda$ and $\left\|d\tau\right\|_{X^{0, p}_\delta}$. The function mapping $\sigma \in L^\infty$ to the solution $\phitil \in L^\infty$ of the Lichnerowicz equation \eqref{eqLichnerowicz} is Lipschitz on the set $\left\{ \sigma \in L^\infty \vert \left\|\sigma\right\| \leq K\right\}$, which means that there exists a constant $C > 0$ such that for all $\sigma_1$ and $\sigma_2$ in this set, one has \[\left\|\phitil_1 - \phitil_2\right\|_{L^\infty} \leq C \left\|\sigma_1 - \sigma_2\right\|_{L^\infty},\] where $\phitil_1$ (resp. $\phitil_2$) is the solution of the Lichnerowicz equation \eqref{eqLichnerowicz} associated to $\sigma_1$ (resp. $\sigma_2$). Hence, if $\phi_1, \phi_2 \in C_\Lambda$, one obtains (the constant $C$ can vary from line to line)

\begin{align*}
\left\|F(\phi_1) - F(\phi_2)\right\|_{L^\infty}
	& \leq C \left\|\sigma_1 - \sigma_2\right\|_{L^\infty}\\
	& \leq C \left\|\lie \psi_1 - \lie \psi_2\right\|_{L^\infty}\\
	& \leq C \left\|\psi_1 - \psi_2\right\|_{X^{2, p}_\delta}\\
	& \leq C \left\|\phi_1^{\kappa+2} - \phi_2^{\kappa+2}\right\|_{L^\infty} \left\|d\tau\right\|_{X^{0, p}_\delta}\\
	& \leq C \Lambda^{\kappa+1} \left\|d\tau\right\|_{X^{0, p}_\delta} \left\|\phi_1 - \phi_2\right\|_{L^\infty}.\\
\end{align*}

We conclude that if $C \Lambda^{\kappa+1} \left\|d\tau\right\|_{X^{0, p}_\delta} < 1$, then $F$ is a contracting map. By the Banach fixed point theorem, there exists a unique solution $(\phi, \psi)$ of the constraint equations such that $\phi \in C_\Lambda$. This completes the proof.
\end{proof}

\section{Convergence of the subcritical solutions and compactness}\label{secBlowUp}

In this section we will always assume given a $C^{l, \beta}$-asymptotically hyperbolic manifold $(M, g)$ with $l+\beta \geq 2$ and constant scalar curvature $\scal = -n(n-1)$, a positive function $\tau : M \to \bR$, such that $\tau - 1 \in X^{1, p}_\delta$ for some $p > n$ and $\delta \in (0; n)$, and a TT-tensor $\sigma_0 \in X^{1, p}_\delta$. Without loss of generality, we can assume that $\tau$ is non-constant: $d\tau \neq 0$.\\

In order to study the behavior of the solutions of the $\epsilon$-sub-critical equations when $\epsilon$ tends to zero we introduce a quantity that measures their growth. If $(\phi, \psi)$ is a solution of the $\epsilon$-sub-critical equations \eqref{eqSubCritConstraint}, we define the \textbf{energy} of $(\phi, \psi)$ by

\begin{equation}
\label{eqDefEnergy}
\gamma_\beta(\phi, \psi) = \int_M \rho^{-2 \beta} \left| \lie \psi \right|_g^2 d \mu_g,
\end{equation}
where $\beta$ satisfies $\beta + \frac{n-1}{2} < \delta$ (so that the integral converges) and $\left|\beta\right| < \frac{n+1}{2}$. The value of $\beta$ will appear to be not relevant since Proposition \ref{propSubCritBehav} asserts that everything ``happens" in a compact subset. To avoid difficulties with small energies, we set \[\gammatil_\beta(\phi, \psi) = \max \left\{\gamma_\beta(\phi, \psi), 1\right\}.\]

Our first aim is to prove the following proposition:

\begin{prop}\label{propBoundPhi}
There exists a constant $C > 0$ such that for any $\epsilon \in [0; 1)$, and for any solution $(\phi, \psi)$ of the $\epsilon$-constraint equations, we have \[\left\| \phi \right\|_{L^\infty} \leq C  \gammatil^{\frac{1}{2\kappa+4}}_\beta(\phi, \psi),\]
for any $\beta$ such that $-1<\beta+\frac{n-1}{2}<0$.
\end{prop}

Remark that in this proposition, we allow the value $\epsilon = 0$. This will be useful to prove compactness of the set of solutions of the constraint equations.\\

The proof consists in several steps. We first note the following inequality that we will need at some point: for any nonnegative $a$ and $b$ with $a \geq b$, and any $q \geq 1$, we have \[(a-b)^q \leq a^q - b^q.\] To prove this fact, simply remark that the inequality holds when $a = b$, and integrate the obvious inequality \[q (a-b)^{q-1} \leq q a^{q-1}\] with respect to $a$. Denote $\sigma = \sigma_0 + \lie \psi$ and let $\phi_-$ be the solution of the prescribed scalar curvature \eqref{eqPrescScal}. The main difficulties in the non-compact case compared to the compact case (see \cite{DahlGicquaudHumbert}) are that the manifold $(M, g)$ has infinite volume, and that one needs to handle weights properly. We first derive two important lemmas:

\begin{lemma}\label{lmEstimatePhi} Let $q \geq 0$ be arbitrary and $\beta'$ be such that $\beta' \leq \beta$ and $\beta' < - \frac{n-1}{2}$. Let $\phibar = \frac{\phi - \phi_-}{\gammatil_\beta^\alpha}$, where $\alpha = \frac{1}{2\kappa+4}$, and $\sigmatil = \gammatil_\beta^{-\frac{1}{2}} \sigma$, then there exists a constant $C_{\beta'}$ which depends only on $\beta'$ such that if $\phi$ is a solution of the Lichnerowicz equation, then  the following inequality holds
\begin{equation}
\label{eqLpIneq1}
-\frac{C_{\beta'}}{\gammatil_{\beta}^{\alpha \kappa}} \left(\int_M \rho^{-2\beta'} \phibar^{2\kappa+4+q} d\mu_g \right)^{\frac{\kappa+4+q}{2\kappa+4+q}}
+ n(n-1) \tau_0^2 \int_M \rho^{-2 \beta'} \phibar^{2\kappa+4+q} d\mu_g \leq \int_M \rho^{-2\beta'} \left|\sigmatil\right|_g^2 \phibar^q d\mu_g,
\end{equation}
where $\tau_0 = \min_M \tau > 0$.
\end{lemma}

Remark that this inequality ensures that if $\int_M \rho^{-2\beta'} \left|\sigmatil\right|_g^2 \phibar^q d\mu_g$ is bounded then $\int_M \rho^{-2\beta'} \phibar^{2\kappa+4+q} d\mu_g$ cannot be arbitrary large since $\gammatil_\beta$ is bounded from below.

\begin{proof}
We first rewrite the Lichnerowicz equation \eqref{eqLichnerowicz} as
\[-\frac{4(n-1)}{n-2} \Delta \left(\phi - \phi_-\right) - n (n-1) \left(\phi - \phi_-\right) + n(n-1) \tau^2 \left(\phi^{\kappa+1} - \phi_-^{\kappa+1}\right) - \left|\sigma\right|_g^2 \phi^{-\kappa-3} = 0.\] Multiplying this equation by $\rho^{- 2 \beta'} (\phi-\phi_-)^{\kappa+3+q}$, and integrating over $M$ we obtain
\begin{align*}
0 & = \frac{4(n-1)}{n-2} \int_M \left\langle d  \left(\rho^{-2\beta'} (\phi-\phi_-)^{\kappa+3+q}\right), d  (\phi-\phi_-)\right\rangle_g d\mu_g
		- n(n-1) \int_M \rho^{-2\beta'} (\phi-\phi_-)^{q + \kappa + 4} d\mu_g\\
	& + n(n-1) \int_M \tau^2 \rho^{-2\beta'} \left(\phi^{\kappa+1} - \phi_-^{\kappa+1}\right) (\phi-\phi_-)^{\kappa+3+q} d\mu_g
		- \int_M \rho^{-2\beta'} \left|\sigma\right|_g^2 \phi^{-\kappa-3} (\phi-\phi_-)^{\kappa+3+q} d\mu_g,
\end{align*}
hence
\begin{align*}
0 & \geq \frac{4(n-1)}{n-2} \int_M \left\langle d  \left(\rho^{-2\beta'} (\phi-\phi_-)^{\kappa+3+q}\right), d  (\phi-\phi_-)\right\rangle_g d\mu_g
		- n(n-1) \int_M \rho^{-2\beta'} (\phi-\phi_-)^{q + \kappa + 4} d\mu_g\\
	& + n(n-1) \int_M \tau^2 \rho^{-2\beta'} (\phi-\phi_-)^{2\kappa+4+q} d\mu_g - \int_M \rho^{-2\beta'} \left|\sigma\right|_g^2 (\phi-\phi_-)^q d\mu_g.
\end{align*}
Expanding the first term, we get
\begin{align*}
	&   \frac{4(n-1)}{n-2} \int_M \left\langle d  \left(\rho^{-2\beta'} (\phi-\phi_-)^{\kappa+3+q}\right), d  (\phi-\phi_-)\right\rangle_g d\mu_g\\
	& = \frac{4(n-1)}{n-2} \int_M \left(\rho^{-2\beta'} \left\langle d  (\phi-\phi_-)^{\kappa+3+q}, d  (\phi-\phi_-)\right\rangle_g
		+ \frac{1}{\kappa+4+q} \left\langle d  \rho^{-2\beta'}, d  (\phi-\phi_-)^{\kappa+4+q}\right\rangle_g \right)d\mu_g\\
	& = \frac{4(n-1)}{n-2} \int_M \left(\rho^{-2\beta'} \frac{\kappa+3+q}{\left(\frac{\kappa+4+q}{2}\right)^2} \left| d  (\phi-\phi_-)^{\frac{\kappa+4+q}{2}}\right|^2_g
		- \frac{1}{\kappa+4+q} (\phi-\phi_-)^{\kappa+4+q} \Delta \rho^{-2\beta'} \right)d\mu_g\\
	& \geq -\frac{4(n-1)}{n-2} \frac{1}{\kappa+4+q} \int_M (\phi-\phi_-)^{\kappa+4+q} \Delta \rho^{-2\beta'} d\mu_g\\
	& \geq - C_{\beta'} \int_M (\phi-\phi_-)^{\kappa+4+q} \rho^{-2 \beta'} d\mu_g,
\end{align*}
where $C_{\beta'} > 0$ is a constant that depends only on $\beta'$ (but that can vary from line to line). As a consequence, we have proved that
\begin{align*}
0 & \geq - \left[n(n-1) + C_{\beta'}\right] \int_M \rho^{-2\beta'} (\phi-\phi_-)^{q + \kappa + 4} d\mu_g\\
	& + n(n-1) \int_M \tau^2 \rho^{-2\beta'} (\phi-\phi_-)^{2\kappa+4+q} d\mu_g - \int_M \rho^{-2\beta'} \left|\sigma\right|_g^2 (\phi-\phi_-)^q d\mu_g.
\end{align*}

Let $v = \frac{2 \kappa + 4 + q}{\kappa + 4 + q}$, $u = \frac{2 \kappa + 4 + q}{\kappa}$ and $\mu = \frac{\kappa}{2 \kappa + 4 + q} = \frac{1}{u}$. Then $\frac{1}{u} + \frac{1}{v} = 1$. We estimate the first term by using the H\"older inequality:
\begin{align*}
\int_M \rho^{-2\beta'} (\phi-\phi_-)^{q + \kappa + 4} d\mu_g
	&  = \int_M \rho^{-2\beta'\mu} \rho^{-2\beta'(1-\mu)} (\phi-\phi_-)^{q + \kappa + 4} d\mu_g\\
	& \leq \left(\int_M \left(\rho^{-2\beta'\mu}\right)^u d\mu_g \right)^{\frac{1}{u}} \left(\int_M \left(\rho^{-2\beta'(1-\mu)} (\phi-\phi_-)^{q + \kappa + 4}\right)^v d\mu_g \right)^{\frac{1}{v}}\\
	& \leq \left(\int_M \rho^{-2\beta'} d\mu_g \right)^{\frac{\kappa}{2 \kappa + 4 + q}} \left(\int_M \rho^{-2\beta'} (\phi-\phi_-)^{2 \kappa + 4 + q} d\mu_g \right)^{\frac{\kappa+4+q}{2\kappa+4+q}}.
\end{align*}
Substituting this inequality into the previous one, we get
\[
\begin{array}{c}
		 \displaystyle{-C_{\beta'} \left(\int_M \rho^{-2\beta'} (\phi-\phi_-)^{2 \kappa + 4 + q} d\mu_g \right)^{\frac{\kappa+4+q}{2\kappa+4+q}}
		 + n(n-1) \int_M \tau^2 \rho^{-2\beta'} (\phi-\phi_-)^{2\kappa+4+q} d\mu_g}\\
\leq \displaystyle{\int_M \rho^{-2\beta'} \left|\sigma\right|_g^2 (\phi-\phi_-)^q d\mu_g.}
	\end{array}
\]
Inequality \eqref{eqLpIneq1} follows.
\end{proof}

\begin{lemma}[Elliptic regularity]\label{lmModifiedSchauder}
Let $p, q \in (1; \infty)$ and let $\delta, \delta' \in \bR$ be such that $\delta' > \delta + \frac{n-1}{p}$. Let $\psi \in L^q_{\delta'}(M, T^* M)$ be such that $\Delta_L \psi \in L^p_\delta(M, T^*M)$. Then $\psi \in W^{2, p}_\delta(M, T^*M)$, and there exists a constant $C > 0$ independent of $\psi$ such that
\[\left\|\psi\right\|_{W^{2, p}_\delta} \leq C \left(\left\|\Delta_L \psi\right\|_{L^p_\delta} + \left\|\psi\right\|_{L^q_{\delta'}}\right).\]
\end{lemma}

\begin{proof}
By \cite[Lemma 2.2]{LeeFredholm}, $M$ can be covered by a countable number of M\"obius charts $B_i = \left(\Phi_{x_i}\right)^{-1} (B_1)$ centered at $x_i$ and of ``radius one'', and there exists a constant $N < \infty$ such that for each $i$ the set $\{j\, \vert\, \tilde{B}_i \bigcap \tilde{B}_j \neq \emptyset\}$ contains less than $N$ elements, where $\tilde{B}_i = \left(\Phi_{x_i}\right)^{-1} (B_2)$ is the M\"obius chart centered at $x_i$ and of ``radius two''. Set $\rho_i = \rho(x_i)$. By classical interior estimates, there exists a constant $C > 0$ which does not depend on $i$ and $\psi$, and  such that (thereafter, $C$ can vary from line to line):
\[\left\|\psi\right\|_{W^{2, p}(B_i)}   \leq C \left(\left\|\Delta_L \psi\right\|_{L^p(\tilde{B}_i)} + \left\|\psi\right\|_{L^q(\tilde{B}_i)}\right).\]
Hence
\[\left\|\psi\right\|_{W^{2, p}(B_i)}^p \leq C \left(\left\|\Delta_L \psi\right\|_{L^p(\tilde{B}_i)}^p + \left\|\psi\right\|_{L^q(\tilde{B}_i)}^p\right).\]
Multiplying by $\rho_i^{-\delta p}$ and summing over $i$, by \cite[Lemma 3.5]{LeeFredholm} we get
\[\left\|\psi\right\|_{W^{2, p}_\delta}^p \leq C \left(\left\|\Delta_L \psi\right\|_{L^p_\delta}^p + \sum_i \rho_i^{-p\delta} \left\|\psi\right\|_{L^q(\tilde{B}_i)}^p\right).\]

However,
\begin{align*}
\sum_i \rho_i^{-p\delta} \left\|\psi\right\|_{L^q(\tilde{B}_i)}^p 
	&  =   \sum_i \left(\rho_i^{-q \delta} \int_{\tilde{B}_i} \left|\psi\right|_g^q d\mu_g \right)^{\frac{p}{q}}\\
	& \leq c \sum_i \left(\rho_i^{-q (\delta-\delta')} \int_M \rho^{-q \delta'} \left|\psi\right|_g^q d\mu_g \right)^{\frac{p}{q}}\\
	& \leq c \left\|\psi\right\|^p_{L^q_{\delta'}} \sum_i \rho_i^{-p (\delta-\delta')},
\end{align*}
where $c$ is a constant that is independent of $\psi$. By our assumption on $p$, $\delta$ and $\delta'$, this last sum converges (see the proof of Lemma \ref{lmYoung2}). Thus we obtain
\[\left\|\psi\right\|_{W^{2, p}_\delta}^p \leq C \left(\left\|\Delta_L \psi\right\|_{L^p_\delta}^p + \left\|\psi\right\|^p_{L^q_{\delta'}}\right),\]
which proves the lemma.
\end{proof}

We now return to the proof of Proposition \ref{propBoundPhi}. We recall the notations $\sigmatil=\gammatil_{\beta} ^{-\frac{1}{2}}\sigma$, $\alpha = \frac{1}{2\kappa+4}$ and introduce $\psitil=\gammatil_{\beta} ^{-\frac{1}{2}} \psi$ and $\phitil=\frac{\phi}{\gammatil _{\beta} ^{\alpha}}$. In the course of the proof, ``bounded'' will mean ``bounded by a constant that does not depend on $(\epsilon, \phi, \psi)$''.\\

\noindent $\bullet$ {\sc Step 1:} Bound for $\phibar^{\kappa+2}$ in the $L^2_\beta$-norm.\\

Setting $q = 0$ and $\beta' = \beta$ in inequality \eqref{eqLpIneq1} of Lemma \ref{lmEstimatePhi}, we obtain that  $\phibar^{\kappa+2}$ is bounded in the $L^2_\beta$-norm by some constant depending only on $\int_M \rho^{-2 \beta} \left|\sigmatil\right|_g^2 d\mu_g \leq 2 \int_M \rho^{-2 \beta} \left|\sigmatil_0\right|_g^2 d\mu_g + 2 \int_M \rho^{-2 \beta} \left|\lie \psitil\right|_g^2 d\mu_g \leq \frac{C(\sigma_0, \beta)}{\gammatil_\beta} + 2$, by definition of $\psitil$ and since $\gammatil_\beta \geq \gamma_\beta$.\\

\noindent $\bullet$ {\sc Step 2:} Bound for $\psitil$.\\

The equation for $\psitil$ reads \[\Delta_L \psitil = (n-1) \gammatil_\beta^{\frac{-\epsilon}{2\kappa+4}} \phitil^{\kappa+2-\epsilon} d\tau.\] 
We denote $\eta = \min_M \phi_- > 0$ and remark that $\phitil \geq \gammatil_\beta^{-\frac{1}{2\kappa+4}} \phi_- \geq \gammatil_\beta^{-\frac{1}{2\kappa+4}} \eta$. Hence
\begin{align*}
\left|\Delta_L \psitil\right|_g
	&  =   (n-1) \gammatil_\beta^{\frac{-\epsilon}{2\kappa+4}} \phitil^{-\epsilon} \phitil^{\kappa+2} \left|d\tau\right|_g\\
	& \leq (n-1) \eta^{-\epsilon} \phitil^{\kappa+2} \left|d\tau\right|_g.\\
\end{align*}

Since $\eta > 0$ and $\epsilon \in [0; 1)$, we conclude that there exists a constant $C > 0$ independent of $\phi$ and $\psi$ such that

\[\left|\Delta_L \psitil\right|_g \leq C \phitil^{\kappa+2} \left|d\tau\right|_g.\]

Let $r_0$ be such that $\frac{1}{r_0} = \frac{1}{2} + \frac{1}{p}$ and $\beta'$ be such that $\beta' + \frac{n-1}{r_0} = \beta + \frac{n-1}{2} < 0 < \delta$. Note that $r_0 > 1$. We claim that $\phitil^{\kappa+2}d\tau$ is bounded in $L^{r_0}_{\beta'}$. Indeed,

\[\phitil^{\kappa+2} = \left(\gammatil_{\beta}^{-\alpha}\phi_-+\phibar\right)^{\kappa+2}\leq 2^{\kappa+2} \left( \gammatil_\beta^{-\frac{1}{2}} \phi_-^{\kappa+2} + \phibar^{\kappa+2} \right).\]

Since $\phi_-$ is bounded and due to our choice of $r_0$ and $\beta'$, by Lemma \ref{lmEmbeddingXIntoL} we obtain $\gammatil_\beta^{-\frac{1}{2}} \phi_-^{\kappa+2} |d\tau|_g \in L^{r_0}_{\beta'}$. Similarly, by Lemma \ref{lmYoung2} we deduce that $\phibar^{\kappa+2} |d\tau|_g$ is bounded in $ L^{r_0}_{\beta'}$. We conclude that $\Delta_L \psitil$ is bounded in $L^{r_0}_{\beta'}$. Finally, $\beta' + \frac{n-1}{r_0} - \frac{n-1}{2}=\beta \in \left(-\frac{n+1}{2}; \frac{n+1}{2}\right)$, and by Proposition \ref{propIsomVectLaplacian} we see that $\psitil$ is bounded in $W^{2, r_0}_{\beta'}$.\\

\noindent $\bullet$ {\sc Step 3:} Induction step.\\

Recall our definition of $r_0$ and $\beta'$ made on the previous step and set $p_0 = 2$, $\beta_0 = \beta$, $\beta'_0 = \beta'$, and $\beta''_0=\beta_0$. In Step 1 and Step 2 we have shown that $\phibar^{\kappa+2}$ is bounded in the $L_{\beta''_0}^{p_0}$ norm and $\psitil$ is bounded in the $W_{\beta'_0}^{2,r_0}$ norm.\\

In what follows we will inductively construct sequences $p_i$, $r_i$, $\beta_i$, $\beta'_i$ such that $\phibar^{\kappa+2}$ is bounded in the $L_{\beta''_i}^{p_i}$, where $\beta''_i=\frac{2\beta_i}{p_i}$, and $\psitil$ is bounded in $W_{\beta'_i}^{2,r_i}$ norm.\\

If $r_i<n$  we define $r^*_i$ and $q_i$ by the formulas $\frac{1}{r_i^*} = \frac{1}{r_i} - \frac{1}{n}$ and $\frac{2}{r_i^*} + \frac{q_i}{(\kappa+2) p_i} = 1$ respectively. Then $p_{i+1}$, $r_{i+1}$, $\beta_{i+1}$, and $\beta'_{i+1}$ are computed as follows:

\begin{equation}\label{IndDefs}
\begin{split}
p_{i+1}															& = 2 + \frac{q_i}{\kappa+2},\\
\frac{1}{r_{i+1}}										& = \frac{1}{p}+ \frac{1}{p_{i+1}},\\
\beta_{i+1}													& = \beta'_i + \frac{q_i}{2\kappa+4} \beta''_i,\\
\beta'_{i+1} + \frac{n-1}{r_{i+1}}	& = \min \left\{ \beta''_{i+1} + \frac{n-1}{p_{i+1}}, \beta'_i - 1\right\}.
\end{split}
\end{equation}

We note that

\begin{equation}\label{IndFormulas}
\begin{split}
p_{i+1}			& = 		p_i \left(1+\frac{2}{n}-\frac{2}{p}\right)=2\left(1+\frac{2}{n}-\frac{2}{p}\right)^i,\\
\beta_{i+1}	& \leq  \beta_i \left(1+\frac{2}{n}-\frac{2}{p}\right)-\frac{n-1}{p}< \beta \left(1+\frac{2}{n}-\frac{2}{p}\right)^i,
\end{split}
\end{equation}
which is checked straightforwardly. Now assume by induction that $\phibar^{\kappa+2} \in L_{\beta''_i}^{p_i}$ and $\psitil \in W_{\beta'_i}^{2,r_i}$ are bounded.\\

By Proposition \ref{propRellich}, $\sigmatil_0\in L_{\delta}^\infty$. Moreover, by \eqref{IndDefs} and \eqref{IndFormulas} we have 

\begin{equation*}
\beta'_i+\frac{n-1}{r_i^*}<\beta''_i+\frac{n-1}{p_i}=\frac{\beta_i+\frac{n-1}{2}}{\frac{p_i}{2}}
	< \frac{\beta\left(1+\frac{2}{n}-\frac{2}{p}\right)^{i-1}+\frac{n-1}{2}}{\left(1+\frac{2}{n}-\frac{2}{p}\right)^{i-1}}
	< \beta+\frac{n-1}{2}
	< \delta.
\end{equation*}

Therefore, as a simple computation shows, $\sigmatil_0 \in L_{\beta'_i}^{r_i^*}$. Finally, by our choice of $r_i ^*$, it follows from Proposition \ref{propRellich} that $\lie \psitil$, and hence $\sigmatil$, is bounded in $L_{\beta'_i}^{r_i^*}$.\\

Since $\left|\sigmatil\right|_g^2$ is bounded in $L^{\frac{r_i^*}{2}}_{2\beta'_i}$ and $\phibar^{q_i}$ is bounded $L^{\frac{p_i (\kappa+2)}{q_i}}_{\frac{q_i \beta''_i}{\kappa+2}}$, by Lemma \ref{lmYoung1} we deduce that the integral \[\int_M \rho^{-2\beta_{i+1}} \left|\sigmatil\right|_g^2 \phibar^{q_i} d\mu_g\] is bounded. It then follows from \eqref{eqLpIneq1} that \[\int_M \rho^{-2 \beta_{i+1}} \phibar^{2\kappa+4+q_i} d\mu_g\] is bounded, which by \eqref{IndDefs} means that $\phibar^{\kappa+2}$ is bounded in $L^{p_{i+1}}_{\beta''_{i+1}}$.\\

By a calculation similar to the one done in Step 2, we get that $\phitil^{\kappa+2-\epsilon} d\tau$ is bounded in $L^{r_{i+1}}_{\beta'_{i+1}}$. Then it follows by the definition of $\beta'_{i+1}$ and Lemma \ref{lmModifiedSchauder} that $\psitil$ is bounded in $W^{2, r_{i+1}}_{\beta'_{i+1}}$.\\

\noindent $\bullet$ {\sc Step 4:} $L^\infty$ bound for $\phitil$.\\

Note that in Step 3 we have assumed that $r_i<n$, whereas it is obvious from \eqref{IndDefs} and \eqref{IndFormulas} that there exists $i_0$ such that $r_{i_0}\geq n$.\\

Suppose first that $r_{i_0}> n$. Since $\psitil\in W_{\beta'_{i_0}}^{2,r_{i_0}}$ and $\sigmatil _0 \in X_{\delta} ^{1,p}$, it follows from Proposition \ref{propRellich} that $\sigmatil \in C_{\gamma'}^{0, \alpha'}$ for some $\alpha' \in (0;1)$ and $\gamma'= \min (\beta'_{i_0},\delta)$. In particular, $\left|\sigmatil\right|_g^2$ is uniformly bounded on any compact subset. From Proposition \ref{propSubCritBehav} we know that there exists $\mu>0$ (independent of $\phi$) such that  $\phi$ attains its maximal value on the compact subset $K_\mu \subset M$. The same is true for $\phitil$, which satisfies the rescaled Lichnerowicz equation

\begin{equation}
\label{eqRescaledLichnerowicz}
\frac{1}{\gammatil_{\beta}^{\alpha\kappa}}\left(- \frac{4(n-1)}{n-2}\Delta\phitil - n(n-1)~\phitil\right) + n(n-1)\tau^2 \phitil^{\kappa+1} = \left|\sigmatil\right|_g^2 \phitil^{-3-\kappa}.
\end{equation}

Let $\lambda$ be such that $\frac{1}{\gammatil^\alpha} < \lambda < \sup_M \phitil$. We multiply equation \eqref{eqRescaledLichnerowicz} by $\phitil^{\kappa+3} \max\left(\phitil-\lambda, 0\right)^\epsilon$, which obviously has compact support, and integrate by parts. We obtain

\begin{align*}
\int_M \max\left(\phitil-\lambda, 0\right)^\epsilon \left|\sigmatil\right|_g^2 d\mu_g
	& = \frac{4(n-1)}{n-2} \frac{1}{\gammatil_{\beta}^{\alpha\kappa}}\int_M \left\langle d \phitil, d \left(\phitil^{\kappa+3} \max\left(\phitil-\lambda, 0\right)^\epsilon\right)\right\rangle_g d\mu_g\\
	& - n(n-1) \frac{1}{\gammatil_{\beta}^{\alpha\kappa}}\int_M \phitil^{\kappa+4} \max\left(\phitil-\lambda, 0\right)^\epsilon d\mu_g\\
	& + n(n-1) \int_M \tau^2 \phitil^{2\kappa+4} \max\left(\phitil-\lambda, 0\right)^\epsilon d\mu_g.	
\end{align*}

It can be easily seen that the first term of the right hand side is nonnegative since the integrand is of the form $\langle d \phitil, d  (f\circ\phitil)\rangle$ for some monotonically increasing function $f$. Hence 

\begin{align*}
\int_M \max\left(\phitil-\lambda, 0\right)^\epsilon \left|\sigmatil\right|_g^2 d\mu_g
	& \geq - n(n-1) \frac{1}{\gammatil_{\beta}^{\alpha\kappa}}\int_M \phitil^{\kappa+4} \max\left(\phitil-\lambda, 0\right)^\epsilon d\mu_g\\
	&			 + n(n-1) \int_M \tau^2 \phitil^{2\kappa+4} \max\left(\phitil-\lambda, 0\right)^\epsilon d\mu_g.	
\end{align*}

We now let $\epsilon \to 0^+$. In this case $\max\left(\phitil-\lambda, 0\right)^\epsilon$ converges to $\chi_{\{\phitil > \lambda\}}$, therefore

\[
\int_{\{\phitil>\lambda\}} \left|\sigmatil\right|_g^2 d\mu_g \geq - n(n-1) \frac{1}{\gammatil_{\beta}^{\alpha\kappa}}\int_{\{\phitil>\lambda\}} \phitil^{\kappa+4} d\mu_g
+ n(n-1) \int_{\{\phitil>\lambda\}
} \tau^2 \phitil^{2\kappa+4} d\mu_g.
\]

As a consequence, we have

\[
- \frac{n(n-1)}{\gammatil_{\beta}^{\alpha\kappa}} (\sup_{K_\mu}\phitil)^{\kappa+4} \mu_g\{\phitil>\lambda\} + n(n-1) \tau_0^2 \lambda^{2\kappa+4} \mu_g\{\phitil>\lambda\}
\leq \sup_{K_\mu} \left|\sigmatil\right|_g^2 \mu_g \{\phitil > \lambda\}.
\]

We divide the above inequality by $\mu_g\{\phitil>\lambda\}$ and let $\lambda$ tend to $\sup \phitil$ to obtain

\[- \frac{n(n-1)}{\gammatil_{\beta}^{\alpha\kappa}} (\sup_{K_\mu}\phitil)^{\kappa+4} + n(n-1) \tau_0^2 (\sup_{K_\mu}\phitil)^{2\kappa+4} \leq \sup_{K_\mu} \left|\sigmatil\right|_g^2.\]

Finally, since $\gammatil_\beta \geq 1$ and $\tau_0 > 0$, we deduce that $\phitil$ is bounded.\\

We conclude the proof by noting that the situation when $r_i=n$ can be transformed into the case $r_{i+1}>n$, which has just been considered. Indeed, recall that $W_{\beta'_i}^{2,r_i}\hookrightarrow W_{\widetilde{\beta'}_i}^{2,\widetilde{r}_i}$ for $\widetilde{r}_i < r_i$ and $\widetilde{\beta'}_i + \frac{n-1}{\widetilde{r}_i} < \beta'_i + \frac{n-1}{r_i}$. This means that we can reset $r_i=n-\epsilon'$ for some $\epsilon'>0$ and retain the already computed values of $p_i$, $\beta_i$ and $\beta'_i$. If we now repeat the computations of Step 3 we get $\psitil\in W_{\beta'_{i+1}} ^{2,r_{i+1}}$ for $r_{i+1}>n$, provided that $\epsilon'>0$ is chosen to be small enough.\\

We next prove that if the energy remains bounded as $\epsilon$ tends to zero, then there exists a solution of the constraint equations:

\begin{prop}[Existence of a solution when the energy is bounded]\label{propExistenceSolution}
Assume that there exists a sequence $\{(\epsilon_i, \phi_i, \psi_i)\}$, where $\epsilon_i \geq 0$, $\epsilon_i \to 0$, and $(\phi_i, \psi_i)$ are solutions of the $\epsilon_i$-constraint equations \eqref{eqSubCritConstraint} with $\psi_i \in X^{2, p}_\delta(M, T^*M)$, $\phi_i > 0$, and $\phi_i - 1 \in X^{2, p}_+(M, \bR)$, such that their energies \[\gammatil_i = \gammatil_\beta (\phi_i, \psi_i)\] are uniformly bounded by some constant $\lambda$. Then there exists a subsequence of $\{(\epsilon_i, \phi_i, \psi_i)\}$, denoted by the same name, such that $\phi_i$ converge in the $X^{2, p}_0$ norm to $\phi_\infty \in X^{2, p}_0$, and $\psi_i$ converge in the  $X^{2, p}_\delta$ norm to $\psi_\infty \in X^{2, p}_\delta$, where $(\phi_\infty,\psi_\infty)$ is a solution of the constraint equations \eqref{eqLichnerowicz} and \eqref{eqVector}.
\end{prop}

\begin{proof}
According to Proposition \ref{propBoundPhi}, the functions $\phi_i$ are uniformly bounded. As a consequence, by \eqref{eqSubCritConstraint} and Proposition \ref{propIsomVectLaplacian}, we obtain that the 1-forms $\psi_i$ are uniformly bounded in $X^{2, p}_\delta$, hence the sequence $\sigma_i = \sigma_0 + \lie \psi_i$ is bounded in $X^{1, p}_\delta$. By the Rellich theorem (Proposition \ref{propRellich}), up to extracting a subsequence of the $(\epsilon_i, \phi_i, \psi_i)$, we can assume that the sequence $\sigma_i$ converges in the $L^\infty$ norm to some $\sigma_\infty \in L^\infty$. It then follows by Proposition \ref{propLichEq}, that the functions $\phi_i$ converge to a function $\phi_\infty \in X^{2, p}_0$. The inclusion $X^{2, p}_0 \subset L^\infty$ being continuous, we also conclude that the sequence $\psi_i$ converges to some $\psi_\infty\in X^{2, p}_\delta$. It is then easily seen that $\sigma_\infty = \sigma_0 + \lie \psi_\infty$, so the pair $(\phi_\infty, \psi_\infty)$ is a solution of the constraint equations \eqref{eqLichnerowicz}-\eqref{eqVector}.
\end{proof}

\begin{prop}[Existence of a solution of the limit equation when the energy is unbounded]\label{propUnboundedEnergy}
Assume that there exists a sequence $\{(\epsilon_i, \phi_i, \psi_i)\}$, where $\epsilon_i \geq 0$, $\epsilon_i \to 0$, and $(\phi_i, \psi_i)$ solve the $\epsilon_i$-constraint equations, such that the energies $\gammatil_i = \gammatil_\beta(\phi_i, \psi_i) \to \infty$. Then there exists a subsequence of $\psitil_j = \gammatil_j^{-\frac{1}{2}} \psi_j$ which converges in the $L^\infty$-norm to a non-zero solution $\psitil_\infty \in X^{2, p}_\delta$ of the \textbf{limit equation}:

\begin{equation}
\Delta_L \psi = \lambda \sqrt{\frac{n-1}{n}} \left| \lie \psi \right|_g \frac{d \tau}{\tau}
\label{eqLimit}
\end{equation}
for some $\lambda \in [\eta; 1]$, where $\eta$ is the constant appearing in Proposition \ref{propSubCritBehav}.
\end{prop}

Before giving the proof of the proposition, we prove a simple lemma:

\begin{lemma}\label{lmContinuityQuadrForm}
The quadratic form $\gamma : \psi \in X^{2, p}_\delta \mapsto \int_M \rho^{-2\beta} \left|\lie \psi\right|^2_g d\mu_g$ is continuous.
\end{lemma}

\begin{proof}
Since by the assumption we have $\beta + \frac{n-1}{2} < \delta$ and $p > n > 2$, one needs a continuous injection $X^{2, p}_\delta(M, T^*M) \into W^{2, 2}_\beta(M, T^*M)$. $\gamma$ is obviously continuous when seen as a map $\gamma : W^{2, 2}_\beta (M, T^*M) \to \bR$. Composing these maps, we see that  $\gamma : X^{2, p}_\delta(M, T^*M) \to \bR$ is continuous.
\end{proof}

We can now return to the proof of the proposition:

\begin{proof}[Proof of Proposition \ref{propUnboundedEnergy}]
By Proposition \ref{propBoundPhi}, the functions $\phitil_i = \gammatil_i^{-\frac{1}{2\kappa+4}} \phi_i$ are uniformly bounded in the $L^\infty$-norm. Hence the sequence $\psitil_i = \gammatil_i^{-\frac{1}{2}} \psi_i$ is bounded in $X^{2, p}_\delta$ by the second line of Equation \eqref{eqSubCritConstraint}. By the Rellich theorem (Proposition \ref{propRellich}), up to extracting a sub-sequence, we can assume that the 1-forms $\psitil_i$ converge in the $C^{1, \alpha}_0$-norm, where $\alpha \in \left(0; 1-\frac{n}{p}\right)$ is arbitrary, to some 1-form $\psitil_\infty$. The tensors $\sigmatil_i = \gammatil_i^{-\frac{1}{2}} \sigma_0 + \lie \psitil_i$ converge to $\sigmatil_\infty = \lie \psitil_\infty \in C^{0,\alpha}_0$. Let $\phitil_\infty = \left(\frac{\left|\sigmatil_\infty\right|^2_g}{n(n-1) \tau^2}\right)^{\frac{1}{2\kappa+4}}$. We have to show that $\phitil_i \to \phitil_\infty$ in the $L^\infty$-norm. Indeed, in this case, the second line of Equation \eqref{eqSubCritConstraint} applied to $(\epsilon_i, \phi_i, \psi_i)$ can be rewritten

\begin{equation}\label{eqRescaledVector}
\Delta_L \psitil_i = (n-1) \gammatil_i^{-\frac{\epsilon_i}{2\kappa+4}} \phitil_i^{\kappa+2-\epsilon_i} d\tau,
\end{equation}

\noindent and $\phitil_i^{\kappa+2-\epsilon_i}$ converges to $\phitil_\infty^{\kappa+2}$. We claim that $\gammatil_i^{-\frac{\epsilon_i}{2\kappa+4}} \leq 1$ and that $\liminf \gammatil_i^{-\frac{\epsilon_i}{2\kappa+4}} \geq \eta$. Then, up to extracting a subsequence, we can assume that $\gammatil_i^{-\frac{\epsilon}{2\kappa+4}}$ converges to $\lambda \in [\eta; 1]$. So the right hand side of Equation \eqref{eqRescaledVector} converges to $(n-1) \lambda \phitil_\infty^{\kappa+2} d\tau$. By the isomorphism theorem (Proposition \ref{propIsomVectLaplacian}) applied to Equation \eqref{eqRescaledVector}, we then obtain that the sequence $\psitil_i$ converges in the $X^{2, p}_\delta$ norm to the solution $\psitil_\infty \in X^{2, p}_\delta$ of \[\Delta_L \psitil_\infty = (n-1) \lambda \phitil_\infty^{\kappa+2} d\tau.\] We deduce from this that $\psitil_\infty$ is a solution of the limit equation \eqref{eqLimit}. Note also that, if $i$ is large enough, the 1-forms $\psitil_i$ satisfy $\gamma(\psitil_i) = 1$, where $\gamma$ is as in Lemma \ref{lmContinuityQuadrForm}. Then it follows from Lemma \ref{lmContinuityQuadrForm} that $\gamma(\psitil_\infty) = 1$, which proves that $\psitil_\infty \neq 0$.\\

We first prove that $\gammatil_i^{-\frac{\epsilon_i}{2\kappa+4}} \in [\mu^{\epsilon_i} \eta; 1]$. If $\epsilon_i = 0$, we are done. Otherwise, since $\gammatil_i \geq 1$, $\gammatil_i^{-\frac{\epsilon_i}{2\kappa+4}} \leq 1$.  We also know from Proposition \ref{propSubCritBehav} that $\phi_i \leq \eta^{-\frac{1}{\epsilon_i}}$. By the isomorphism theorem (Proposition \ref{propIsomVectLaplacian}), \[\left\|\psi_i\right\|_{X^{2, p}_\delta} \leq C \left\|\phi_i^{\kappa+2-\epsilon_i} d\tau\right\|_{X^{0, p}_\delta} \leq C \left\|d\tau\right\|_{X^{0, p}_\delta} \eta^{1-\frac{\kappa+2}{\epsilon_i}}.\] So the energy satisfies the estimate $\gammatil_i = \gamma_\beta(\phi_i, \psi_i) \leq C' \eta^{2-\frac{2\kappa+4}{\epsilon_i}}$ for some constant $C' > 0$. Hence, \[\gammatil_i^{-\frac{\epsilon_i}{2\kappa+4}} \geq (C' \eta^2)^{-\frac{\epsilon_i}{2\kappa+4}} \eta.\]

We next show the convergence of the $\phitil_i$ to $\phitil_\infty$. We first prove that for any $\epsilon > 0$ and $i$ large enough we have $\phitil_i \leq \phitil_\infty + \epsilon$. Let $\epsilon > 0$ be arbitrary. Select a function $u_\epsilon \in C^{2,0}_0(M, \bR)$ such that $\phitil_\infty + \frac{\epsilon}{2} \leq u_\epsilon \leq \phitil_\infty + \epsilon$. Such a function exists since $\phitil_\infty$ is H\"older continuous with H\"older exponent $\frac{\alpha}{\kappa+2}$; namely, $u_\epsilon$ can be constructed by a regularization procedure. The functions $\phitil_i$ satisfy the rescaled Lichnerowicz equation \eqref{eqRescaledLichnerowicz}. We now claim that if $\gammatil_i$ is large enough, then $u_\epsilon$ is a super-solution of this equation. Indeed, by the definition of $u_\epsilon$, we have

\begin{align*}
n(n-1) \tau^2 u_\epsilon^{2\kappa+4}
	& \geq n(n-1) \tau^2 \left[\phitil_\infty^{2\kappa+4} + \left(\frac{\epsilon}{2}\right)^{2\kappa+4} \right]\\
	& = \left|\sigmatil_\infty\right|_g^2 + n(n-1) \tau^2 \left(\frac{\epsilon}{2}\right)^{2\kappa+4}.
\end{align*}
Selecting $i_0$ large enough, we can assume that $\left|\sigmatil_\infty\right|_g^2 \geq \left|\sigmatil_i\right|_g^2 - \frac{1}{2} n(n-1) \tau^2 \left(\frac{\epsilon}{2}\right)^{2\kappa+4}$ on $M$ for $i\geq i_0$, which yields
\[ n(n-1) \tau^2 u_\epsilon^{2\kappa+4} \geq \left|\sigmatil_i\right|_g^2 + \frac{1}{2} n(n-1) \tau^2 \left(\frac{\epsilon}{2}\right)^{2\kappa+4}.\]
Since $u_\epsilon$ is bounded from above by $\sup_M \phitil_\infty + \epsilon$, we deduce that, for some $\epsilon' > 0$,
\[ n(n-1) \tau^2 u_\epsilon^{\kappa+1} \geq \left|\sigmatil_i\right|_g^2 u_\epsilon^{-\kappa-3} + \epsilon'.\]
Now recall that the second order derivatives of $u_\epsilon$ are bounded on $M$, and that, by increasing $i_0$, the energy can be made as large as needed. Then  it is easy to see that if $i \geq i_0$ then $u_\epsilon$ is a super-solution of the rescaled Lichnerowicz equation \eqref{eqRescaledLichnerowicz}:
\[\frac{1}{\gammatil_i^{\alpha\kappa}}\left(- \frac{4(n-1)}{n-2}\Delta u_\epsilon - n(n-1)~u_\epsilon\right) + n(n-1)\tau^2 u_\epsilon^{\kappa+1} \geq \left|\sigmatil_i\right|^2 u_\epsilon^{-3-\kappa}.\]
Arguing as in the proof of Proposition \ref{propLichEq}, we conclude that if $i \geq i_0$, then $\phitil_i \leq u_\epsilon \leq \phitil_\infty + \epsilon$.\\

The lower bound is slightly more difficult to prove because there is to take into consideration the fact that $\phitil_\infty$ can vanish. We indicate the differences. Select a function $v_\epsilon$ such that $v_\epsilon \in C^{2,0}_0(M, \bR)$ such that $\phitil_\infty - \epsilon \leq v_\epsilon \leq \phitil_\infty - \frac{\epsilon}{2}$. Then it can be proved (exactly as for $u_\epsilon$) that for $i \geq i_0$ for some $i_0$ large enough, $v_\epsilon$ is a sub-solution of the rescaled Lichnerowicz equation \eqref{eqRescaledLichnerowicz} whenever $v_\epsilon > 0$. Note that we cannot apply directly the proof of Proposition \ref{propLichEq} since it presupposes taking the logarithm of $v_\epsilon$. However, remark that $\phitil_i \geq \gammatil_i^{\frac{1}{2\kappa+4}} \phi_-$ so we can restrict ourselves to the open subset $\Omega = \left\{ v_\epsilon > \gammatil_i^{\frac{1}{2\kappa+4}} \phi_- \right\}$. If $\Omega$ is empty then we are done, otherwise since $\Omega$ is relatively compact then the proof of Proposition \ref{propLichEq} applied on $\Omega$ works.
\end{proof}

\begin{rk}
 If the sequence is such that $\epsilon_i = 0$ for all $i$, then $\gamma_i^{-\frac{\epsilon_i}{2\kappa+4}} = 1$ so $\lambda = 1$.
\end{rk}

We immediately obtain the following corollary:

\begin{cor}\label{corFundamentalCorollary}
Let $\tau-1 \in X^{1, p}_\delta$, where $p \in (n; \infty)$ and $\delta \in (0; n)$, be such that the limit equation admits no non-zero solution in $W^{2,2}_0(M,T^*M)$. Then the constraint equations \eqref{eqLichnerowicz}-\eqref{eqVector} admit a solution. Furthermore, the set of pairs $(\phi,\psi) \in X^{2, p}_0(M, \bR) \times X^{2, p}_0(M, T^*M)$ solving the constraint equations \eqref{eqLichnerowicz}-\eqref{eqVector} is compact.
\end{cor}

\begin{proof}
Set $\epsilon_i = \frac{1}{i}$. The $\epsilon_i$-sub-critical constraint equations \eqref{eqSubCritConstraint} admit a solution $(\phi_i, \psi_i)$, where $\phi_i \in X^{2, p}_+(M, \bR)$ and $\psi_i \in X^{2, p}_\delta(M, T^* M)$. If the energy $\gammatil_i = \gammatil_\beta(\phi_i, \psi_i)$ is not bounded, then Proposition \ref{propUnboundedEnergy} asserts that there exists a non-zero solution $\psi \in X^{2, p}_\delta(M, T^*M)$ of the limit equation. However, we remark that any solution $\psi \in X^{2, p}_\delta(M, T^* M)$ of the limit equation \eqref{eqLimit} belongs to $W^{2, 2}_0(M, T^*M)$. Indeed, by the Sobolev embedding theorem, $\lie \psi \in C^{0, 0}_\delta$, hence $\left|\lie\psi\right|_g \frac{d\tau}{\tau} \in X^{0, p}_{2\delta}$. Let $\delta'$ be such that $\frac{n-1}{2} < \delta' < n$. By Proposition \ref{propIsomVectLaplacian} (isomorphism theorem for the vector Laplacian), we obtain that $\psi \in X^{2, p}_{\delta''}(M, T^*M)$, where $\delta'' = \min \left\{2\delta, \delta'\right\}$. Upon iterating this argument, we obtain that $\psi \in X^{2, p}_{\delta^{(k)}} (M, T^* M) \subset W^{2, 2}_0(M, T^*M)$ for some $\delta^{(k)} > \frac{n-1}{2}$ by Lemma \ref{lmEmbeddingXIntoL}. We get a contradiction with the assumption of the corollary. This proves that the energy $\gammatil_i$ is bounded, and we conclude by Proposition \ref{propExistenceSolution} that there exists a solution $(\phi_\infty, \psi_\infty)$ of the constraint equations.\\

We now prove that the set of solutions is compact. First remark that there exists a constant $C > 0$ such that for any pair $(\phi,\psi) \in X^{2, p}_0(M, \bR) \times X^{2, p}_0(M, T^*M)$ we have $\gammatil_\beta(\phi, \psi) \leq C$. Indeed, if such a constant does not exist, there exist a sequence $(\phi_i, \psi_i)$ such that $\gammatil_\beta(\phi_i, \psi_i) \to \infty$. In this case Proposition \ref{propUnboundedEnergy} asserts the existence of a non-zero solution of the limit equation \eqref{eqLimit}, which is a contradiction. Let now $(\phi_i,\psi_i) \in X^{2, p}_0(M, \bR) \times X^{2, p}_0(M, T^*M)$ be an arbitrary sequence of solutions of the constraint equations \eqref{eqLichnerowicz}-\eqref{eqVector}. Proposition \ref{propExistenceSolution} then gives that some subsequence of $(\phi_i,\psi_i)$ converges to $(\phi_\infty,\psi_\infty) \in X^{2, p}_+(M, \bR) \times X^{2, p}_\delta(M, T^*M)$, a solution of the constraint equations.
\end{proof}

\section{End of the proof of Theorem \ref{thmMainTheorem}}\label{secHigherRegularity}

In this section, we complete the proof of Theorem \ref{thmMainTheorem} by studying the regularity of the solutions of the constraint equations. We assume that the limit equation \eqref{eqLimit} does not admit a non-trivial solution. Higher regularity will be gained by iteration. We start by proving that the solutions of the Lichnerowicz equation enjoy better regularity:

\begin{prop}\label{propRegLichnerowicz}
Assume that $(M, g)$ is a $C^{l, \beta}$-asymptotically hyperbolic manifold with constant scalar curvature.
\begin{itemize}
	\item If $\tau - 1 \in X^{k-1, p}_\delta(M, \bR)$ and $\sigma \in X^{k-1, p}_\delta$ for some $2\leq k \leq l$, $p \in (n; \infty)$ and $\delta \in (0; n)$, then the solution $\phi$ of the Lichnerowicz equation \[-\frac{4(n-1)}{n-2} \Delta \phi - n(n-1) \phi + n(n-1) \tau^2 \phi^{\kappa+1} = \left|\sigma\right|^2_g \phi^{-\kappa-3}\] is such that $\phi - 1 \in X^{k, p}_\delta(M, \bR)$. Further, the map $X^{k-1, p}_\delta \ni \sigma \mapsto \phi-1 \in X^{k, p}_\delta$ is continuous.
	\item If $\tau - 1 \in C^{k-1, \alpha}_\delta(M, \bR)$ and $\sigma \in C^{k-1, \alpha}_\delta$ for some $k\geq 2$ such that $k + \alpha \leq l + \beta$ and $\delta \in (0; n)$, then the solution $\phi$ of the Lichnerowicz equation is such that $\phi - 1 \in C^{k, \alpha}_\delta(M, \bR)$. Further, the map $C^{k-1, \alpha}_\delta \ni \sigma \mapsto \phi-1 \in C^{k, \alpha}_\delta$ is continuous.
\end{itemize}
\end{prop}

\begin{proof}
We first give the proof for local Sobolev spaces. The first step is to prove that $\phi-1 \in X^{k, p}_\delta(M, \bR)$ (see also \cite[Theorem 3.3]{Gicquaud}). From a naive application of elliptic regularity in M\"obius charts, see e.g. \cite[Lemma 4.8]{LeeFredholm}, we get that $\phi \in X^{k, p}_0(M, \bR)$. We rewrite the Lichnerowicz equation \eqref{eqLichnerowicz} as follows:

\[ -\frac{4(n-1)}{n-2} \Delta \left(\phi - \phi_-\right) + n(n-1) \left[\tau^2 (\kappa+1) \int_0^1 \left(\lambda \phi + (1-\lambda) \phi_-\right)^\kappa d\lambda - 1\right] \left(\phi - \phi_-\right) = \left|\sigma\right|^2_g \phi^{-\kappa-3}\]
where we have subtracted the equation \eqref{eqPrescScal} satisfied by $\phi_-$ from the Lichnerowicz equation and used a standard integration trick. From our naive estimate, we have that the coefficient in brackets belongs to $X^{k-1, p}_0(M, \bR) \subset C^{k-2, 0}_0(M, \bR)$. Applying once more elliptic regularity in M\"obius charts, and using the $L^\infty$-estimate $\phi - 1 = O(\rho^\delta)$ (Proposition \ref{propSubCritBehav2}), we infer that $\phi-1 \in X^{k, p}_\delta(M, \bR)$.\\

The second step is to prove the continuous dependence of $\phi-1 \in X^{k,p}_\delta(M, \bR)$ with respect to $\sigma \in X^{k-1, p}_\delta$. The proof is based on the implicit function theorem. Set $u = \log \phi$. Then $u$ satisfies the following equation (see also Proposition \ref{propLichEq} and recall that $g$ has constant scalar curvature): \[-\frac{4(n-1)}{n-2}\left(\Delta u + \left|d u\right|_g^2\right) + n(n-1)\left(\tau^2 e^{\kappa u} - 1\right) - \left|\sigma\right|_g^2 e^{-(\kappa+4) u}= 0.\] From the first step, $u \in X^{k,p}_\delta(M, \bR)$. Remark that the map
\[
\begin{array}{rccc}
H : & X^{k, p}_\delta(M, \bR) & \to & X^{k-2, p}_\delta(M, \bR)\\
		& v & \mapsto &  -\frac{4(n-1)}{n-2}\left(\Delta v + \left|d v\right|_g^2\right) + n(n-1)\left[\tau^2 \left(e^{\kappa v} - 1\right) + \tau^2 - 1\right] - \left|\sigma\right|_g^2 e^{-(\kappa+4) v}
\end{array}
\]
is well defined and everywhere $C^1$. This comes from the following facts:
\begin{itemize}
	\item from the fact that $k \geq 2$, $p > n$, $d v \in X^{k-1, p}_\delta(M, T^*M)$ and that $\otimes : X^{k-1, p}_\delta(M, T^*M) \times X^{k-1, p}_\delta(M, T^*M) \to X^{k-1, p}_\delta(M, \left(T^*M\right)^{\otimes 2})$ is a continuous bilinear map, hence $v \mapsto |dv|_g^2$ is $C^1$;
	\item $X^{k, p}_\delta(M, \bR)$ is a (non unital) Banach algebra, so the map $u \mapsto e^{\kappa u} - 1 = \sum_{i=1}^\infty \frac{\kappa^i u^i}{i!}$ is analytic;
	\item similarly, $X^{k, p}_0(M, \bR) \supset X^{k, p}_\delta(M, \bR)$ is a unital Banach algebra, so $v \mapsto e^{-(\kappa+4) v}$ is analytic as a map from $X^{k, p}_0(M, \bR)$ to itself. Furthermore, multiplication by $\left|\sigma\right|_g^2 \in X^{k-2, p}_{2 \delta}(M, \bR)$ yields  a continuous linear mapping from $X^{k, p}_0(M, \bR) \subset W^{k-2, \infty}_0(M, \bR)$ to $X^{k-2, p}_\delta(M, \bR)$.
\end{itemize}

The differential $DH$ of $H$ at $u$ is given by

\[DH_u(v) = -\frac{4(n-1)}{n-2}\left(\Delta v + 2 \left\langle du, dv \right\rangle_g\right) + n(n-1)\kappa \tau^2 e^{\kappa u} v + (\kappa+4) \left|\sigma\right|_g^2 e^{-(\kappa+4) u} v.\]

We prove that $DH_u$ is an isomorphism. Indeed, let
\[
\begin{array}{rcccl}
P_0 : & X^{k, p}_\delta(M, \bR) & \to & X^{k-2, p}_\delta(M, \bR)\\
			& v & \mapsto &  -\frac{4(n-1)}{n-2} \Delta v + n(n-1) \kappa v & = \frac{4(n-1)}{n-2} \left(-\Delta v + n v\right),
\end{array}
\]
then  it is easy to show that $P_0$ satisfies the assumptions of Theorem \ref{thmFredholmLocalSobolev}. We compute as in the proof of Proposition \ref{propSubCritBehav}:
\begin{align*}
-\Delta \rho^\delta + n \rho^\delta
	& = \left[- \delta (\delta - n + 1) + n\right] \rho^\delta + O(\rho^{\delta+1})\\
	& = - (\delta+1)(\delta-n) \rho^\delta + O(\rho^{\delta+1}).
\end{align*}
Hence the indicial radius of $P_0$ (i.e. half the difference between the two roots of the equation $(\delta+1)(\delta-n) = 0$)  is $\frac{n+1}{2}$. Therefore, by Theorem \ref{thmFredholmLocalSobolev}, $P_0$ is a Fredholm operator with zero index for any $p \in (1, \infty)$, $2 \leq k \leq l$ and $\left|\delta - \frac{n-1}{2}\right| < \frac{n+1}{2}$, i.e. for any $\delta \in (-1; n)$. Let

\[
\begin{array}{rccc}
P_1 : & X^{k, p}_\delta(M, \bR) & \to & X^{k-2, p}_\delta(M, \bR)\\
			& v & \mapsto & - \frac{4(n-1)}{n-2}\left(2 \left\langle du, dv \right\rangle_g + n \left(e^{\kappa u}\tau^2-1\right) v\right) + (\kappa+4) \left|\sigma\right|_g^2 e^{-(\kappa+4) u} v.
\end{array}
\]

We claim that $P_1$ is a compact operator, thus $DH_u = P_0 + P_1$ is a Fredholm operator with zero index. Indeed, by Lemma \ref{lmBanachAlgebra}, $X^{k, p}_\delta \ni v \mapsto \langle du, dv\rangle_g \in X^{k, p-1}_{2\delta}$ is continuous. Then it follows by Proposition \ref{propRellich} that  the map $X^{k, p}_\delta \ni v \mapsto \langle du, dv\rangle_g \in X^{k, p-2}_\delta$ is compact. The proof is similar for the other two terms. From Lemma \ref{lmAPrioriX2p} we immediately get that the $X^{k, p}_\delta$-kernel of $DH_u$ is zero for any $\delta \in (0; n)$. This completes the proof of the fact that $DH_u$ is an isomorphism.\\

Applying the implicit functions theorem, we get that the map $\sigma \mapsto u$ is continuous (in fact analytic).
\end{proof}

We first prove Theorem \ref{thmMainTheorem} for the $X^{k, p}_\delta$-spaces. Let us denote 
\[S = \{(h, \psi) | (\phi = 1+h, \psi) \text{ solves \eqref{eqLichnerowicz}-\eqref{eqVector}}\} \subset X^{2, p}_+(M, \bR) \times X^{2, p}_0(M, T^*M).\] From Corollary \ref{corFundamentalCorollary}, $S$ is non-empty and compact. Let $L \subset X^{2, p}_+(M, \bR)$ and $V \subset X^{2, p}_0(M, T^*M)$ be the projections onto the first and second factor. Both $L$ and $V$ are compact. From Proposition \ref{propVectEq} combined with the fact the the image of a compact subset under a continuous map is compact, we infer that $V \subset X^{2, p}_\delta(M, T^*M)$ is compact. We apply next Proposition \ref{propRegLichnerowicz} and obtain that $L \subset X^{2, p}_\delta(M, \bR)$ is compact. This process can be continued inductively and we get that $V \subset X^{k, p}_\delta(M, T^*M)$ and $L \subset X^{k, p}_\delta(M, \bR)$ are compact. Hence $S \subset L \times V \subset X^{k, p}_\delta(M, \bR) \times X^{k, p}_\delta(M, T^*M)$ is precompact. Closedness of $S \subset X^{k, p}_\delta(M, \bR) \times X^{k, p}_\delta(M, T^*M)$ is a consequence of the fact that $S$ is the set of solutions of the constraint equations (hence the zero set of a certain continuous map $\cC: X^{k, p}_\delta(M, \bR) \times X^{k, p}_\delta(M, T^*M) \to X^{k-2, p}_\delta(M, \bR) \times X^{k-2, p}_\delta(M, T^*M)$).\\

We now turn our attention to the case of H\"older spaces. Let $p = \frac{n}{1-\alpha}$. The previous argument shows that the set $S$ is a compact subset of $X^{k, p}_\delta(M, \bR) \times X^{k, p}_\delta(M, T^*M)$. From Proposition \ref{propRellich} (Sobolev injection), we get that $S \subset\subset C^{k-1, \alpha}_\delta(M, \bR) \times C^{k-1, \alpha}_\delta(M, T^*M)$. Arguing as in the case of local Sobolev spaces, we conclude from Proposition \ref{propVectEq} and Proposition \ref{propRegLichnerowicz} that $S \subset C^{k, \alpha}_\delta(M, \bR) \times C^{k, \alpha}_\delta(M, T^*M)$ is compact.

\section{The limit system}\label{secLimit}

In this section we study the properties of the limit equation \eqref{eqLimit}: \[\Delta_L \psi = \lambda \sqrt{\frac{n-1}{n}} \left| \lie \psi \right|_g \frac{d \tau}{\tau},\] where $\lambda \in (0; 1]$, $\tau > 0$, $\tau - 1 \in X^{1, p}_\delta$ for some $p > n$ and $0<\delta<n$. In particular, we try to answer the following question: How large is the set of functions $\tau$ such that this equation admits no non-zero solution? The example constructed in \cite{DahlGicquaudHumbert} shows that this study may be complicated. We state here two propositions regarding the near-CMC case (Propositions \ref{propLnNearCMC} and \ref{propLinftyNearCMC}). In particular, the second one gives an explicit large upper bound for the $L^\infty$-norm of $\frac{d\tau}{\tau}$.

\begin{prop}[The $L^n$-near CMC case]\label{propLnNearCMC}
Let $C_g$ be the constant defined in Appendix \ref{secNoL2CKV} (Equation \ref{eqSobolevConstant}). If $\left\|\frac{d\tau}{\tau}\right\|_{L^n} < \frac{1}{2} \sqrt{\frac{n}{n-1}} C_g$, then the limit equation \eqref{eqLimit} does not admit any non-zero solution $\psi \in W^{1, 2}_0$.
\end{prop}

\begin{proof}
Indeed, assume that $\psi$ is a non-zero solution of the limit equation \eqref{eqLimit}. Then, multiplying \eqref{eqLimit} by $\psi$ and integrating by parts, we get:
\begin{align*}
- \frac{1}{2} \int_M \left|\lie \psi\right|_g^2 d\mu_g
	& =		 \sqrt{\frac{n-1}{n}} \int_M \left|\lie \psi\right|_g \left\langle \frac{d\tau}{\tau}, \psi\right\rangle_g d\mu_g,\\
\frac{1}{2} \int_M \left|\lie \psi\right|_g^2 d\mu_g
	& \leq \sqrt{\frac{n-1}{n}} \left(\int_M \left|\lie \psi\right|_g^2 d\mu_g\right)^{\frac{1}{2}} \left(\int_M \left|\frac{d\tau}{\tau}\right|_g^n d\mu_g\right)^{\frac{1}{n}}
				 \left(\int_M \left|\psi\right|_g^{2^*} d\mu_g\right)^{\frac{1}{2^*}},\\
\frac{1}{2} \left(\int_M \left|\lie \psi\right|_g^2 d\mu_g\right)^{\frac{1}{2}}
	& \leq \sqrt{\frac{n-1}{n}} \left(\int_M \left|\frac{d\tau}{\tau}\right|_g^n d\mu_g\right)^{\frac{1}{n}} \left(\int_M \left|\psi\right|_g^{2^*} d\mu_g\right)^{\frac{1}{2^*}}.\\
\end{align*}
Consequently,
\[\frac{1}{2} C_g \leq \frac{1}{2} \frac{\left(\int_M \left|\lie \psi\right|_g^2 d\mu_g\right)^{\frac{1}{2}}}{\left(\int_M \left|\psi\right|_g^{2^*} d\mu_g\right)^{\frac{1}{2^*}}} \leq \sqrt{\frac{n-1}{n}} \left(\int_M \left|\frac{d\tau}{\tau}\right|_g^n d\mu_g\right)^{\frac{1}{n}},\]
which contradicts the assumption of the proposition.
\end{proof}

\begin{prop}[The $L^\infty$-near CMC case]\label{propLinftyNearCMC}
Let $f : M \to \bR$ be a continuous function such that $\ric \leq -f g$. If $\left|\frac{d\tau}{\tau}\right|_g^2 < \frac{n}{n-1} f$, then the limit equation \eqref{eqLimit} admits no non-zero solution $\psi \in W^{1, 2}_0$ for any $\lambda \in (0; 1]$.
\end{prop}

\begin{rk}\label{rkEinstein}
 Assume that the manifold $(M, g)$ is Einstein, i.e. $\ric = -(n-1) g$, then Proposition \ref{propLinftyNearCMC} applies with $f = n-1$. In this case, we get that the limit equation has no solution if $\left|\frac{d\tau}{\tau}\right|^2_g < n$.
\end{rk}

\begin{proof}
We start from the Bochner-type formula (Equation \ref{eqBochnerVectLaplacian}) given in Appendix \ref{secNoL2CKV}:
\[\Delta_L \psi_i = \Delta \psi_i + \left(1 - \frac{2}{n}\right) \nabla_i \divg \psi + \ricud{j}{i} \psi_j.\]
Inserting this formula in the limit equation \eqref{eqLimit}, contracting with $\psi$ and integrating by parts, we get:

\begin{align*}
\int_M \left(\left|\nabla \psi\right|_g^2 + \left(1 - \frac{2}{n}\right) \left|\divg \psi\right|^2 - \ric(\psi, \psi)\right) d\mu_g
	&  =   \lambda \sqrt{\frac{n-1}{n}} \int_M \left|\lie \psi\right|_g \left\langle \frac{d\tau}{\tau}, \psi \right\rangle d\mu_g\\
	& \leq 2 \lambda \sqrt{\frac{n-1}{n}} \int_M \left|\nabla \psi\right|_g \left|\frac{d\tau}{\tau}\right|_g \left|\psi\right|_g d\mu_g
\end{align*}
At any point where $d\tau \neq 0$, we set $\epsilon^2 = \frac{1}{\lambda \sqrt{\frac{n-1}{n}}} \left|\frac{d\tau}{\tau}\right|_g^{-1}$ and use the standard inequality \[2 \left|\nabla \psi\right|_g \left|\psi\right|_g \leq \epsilon^2 \left|\nabla \psi\right|_g^2 + \epsilon^{-2} \left|\psi\right|_g^2\] and obtain that \[2 \lambda \sqrt{\frac{n-1}{n}} \left|\nabla \psi\right|_g \left|\frac{d\tau}{\tau}\right|_g \left|\psi\right|_g \leq \left|\nabla \psi\right|^2 + \lambda^2 \frac{n-1}{n} \left|\frac{d\tau}{\tau}\right|_g^2 \left|\psi\right|_g^2.\] Remark that this inequality trivially holds at points where $d\tau = 0$. Hence, \[\int_M \left(\left|\nabla \psi\right|_g^2 + \left(1 - \frac{2}{n}\right) \left|\divg \psi\right|^2 - \ric(\psi, \psi)\right) d\mu_g \leq \int_M \left(\left|\nabla \psi\right|_g^2 + \lambda^2 \frac{n-1}{n} \left|\frac{d\tau}{\tau}\right|_g^2 \left|\psi\right|_g^2\right) d\mu_g.\]
In particular, the following inequality holds: \[\int_M \left(-\ric(\psi, \psi) - \lambda^2 \frac{n-1}{n} \left|\frac{d\tau}{\tau}\right|_g^2 \left|\psi\right|_g^2\right) d\mu_g \leq 0.\]
Since we assumed that $\ric \leq -f g$, this implies that \[\int_M \left(f - \lambda^2 \frac{n-1}{n} \left|\frac{d\tau}{\tau}\right|_g^2\right) \left|\psi\right|_g^2 d\mu_g \leq 0.\]
It follows that $\psi = 0$.
\end{proof}

Remark that we can assume that $f \to n-1$ at infinity. Since $\frac{d\tau}{\tau}$ tends to zero at infinity, the assumption of the proposition is fulfilled outside a relatively compact open subset $\Omega$. In view of \cite[Theorem 4.1]{Lohkamp}, there exists metrics which coincide with $g$ outside $\Omega$, are arbitrarily close to $g$ for the $C^0$-norm in $\Omega$ and have arbitrarily low Ricci curvature inside $\Omega$. Hence, we easily get the following corollary:

\begin{cor}\label{corDensity}
Let $(M, g_0)$ be a $C^{l, \beta}$-asymptotically hyperbolic manifold with $l+\beta \geq 2$. Given $\tau \in C^1$, $\tau > 0$ such that $\left|\frac{d\tau}{\tau}\right|_g = O(\rho^\delta)$ for some $\delta > 0$, there exist metrics $g$ such that $g-g_0$ is arbitrarily $C^0$-small, compactly supported and such that the limit equation \eqref{eqLimit} (defined with respect to the metric $g$ and $\tau$) admits no non-zero solution.
\end{cor}

\appendix
\section{Fredholm theorem for the local Sobolev spaces}\label{secFredholm}

The aim of this section is to prove the analog of \cite[Theorem C]{LeeFredholm} for the weighted local Sobolev spaces:
\begin{theorem}[Fredholm Theorem]\label{thmFredholmLocalSobolev}
Let $(M,g)$ be a connected asymptotically hyperbolic $n$-manifold of class $C^{l,\beta}$, with $n\geq 1$, $l\geq 2$, and $0\leq \beta< 1$, and let $E\rightarrow M$ be a geometric tensor bundle over $M$. Suppose that $P:C^\infty(M,E)\rightarrow C^\infty(M,E)$ is an elliptic, formally self-adjoint, geometric partial differential operator of order $m$, $0<m\leq l$, and assume that there exists a compact set $K\subset M$ and a positive constant $C$ such that 
\begin{equation}\label{EstimateInfinity}
\|u\|_{L^2}\leq C\,\|Pu\|_{L^2} \text{ for all } u\in C^\infty_c(M\setminus K;E).
\end{equation}
Let $R$ be the indicial radius of $P$.  If $1<p<\infty$ and $m\leq k\leq l$ then the natural extension 
\begin{equation*}
P:X^{k,p}_{\delta}(M,E)\rightarrow X^{k-m,p}_\delta(M,E)
\end{equation*}
is Fredholm if and only if $|\delta-\frac{n-1}{2}|<R$. In that case, its index is zero, and its kernel is equal to the $L^2$ kernel of $P$.
\end{theorem}

The proof of this statement is very similar to that of \cite[Theorem C]{LeeFredholm}, which deals with the case of the weighted Sobolev and H\"older spaces, and follows essentially the same line.\\


First, note that arguing as in \cite[Lemma 4.6]{LeeFredholm} one easily shows that for $\delta\in \bR$, $1<p<\infty$, and $m\leq k\leq l$, $P$ extends naturally to a bounded mapping $P:X_\delta ^{k,p}(M,E)\rightarrow X_\delta^{k-m,p}(M,E)$.  Moreover, arguing as in the proof of \cite[Lemma 4.8]{LeeFredholm}, it is straightforward to prove the following consequence of elliptic regularity applied in M\"obius charts:

\begin{lemma}\label{lemmaEllipticRegularity1}
Let $P:C^\infty(M,E)\rightarrow C^\infty(M,E)$ be a geometric elliptic operator of order $m$. If $u\in X_{\delta}^{0,p}(M,E)$ for $\delta\in\bR$, $1<p<\infty$, $m\leq k\leq l$, and $Pu\in X_\delta^{k-m,p}(M,E)$, then $u\in X_\delta^{k,p}(M,E)$ and
\begin{equation*}
\|u\|_{X^{k,p}_\delta}\leq C\left(\|Pu\|_{X^{k-m,p}_\delta}+\|u\|_{X_\delta^{0,p}}\right).
\end{equation*}
\end{lemma}


The next step is to study the model case, i.e. the situation when $P$ is an operator on hyperbolic space satisfying the estimate \eqref{EstimateInfinity}. For this we will use the Poincar\'e ball model, identifying hyperbolic space with the unit ball $\bB\in \bR^n$ with coordinates $(\xi_1,\dots,\xi_n)$, and with the hyperbolic metric $\gbrev=4(1-|\xi|)^{-2}\sum_i(d\xi_i)^2$. The Green kernel of $P$ will be denoted by $K$. We refer to  \cite[Chapter 5]{LeeFredholm} for estimates of $K$ and other facts about $K$.

We set
\begin{equation*}
P^{-1}f(\xi)=\int_{\bB}K(\xi,\eta)f(\eta)d\mu_{\gbrev}(\eta)
\end{equation*}
and prove the following estimate:

\begin{prop}\label{propEstimateInverseLocalSobolev}
If $1<p<\infty$, and $|\delta-\frac{n-1}{2}|<R$, then there exists a constant $C>0$ such that 
\begin{equation*}
\|P^{-1}f\|_{X_\delta^{0,p}}\leq C\|f\|_{X_\delta^{0,p}}
\end{equation*}
for all $f\in X_{\delta}^{0,p}(\bB,E)$. 
\end{prop}

\begin{proof}
Let us choose a countable uniformly locally finite covering of $\bB$ by M\"obius charts as in \cite[Lemma 2.2]{LeeFredholm}, and note that in this particular situation $B_i=(\Phi_{\xi_i})^{-1}(B_1)$ are geodesic balls centered at $\xi_i$ of radius 1. Below we will denote by $\lambda B_i$ the geodesic ball centered at $\xi_i$ of radius $\lambda$.

We will use the decomposition
\begin{equation*}
\int_{\bB}K(\xi,\eta)f(\eta)d\mu_{\breve{g}}(\eta)=\int_{4B_i}K(\xi,\eta)f(\eta)d\mu_{\breve{g}}(\eta)+\int_{\bB\setminus 4B_i}K(\xi,\eta)f(\eta)d\mu_{\breve{g}}(\eta),
\end{equation*}
to estimate
\begin{align*}
\rho_i^{-\delta p}\int_{B_i}|P^{-1}f|^p_{\breve{g}}d\mu_{\breve{g}}(\xi) & = \rho_i^{-\delta p}\int_{B_i}\left|\int_{\bB}K(\xi,\eta)f(\eta)d\mu_{\breve{g}}(\eta)\right|^p_{\breve{g}}d\mu_{\breve{g}}(\xi)\\
& \leq 2^p\rho_i^{-\delta p}\int_{B_i}\left| \int_{4B_i}K(\xi,\eta)f(\eta)d\mu_{\breve{g}}(\eta) \right|^p_{\gbrev} d\mu_{\gbrev}(\xi)\\ & + 2^p\rho_i^{-\delta p} \int_{B_i}\left| \int_{\bB\setminus 4B_i}K(\xi,\eta)f(\eta)d\mu_{\breve{g}}(\eta) \right|^p_{\gbrev} d\mu_{\gbrev}(\xi).
\end{align*}

First, recall that $K_0=K(0,\cdot)$ is in $L^1_{loc}$ \cite[proof of Lemma 5.5]{LeeFredholm}. Hence, if $\xi\in B_i$ is arbitrary, and $\phi$ is any M\"obius transformation sending $\xi$ to $0$, then the change of variables $\eta'=\phi(\eta)$ yields
\begin{equation*}
\int_{4B_i}|K(\xi,\eta)|d\mu_{\breve{g}}(\eta)=\int_{\phi(4B_i)}|K(0,\eta')|d\mu_{\breve{g}}(\eta')\leq C,
\end{equation*}
where $C>0$ does not depend on $i$ and $\xi$. Consequently, we have a uniform estimate
\begin{equation*}
\sup_{\xi\in B_i} \int_{4B_i}|K(\xi,\eta)|d\mu_{\gbrev}(\eta)+\sup_{\eta\in 4B_i}\int_{B_i}|K(\xi,\eta)|d\mu_{\gbrev}(\xi)\leq C 
\end{equation*}
and it follows by Young inequality \cite[Theorem 0.3.1]{Sogge} that there exists $C$ such that for all $i$
\begin{equation*}
\left\|\int_{4B_i}K(\xi,\eta)f(\eta)d\mu_{\breve{g}}(\eta)\right\|_{L^p(B_i)}\leq C\|f\|_{L^p(4B_i)}.
\end{equation*}

We use this inequality to estimate the first integral. Let $I_i=\{j:B_j\cap 4B_i\neq\emptyset\}$. Note that the number of elements in $I_i$ is uniformly bounded, and that if $j\in I_i$ then $C^{-1}\rho_i\leq \rho_j\leq C\rho_i$ for $C>0$ sufficiently large, but independent of $i$ and $j$. We have
\begin{align*}
\rho_i^{-p \delta}\int_{B_i}\left| \int_{4B_i}K(\xi,\eta)f(\eta)d\mu_{\breve{g}}(\eta) \right|^p_{\gbrev} d\mu_{\gbrev}(\xi) & \leq C\rho_i^{-p\delta}\|f\|^p_{L^p(4B_i)} \\ & \leq C \rho_i^{-p \delta} \sum _{j\in I_i}\|f\|^p_{L^p(B_j)} \\ & \leq C\sum_{j\in I_i}\rho_j^{-p\delta}\|f\|^p_{L^p(B_j)} \\ & \leq C\|f\|^p_{X_\delta^{0,p}}.
\end{align*}

We turn to the second integral. It is obvious that
\begin{equation*}
\int_{\bB\setminus 4B_i}|K(\xi,\eta)||f(\eta)|_{\gbrev}d\mu_{\gbrev}(\eta)\leq \sum_{j\in J_i}\int_{B_j}\left|K(\xi,\eta)\right|\left|f(\eta)\right|_{\gbrev}d\mu_{\gbrev}(\eta),
\end{equation*}
where $J_i=\{j\,:\, (\bB \setminus 4B_i)\cap B_j \neq \emptyset\}$. Moreover, if $j\in J_i$ then $\xi_j\in \bB\setminus 3B_i$, hence $B_j\subset \bB\setminus 2B_i$, and thus $d_{\gbrev}(\xi,\eta)\geq 1$ for $\xi \in B_i$ and $\eta \in B_j$. Therefore we can apply \cite[Proposition 5.2]{LeeFredholm}, which says that for every $\epsilon>0$ there is a constant $C$ such that $|K(\xi,\eta)|\leq C\rho(\xi,\eta)^{\frac{n-1}{2}+R-\epsilon}$, where $\rho(\xi,\eta)=(\cosh d_{\gbrev}(\xi,\eta))^{-1}$, provided that  $d_{\gbrev}(\xi,\eta)\geq1$. Remark also that $C^{-1}\rho(\xi,\xi_j)\leq\rho(\xi,\eta)\leq C \rho(\xi,\xi_j)$ for $\xi \in B_i$ and $\eta \in B_j$, where $j\in J_i$. Let us now choose $\epsilon>0$ such that $\frac{n-1}{2}-R+\epsilon<\delta<\frac{n-1}{2}+R-\epsilon$ and apply H\"older inequality:
\begin{align*}
\int_{B_j}|K(\xi,\eta)||f(\eta)|_{\gbrev}d\mu_{\gbrev}(\eta) & \leq \|K(\xi,\cdot)\|_{L^{p'}(B_j)}\|f\|_{L^p(B_j)}\\ & \leq \|K(\xi,\cdot)\|_{L^{p'}(B_j)}\rho_j^{\delta}\|f\|_{X_{\delta}^{0,p}}\\ & \leq C\rho(\xi,\xi_j)^{\frac{n-1}{2}+R-\epsilon}\rho_j^{\delta}\|f\|_{X_{\delta}^{0,p}},
\end{align*}
where $\frac{1}{p}+\frac{1}{p'}=1$. 
Summing over  $j \in J_i$ and applying \cite[Lemma 5.4]{LeeFredholm}, we obtain
\begin{align*}
\sum_j \int_{B_j}\left|K(\xi,\eta)\right|\left|f(\eta)\right|_{\gbrev}d\mu_{\gbrev}(\eta)
&  \leq C\sum_j \rho(\xi,\xi_j)^{\frac{n-1}{2}+R-\epsilon}\rho_j^{\delta}\|f\|_{X_{\delta}^{0,p}}\\
& \leq C\int_{\bB\setminus 2B_i} \rho(\xi,\eta)^{\frac{n-1}{2}+R-\epsilon}\rho(\eta)^{\delta}d\mu_{\gbrev}(\eta)\|f\|_{X_{\delta}^{0,p}}\\
& \leq C\rho (\xi)^{\delta}\|f\|_{X_{\delta}^{0,p}}.
\end{align*}
Finally, 
\begin{align*}
\rho_i^{-p\delta}\int_{B_i}\left|\int_{\bB \setminus 4B_i}K(\xi,\eta)f(\eta)d\mu_{\gbrev}(\eta)\right|_{\gbrev}^p d\mu_{\gbrev}(\xi)
& \leq C\rho_i^{-p\delta}\int_{B_i}\rho(\xi)^{p\delta}d\mu_{\gbrev}(\xi)\,\|f\|^p_{X_{\delta}^{0,p}} \\ & \leq C\|f\|^p_{X_{\delta}^{0,p}}.
\end{align*}

Consequently,
\begin{equation*}
\|P^{-1}f\|^p_{X_\delta^{0,p}}=\sup_i \rho_i^{-\delta p}\int_{B_i}|P^{-1}f|^p_{\breve{g}}d\mu_{\breve{g}}(\xi) \leq C\|f\|^p_{X^{0,p}_\delta},
\end{equation*}
and the statement is proved.
\end{proof}

\begin{theorem}\label{theoremIsomHypSpace}
Let $P$ be a formally self-adjoint geometric elliptic operator satisfying \eqref{EstimateInfinity}. If $k\geq m$, $1<p<\infty$, $|\delta-\frac{n-1}{2}|<R$ then $P:X_\delta^{k,p}(\bB,E)\rightarrow X_\delta^{k-m,p}(\bB,E)$ is an isomorphism.
\end{theorem}

\begin{proof}
We argue as in \cite[Theorem 5.9]{LeeFredholm}. 

We first prove injectivity. Choose $\delta'$ such that
\begin{equation*}
-R+\frac{n-1}{2}-\frac{n-1}{p}<\delta'< \delta-\frac{n-1}{p}.
\end{equation*} 
Then, since $\delta'+\frac{n-1}{p}< \delta$, we have $X_\delta^{k,p}\subset W_{\delta'}^{k,p}$. Moreover, our choice of $\delta'$ implies $|\delta'+\frac{n-1}{p}-\frac{n-1}{2}|<R$. Since $P$ is injective on $W_{\delta'}^{k,p}$ for such $\delta'$ by \cite[Theorem 5.7]{LeeFredholm}, it is also injective on the  smaller space  $X_\delta^{k,p}$.

Now let $f\in X_{\delta}^{k,p}$ be arbitrary and set $u=P^{-1}f$. Then $u\in X_\delta^{0,p}$ by Proposition \ref{propEstimateInverseLocalSobolev}. We will show that $Pu=f$ holds in the distributional sense. From Proposition \ref{propEstimateInverseLocalSobolev} it is obvious that $\int_{\bB} |K(\xi,\eta)||f(\eta)|_{\gbrev}d\mu_{\gbrev}(\eta)\in W^{0,1}_{loc}$ hence for any $\phi\in C^\infty_0$ we have
\begin{equation*}
\int_{\bB}\int_{\bB}|\left\langle K(\xi,\eta)f(\eta),(P\phi)(\xi)\right\rangle_{\gbrev}|d\mu_{\gbrev}(\eta)d\mu_{\gbrev}(\xi)<\infty.
\end{equation*}
Consequently, one can apply Fubini theorem and the fact that $K$ satisfies the symmetry condition $K(\eta,\xi)=K(\xi,\eta)^*$ to compute
\begin{align*} 
\int_{\bB}\left\langle\int_{\bB}K(\xi,\eta)f(\eta)d\mu_{\gbrev}(\eta),(P\phi)(\xi)\right\rangle_{\gbrev}d\mu_{\gbrev}(\xi) & = 
\int_{\bB}\int_{\bB}\left\langle K(\xi,\eta)f(\eta),(P\phi)(\xi)\right\rangle_{\gbrev} d\mu_{\gbrev}(\eta)d\mu_{\gbrev}(\xi)\\ & = 
\int_{\bB}\int_{\bB}\left\langle f(\eta),K(\eta,\xi)(P\phi)(\xi) \right\rangle_{\gbrev}d\mu_{\gbrev}(\xi)d\mu_{\gbrev}(\eta)\\ & = 
\int_{\bB}\left\langle f(\eta),\int_{\bB}K(\eta,\xi)(P\phi)(\xi) d\mu_{\gbrev}(\xi)\right\rangle_{\gbrev}d\mu_{\gbrev}(\eta)\\ & = 
\int_{\bB}\left\langle f(\eta),\int_{\bB}(\delta_\xi(\eta)\id_{E_\xi})(\phi)d\mu_{\gbrev}(\xi)\right\rangle_{\gbrev}d\mu_{\gbrev}(\eta)\\ & = 
\int_{\bB}\left\langle f(\eta),\phi(\eta)\right\rangle_{\gbrev} d\mu_{\gbrev}(\eta).
\end{align*}
Hence $u\in X_\delta ^{k,p}$ by Lemma \ref{lemmaEllipticRegularity1}.

The continuity of the inverse follows from Lemma \ref{lemmaEllipticRegularity1} and Proposition \ref{propEstimateInverseLocalSobolev}.
\end{proof}


Following \cite{LeeFredholm} we use the resulting inverse map for $P:X_\delta^{k,p}(\bB,E)\rightarrow X_\delta^{k-m,p}(\bB,E)$ to construct a parametrix in the general case $P:X_\delta^{k,p}(M,E)\rightarrow X_\delta^{k-m,p}(M,E)$. For this we first recall how the boundary M\"obius coordinates $(\rho,\thetatil)$, which are defined all the way up to $\partial M$, are constructed. Recall the notations $M_\mu=\rho^{-1}(0;\mu)$, and $K_\mu=M\setminus M_\mu$, and that $(\bH, \gbrev)$ denotes the upper half-space model of the hyperbolic space. 

For sufficiently small $c$ each $\phat\in \partial M$ has a neighborhood $\Omega (\phat)$ on which the background coordinates $(\rho,\theta)$ are defined on the set $\{(\rho,\theta): 0\leq\rho<c,\, |\theta-\theta(\phat)|<c\}$. Recall that $|d\rho|_{\gbar}^2=1$ along $\partial M$ and choose a $C^{l,\beta}$ orthonormal frame $\overline{w} = (\overline{w}^0,\overline{w}^1,\ldots,\overline{w}^{n-1})$ on $\Omega(\phat)$ with respect to the metric $\gbar$, such that $\overline{w}^0 = d\rho$ on $\partial M \cap \Omega(\phat)$. We use the coefficients of the expansion 
\begin{equation*}
\overline{w}^{\alpha}_{\phat}=A_{\beta}^\alpha d\theta_{\phat}^\beta+B^\alpha d\rho_{\phat}
\end{equation*}
to define the boundary M\"obius coordinates $(\rho,\thetatil)$ by setting
\begin{equation*}
\thetatil^\alpha=A_{\beta}^\alpha \theta^\beta+B^\alpha \rho.
\end{equation*}
It is obvious that in these new coordinates at $\phat$ we have $\gbar_{ij}=\delta_{ij}$. We proceed by defining for $0<r<c$ and $a>0$ open subsets
\begin{align*}
Y_a &= \{(y,x)\in \bH:0<y<a,\,|x|<a\}\subset\bH,\\
Z_r(\phat) &= \{(\rho,\thetatil)\in\Omega(\phat):0<\rho<r,\,|\thetatil|<r\}\subset\Omega(\phat)\subset M,
\end{align*}
and a boundary M\"obius chart $\Psi_{\phat,r}:Z_r(\phat)\rightarrow Y_1$ by 
\begin{equation*}
(y,x)=\Psi_{\phat,r}(\rho,\thetatil)=\left(\frac{\rho}{r},\frac{\thetatil}{r}\right).
\end{equation*}

An important property of the boundary M\"obius coordinates to be used in the parametrix construction is stated in \cite[Lemma 6.1]{LeeFredholm}. Namely, there exists $C>0$ such that for any $\phat\in \partial M$ and sufficiently small $r>0$ we have
\begin{equation}\label{ClosenessMetric}
\|(\Psi^{-1}_{\phat,r})^*g-\gbrev\|_{C^{l,\beta}(Y_1)}\leq rC.
\end{equation} 

Although $y$ cannot be used as a defining function for $\bH$ since it blows up at infinity, one can construct using a partition of unity a smooth defining function $\rho'$ for $\bH$ such that $\rho'=y$ on $Y_1$. Then for $0<r<c$ we have $(\Psi_{\phat,r}^{-1})^*\rho=y=\rho'$, and it is therefore straightforward to check that for the weighted local Sobolev norms it is true that
\begin{equation}\label{ScalingBehavior}
C^{-1}r^{-\delta}\|(\Psi_{\phat,r}^{-1})^*u\|_{X_{\delta}^{k,p}(Y_1)}\leq\|u\|_{X_{\delta}^{k,p}(Z_r(\phat))}\leq Cr^{-\delta}\|(\Psi_{\phat,r}^{-1})^*u\|_{X_{\delta}^{k,p}(Y_1)}.
\end{equation}
Note also the similar scaling behaviour of H\"older and Sobolev norms:
\begin{align}
C^{-1}r^{-\delta}\|(\Psi_{\phat,r}^{-1})^*u\|_{C_{\delta}^{k,\alpha}(Y_1)}\leq\|u\|_{C_{\delta}^{k,\alpha}(Z_r(\phat))}\leq Cr^{-\delta}\|(\Psi_{\phat,r}^{-1})^*u\|_{C_{\delta}^{k,\alpha}(Y_1)},\label{ScalingBehaviorHolder}\\
C^{-1}r^{-\delta}\|(\Psi_{\phat,r}^{-1})^*u\|_{W_{\delta}^{k,p}(Y_1)}\leq\|u\|_{W_{\delta}^{k,p}(Z_r(\phat))}\leq Cr^{-\delta}\|(\Psi_{\phat,r}^{-1})^*u\|_{W_{\delta}^{k,p}(Y_1)}.\label{ScalingBehaviorSobolev}
\end{align}

Now choose a specific smooth bump function $\psi:\bH\rightarrow [0,1]$ supported on $Y_1$ and equal to 1 on $Y_{1/2}$, and define the functions 
\begin{equation*}
\psi_{\phat,r}(\rho,\thetatil)=\Psi_{\phat,r}^*\psi=\psi\left(\frac{\rho}{r},\frac{\thetatil}{r}\right)
\end{equation*} 
supported on $Z_r(\phat)$ and equal to 1 on $Z_{r{/2}}(\phat)$.
Note  that the functions $\psi_{\phat,r}$ are uniformly bounded in $C^{l,\beta}(M)$ by some constant independent of $\phat$ and $r$. Moreover, arguing as in \cite[Lemma 2.2]{LeeFredholm} one can show that there exists a number $N$ such that for any $r>0$ we can choose finitely many points $\{\phat_1,\dots\phat_m\}\in \partial M$ such that $M_{r/2}$ is covered by the sets $Z_{r/2}(\phat_i)$ and at most $N$ of the sets $Z_r(\phat_i)$ intersect nontrivially at any point. Let $\psi_{0,r}$ denote a smooth bump function equal to 1 on $K_{r/2}$ and supported in $K_{r/4}$, and introduce the notation $\Psi_{i,r}=\Psi_{\phat_i,r}$ and $\psi_{i,r}=\psi_{\phat_i,r}$. Then the functions 
\begin{equation*}
\phi_{i,r}=\frac{\psi_{i,r}}{\left(\sum_{j=0}^m\psi_{j,r}^2\right)^{\frac{1}{2}}}
\end{equation*}
constitute a partition of unity for $\overline{M}$ subordinate to the cover $\{K_{r/4},Z_r(\phat_i)\}$. It is also easy to see that these functions are uniformly bounded in $C^{l,\beta}_{(0)}(\overline{M})$.

We will now construct a bundle $\Ebrev \to \bH$ which is in a certain sense similar to $E$ together with an isomorphism $\Upsilon_{i, r}: \Ebrev|_{Y_1} \to E|_{Z_r(\phat)}$. Let $T^{r_1}_{r_2} = \left(\bR^n\right)^{r_2} \otimes \left(\bR^{n*}\right)^{r_1}$ be the standard tensor space. It comes with a natural $O(n)$-action and is such that $T^{r_1}_{r_2} M = F \times_{O(n)} T^{r_1}_{r_2}$, where $F$ is the orthonormal frame bundle of $TM$ with respect to the metric $g$. Since $E$ is a geometric tensor bundle, there exists an invariant subspace $\vec{E}$ such that $E = F \times_{O(n)} \vec{E}$. Let $\breve{F}$ (resp. $F_{i,r}$) be the (oriented) orthonormal frame bundle of the hyperbolic metric $\gbrev$ on $\bH$ (resp. $g_{i, r} = \left(\Psi_{i, r}^{-1}\right)^* g$ on $Y_1$). Define $\Ebrev = \breve{F} \times_{O(n)} \vec{E}$. Let $\wbrev = (\frac{dy}{y}, \frac{dx^1}{y}, \ldots, \frac{dx^{n-1}}{y})$ be the standard orthonormal frame on the upper half-space model $(\bH,\gbrev)$ of the hyperbolic space and let $w=(w_0,w_1,\dots,w_{n-1})$, where  $w^0 = \frac{1}{\rho} \overline{w}^0, \ldots, w^{n-1} = \frac{1}{\rho} \overline{w}^{n-1}$, be the orthonormal coframe for the metric $g$ associated to $\overline{w}$. We finally define $w_{i, r} = \left(\Psi_{i, r}^{-1}\right)^* w$ and $\gbar_{i, r} = \left(\Psi_{i, r}^{-1}\right)^* \gbar$. By the definition of $\thetatil$, we have $w_{i, r} = \wbrev$ at $0 = \Psi_{i, r}(\phat)$. Then there exists a unique equivariant map $\Xi_{i,r}: \breve{F}|_{Y_1} \to F_{i, r}$ such that $\Xi_{i, r}(\breve{w}) = w_{i, r}$. Moreover, $\Xi_{i,r}$ descends to a bundle map which will be denoted by the same name: \[\Xi_{i, r}: \Ebrev|_{Y_1} \to \left(\Psi_{i, r}^{-1}\right)^* \left(E|_{Z_r(\phat)}\right).\] From the fact that $w_{i, r}(0) = \wbrev(0)$, we get
\begin{equation}\label{ClosednessVielbein}
\left\| \Xi_{i, r} - Id\right\|_{C^{l, \beta}_0(Y_1, End(T^{r_1}_{r_2} Y_1))} \leq C r,
\end{equation}
where $\Xi_{i, r}$ denotes the bundle endomorphism of $T^{r_1}_{r_2} Y_1$. We finally set \[\Upsilon_{i, r} = \Xi_{i, r} \circ \left(\Psi_{i, r}^{-1}\right)^*: \Ebrev\vert_{Y_1} \to E\vert_{Z_r(\phat)}.\] 

Let $\Pbrev$ be the operator on hyperbolic space with the same local coordinate expression as $P$. For each $i$ consider the operator $P_{i,r}:C^\infty(Y_1,\Ebrev)\rightarrow C^\infty(Y_1,\Ebrev)$ defined by 
\begin{equation*}
P_{i,r}u= \Upsilon_{i,r}^{-1} P(\Upsilon_{i,r}u).
\end{equation*}
Since $P$ is a geometric operator, by \eqref{ClosenessMetric} and \eqref{ClosednessVielbein}, we conclude that for every $u\in X_\delta ^{k,p}(Y_1,\Ebrev)$,
\begin{equation}\label{ClosenessOperator}
\|P_{i,r}u-\Pbrev u\|_{X_\delta^{k-m,p}}\leq Cr \|u\|_{X_\delta^{k,p}}. 
\end{equation} 

Now suppose that $P$ satisfies \eqref{EstimateInfinity}. Using \eqref{ScalingBehaviorSobolev} it is easy to show that $P_{i,r}$ also satisfies \eqref{EstimateInfinity}.  Consequently, if $r$ is small enough it follows by \eqref{ClosenessOperator} and \cite[Lemma 4.8 (a)]{LeeFredholm} that $\Pbrev$ satisfies an analogous estimate (perhaps with a different constant) for all $u\in C_c^\infty(Y_1,\Ebrev)$. The same estimate holds globally on $\bH$ since for an arbitrary  $u\in C_c^\infty(\bH,\Ebrev)$ there is a M\"obius transformation taking $\supp u$ into $Y_1$. By Theorem \ref{theoremIsomHypSpace} we conclude that the operator $\Pbrev$ is invertible on $X_\delta ^{k,p}(\bH,\Ebrev)$ provided that $|\delta-\frac{n-1}{2}|<R$.

For any sufficiently small $r>0$ we define operators $Q_r,\,S_r,\,T_r:C_c^\infty(M,E)\rightarrow C_c^\infty(M,E)$ by
\begin{align*}
Q_r (u) &=\sum_i\phi_{i,r}\Upsilon_{i,r} \Pbrev^{-1} \Upsilon_{i,r}^{-1} (\phi_{i,r}u),\\
S_r (u) &=\sum_i\phi_{i,r}\Upsilon_{i,r} \Pbrev^{-1}(P_{i,r}-\Pbrev) \Upsilon_{i,r}^{-1} (\phi_{i,r}u),\\
T_r (u) &=\sum_i\phi_{i,r}\Upsilon_{i,r} \Pbrev^{-1} \Upsilon_{i,r}^{-1} ([\phi_{i,r},P]u).
\end{align*}
\begin{prop}
Let $P:C^\infty(M,E)\rightarrow C^\infty(M,E)$ satisfy \eqref{EstimateInfinity}. If $|\delta-\frac{n-1}{2}|<R$ and $1<p<\infty$ then $Q_r$, $S_r$, and $T_r$ extend to bounded maps as follows:
\begin{align*}
Q_r &: X_\delta^{0,p}(M,E)\rightarrow X_\delta^{m,p}(M,E),\\
S_r &: X_\delta^{m,p}(M,E)\rightarrow X_\delta^{m,p}(M,E),\\
T_r &: X_\delta^{m-1,p}(M,E)\rightarrow X_{\delta_1}^{m,p}(M,E),
\end{align*}
for any $\delta _1$ such that $\delta\leq\delta_1\leq\delta+1$ and $|\delta_1-\frac{n-1}{2}|<R$. Moreover, there exists $r_0>0$ such that if $u\in X_\delta^{m,p}(M,E)$ is supported in $M_r$ for $0<r<r_0$ then 
\begin{equation}\label{ParametrixPrelim}
Q_rPu=u+S_ru+T_ru
\end{equation}
and 
\begin{equation}\label{SrBounded}
\|S_ru\|_{X_\delta^{m,p}}\leq Cr\|u\|_{X_\delta^{m,p}}.
\end{equation}
\end{prop}
\begin{proof}
The proof is identical to that of \cite[Proposition 6.2]{LeeFredholm}. In particular, the fact that \eqref{ParametrixPrelim} holds in $M_r$ is verified by a straightforward computation. Further, recall that the functions $\phi_{i,r}$ are uniformly bounded in $C^{l,\beta}_{(0)}(\overline{M})\subset C^{l,\beta}(M)$. Then it is easy to check that multiplication by $\phi_{i,r}$ is a bounded map from $X_\delta^{j,p}(Z_r(\phat_i),E)$ into itself for each $i$ and all $0\leq j\leq l$, with norm bounded uniformly in $i$ and $r$. Using this fact, \eqref{ScalingBehavior}, \eqref{ClosednessVielbein}, and \eqref{ClosenessOperator} we can estimate ($C$ might vary from line to line):
\begin{align*}
\|S_r u\|&=\sum_i\|\phi_{i,r}\Upsilon_{i,r} \Pbrev^{-1}(P_{i,r}-\Pbrev) \Upsilon_{i,r}^{-1} (\phi_{i,r}u)\|_{X_\delta^{m,p}}\\
       &\leq C \sum_i\|\Upsilon_{i,r} \Pbrev^{-1}(P_{i,r}-\Pbrev) \Upsilon_{i,r}^{-1} (\phi_{i,r}u)\|_{X_\delta^{m,p}}\\
        &\leq C r^{\delta}\sum_i\|\Pbrev^{-1}(P_{i,r}-\Pbrev) \Upsilon_{i,r}^{-1} (\phi_{i,r}u)\|_{X_\delta^{m,p}}\\ 
        &\leq C r^{\delta}\sum_i\|(P_{i,r}-\Pbrev) \Upsilon_{i,r}^{-1} (\phi_{i,r}u)\|_{X_\delta^{m,p}}\\ 
        &\leq C r^{\delta+1}\sum_i\|\Upsilon_{i,r}^{-1} (\phi_{i,r}u)\|_{X_\delta^{m,p}}\\ 
        &\leq C r\sum_i\|\phi_{i,r}u\|_{X_\delta^{m,p}}\\ 
        &\leq C r\|u\|_{X_\delta^{m,p}}.
\end{align*}
since $\Pbrev$ is invertible and the cover $\{K_{r/4},Z_r(\phat_i)\}$ is locally finite.

The mapping properties of $T_r$ are easily verified by a similar argument, once it is shown that the commutator $[\phi_{i,r},P]$ maps $X_{\delta}^{m-1,p}$ to $X_{\delta_1}^{0,p}$. To prove the later statement we note that $[\phi_i,P]$ is a partial differential operator of order $m-1$, each term in its coordinate expression being the product of a polynomial in $g$, $(\det g)^{-1/2}$, and their derivatives up to order $q$, an $s$-th covariant derivative of $u$, and a $t$-th covariant derivative of $\phi_{i,r}$, where $q+s+t\leq m$, and $t\geq 1$. Recall that $\phi_{i,r}$ are uniformly bounded in $C^{l,\beta}_{(0)}(\overline{M})$ hence $t$-th covariant derivatives of $\phi_{i,r}$ are uniformly bounded in $C^{0,\beta}_1(M,E)$ by \cite[Lemma 3.7]{LeeFredholm}. Consequently, $[\phi_i,P]u \in X_{\delta_1}^{0,p}(M,E)$ for any $u\in X_\delta^{m-1,p}(M,E)$.
\end{proof}

\begin{cor}\label{corParametrix}
Let $P:C^\infty(M,E)\rightarrow C^\infty(M,E)$ satisfy \eqref{EstimateInfinity}. If $|\delta-\frac{n-1}{2}|<R$, $|\delta_1-\frac{n-1}{2}|<R$, $\delta\leq\delta_1\leq\delta+1$, and $1< p< \infty$, then there exists $r>0$ and bounded operators 
\begin{align*}
\Qtil &: X_\delta^{0,p}(M,E)\rightarrow X_\delta^{m,p}(M,E),\\
\Ttil &: X_\delta^{m-1,p}(M,E)\rightarrow X_{\delta_1}^{m,p}(M,E)
\end{align*}
such that 
\begin{equation}\label{parametrix}
\Qtil Pu=u+\Ttil u
\end{equation}
if $u\in X_\delta ^{k,p}(M;E)$ is supported in $M_r.$
\end{cor}
\begin{proof}
Choose $r$ so that $Cr<\frac{1}{2}$ in \eqref{SrBounded}. Then $(\id+S_r)^{-1}$ is a bounded operator.  That 
\begin{align*}
\Qtil &= (\id+S_r)^{-1}\circ Q_r,\\
\Ttil &= (\id+S_r)^{-1}\circ T_r
\end{align*}
satisfy \eqref{parametrix} is a consequence of \eqref{ParametrixPrelim}.
\end{proof}


Using the above parametrix construction we can improve the elliptic regularity result of Lemma \ref{lemmaEllipticRegularity1} (cf. \cite[Lemma 6.4 and Proposition 6.5]{LeeFredholm}).

\begin{lemma}\label{lemmaEllipticRegularity2}
Assume that $P:C^\infty(M,E)\rightarrow C^\infty(M,E)$ satisfies \eqref{EstimateInfinity}. If for $1< p< \infty$ and $m\leq k \leq l$ we have $u\in X_\delta ^{0,p}(M,E)$, where $|\delta-\frac{n-1}{2}|<R$, and $Pu\in X_{\delta'}^{k-m,p}(M,E)$, where $|\delta'-\frac{n-1}{2}|<R$, then $u\in X_{\delta'}^{k,p}(M,E)$.
\end{lemma}
\begin{proof}
If $\delta'\leq \delta$ then $u\in X_{\delta'}^{k,p}(M,E)$ by Proposition \ref{propRellich} and Lemma \ref{lemmaEllipticRegularity1}. Assume therefore $\delta'>\delta$. We will show that $u\in X_{\delta'}^{0,p}$. We choose $r>0$ as in Corollary \ref{corParametrix} and a smooth bump function $\psi$ supported on $M_r$ and equal to 1 on $M_{r/2}$. Then we can represent $u$ as $u=u_0+u_\infty$ where $u_0=(1-\psi)u$ is supported in $K_{r/2}$ and $u_\infty=\psi u$ is supported on $M_r$. That $u_0\in X_{\delta'}^{k,p}$ is a consequence of local elliptic regularity. As for $u_\infty$, we note that $u_\infty=u$ outside a compact set, hence $Pu_\infty\in X_{\delta'}^{k-m,p}$. By Corollary \ref{corParametrix} we conclude that $u_\infty=\Qtil Pu_\infty-\Ttil u_\infty\in X_{\delta_2}^{0,p}$ where $\delta_2=\min\{\delta',\delta_1\}$, $\delta\leq\delta_1\leq\delta+1$, $|\delta_1-\frac{n-1}{2}|<R$. Iterating this argument finitely many times we conclude that $u_\infty \in X_{\delta'}^{0,p}$, and the statement follows by Lemma \ref{lemmaEllipticRegularity1}.
\end{proof}

\begin{prop}\label{propEllipticRegularity2}
Suppose that $P:C^\infty(M,E)\rightarrow C^\infty(M,E)$ satisfies \eqref{EstimateInfinity}. Assume that $u$ is either in $X_{\delta_0}^{0,p_0}(M,E)$ or in $C_{\delta_0}^{0,0}$ for $|\delta_0-\frac{n-1}{2}|<R$ and $1<p_0<\infty$. If $Pu\in X_\delta^{k-m,p}(M,E)$ for $|\delta-\frac{n-1}{2}|<R$ and $1<p<\infty$ then $u\in X_\delta^{k,p}(M,E)$. 
\end{prop}
\begin{proof}
Since $C_{\delta_0}^{0,0}\subset X_{\delta_0}^{0,p_0}$ for any $1<p_0<\infty$ it suffices to prove the statement under the assumption $u\in X_{\delta_0}^{0,p_0}$. Suppose that $Pu\in X_\delta^{k-m,p}$ and consider the set
\begin{equation*}
\mathcal{P}=\left\{p'\in (1,\infty):\,u\in X_{\delta'}^{0,p'} \text{ for some } \delta' \text{ such that } |\delta'-(n-1)/2|<R\right\}.
\end{equation*}
It is obvious that $p_0\in\mathcal{P}$. It is also clear that if $p_1\in\mathcal{P}$ then $(1,p_1]\subset\mathcal{P}$ since $X_{\delta'}^{0,p_1}\subset X_{\delta'}^{0,p'}$ for any $p'<p_1$. Moreover, if $p_1\in\mathcal{P}$ and $p_2$ is such that $p_1<p_2\leq \min\left(p,\frac{n+1}{n}p_1\right)$ then $p_2\in \mathcal{P}$. Indeed, in this case $u\in X_{\delta_1}^{0,p_1}$ where $|\delta_1-\frac{n-1}{2}|<R$, and $Pu\in X_\delta^{k-m,p}\subset X_\delta^{k-m,p_1}$ with $|\delta-\frac{n-1}{2}|<R$, hence $u\in X_\delta^{k,p_1}$ by Lemma \ref{lemmaEllipticRegularity2}. Note also that 
\begin{equation*}
\frac{n}{p_1}\leq\frac{n+1}{p_2}\leq\frac{n}{p_2}+k,
\end{equation*}
thus $X_\delta ^{k,p_1}\subset X_\delta^{0,p_2}$ by Proposition \ref{propRellich}. Consequently $p_2\in \mathcal{P}$, and after finitely many iterations we conclude that $u\in X_{\delta}^{0,p}$. Finally, Lemma \ref{lemmaEllipticRegularity2} (or Lemma \ref{lemmaEllipticRegularity1}) yields $u\in X_\delta^{k,p}$.
\end{proof}

Now assume that $P:C^\infty(M,E)\rightarrow C^\infty(M,E)$ satisfies the hypotheses of Theorem \ref{thmFredholmLocalSobolev}, in particular that  \eqref{EstimateInfinity} holds. We let $Z=\Ker P\cap L^2$ and note that $Z=\Ker P\cap W^{m,2}$ by \cite[Proposition 6.5]{LeeFredholm}. By \cite[Lemma 4.10]{LeeFredholm} $P:W^{m,2}(M,E)\rightarrow W^{0,2}(M,E)$ is Fredholm, hence $Z$ is finite dimensional. 

An important observation is that the kernel of the natural extension $P:X_{\delta}^{k,p}(M,E)\rightarrow X_\delta ^{k-m,p}(M,E)$, where $m\leq k\leq l$, $1<p<\infty$, and $|\delta-\frac{n-1}{2}|<R$, is equal to $Z$. Indeed, on the one hand, if $u\in Z$ then, by \cite[Proposition 6.5]{LeeFredholm}, $u\in C^{k,\alpha}_\delta\subset X_\delta ^{k,p}$. On the other hand, if $u\in X_{\delta}^{k,p}$ is such that $Pu=0$, then it follows by Proposition \ref{propEllipticRegularity2} that $u\in X_{\delta'}^{k,q}$ for any $q\in(1,\infty)$ and $\delta'$ such that $|\delta'-\frac{n-1}{2}|<R$. Let us choose $\delta'$ and $q$ so that $0<\delta'-\frac{n-1}{2}<R$ and $q\geq 2$, then $u\in L^2$ by Lemma \ref{lmEmbeddingXIntoL}.

If we choose $\delta_0$ so that $\delta>\delta_0+\frac{n-1}{p}>\frac{n-1}{2}-R$ then  $X_\delta ^{k,p}\subset W_{\delta_0}^{0,p}$ by Lemma \ref{lmEmbeddingXIntoL}. Note also that $|-\delta_0+\frac{n-1}{p'}-\frac{n-1}{2}|=|\delta_0+\frac{n-1}{p}-\frac{n-1}{2}|<R$ where $\frac{1}{p}+\frac{1}{p'}=1$. Consequently, it follows from the discussion preceeding \cite[Theorem 6.6]{LeeFredholm} that
\begin{equation*}
Z\subset W_{-\delta_0}^{0,p'}\subset( W^{0,p}_{\delta_0})^*\subset (X_\delta^{k,p})^*.
\end{equation*} 
We conclude that 
\begin{equation*}
Y_\delta^{k,p}=\{u\in X_\delta^{k,p}\,:\, (u,v)=0 \text{ for all $v\in Z$}\}
\end{equation*}
is a well-defined closed subspace.

\begin{theorem}[Structure Theorem for elliptic operators acting on local Sobolev spaces]\label{StructureTheorem}
Suppose that $P:C^\infty(M,E)\rightarrow C^\infty(M,E)$ satisfies the hypotheses of Theorem \ref{thmFredholmLocalSobolev}. If $1<p<\infty$, $0\leq k \leq l$, and $|\delta-\frac{n-1}{2}|<R$, then there exist bounded operators $G,H:X_\delta^{k,p}(M,E)\rightarrow X_\delta^{k,p}(M,E)$ such that $G(X_\delta^{k-m,p}(M,E))\subset X_\delta^{k,p}(M,E)$ for $k\geq m$, and 
\begin{align}
Y_\delta^{k,p} &= \Ker H, \label{Ker} \\
Z &= \im H, \label{Im} \\
u &= GPu+Hu \text{ for } u\in X_\delta ^{k,p}(M,E),\, m\leq k\leq l, \label{structure1} \\
u &= PGu+Hu \text{ for } u\in X_\delta ^{k,p}(M,E),\, 0\leq k\leq l. \label{structure2}
\end{align}
\end{theorem}
\begin{proof}

Let $\delta'$ be such that $-R+\frac{n-1}{2}<\delta'+\frac{n-1}{p}<\delta$. Then $X_{\delta}^{k,p}\subset W_{\delta'}^{k,p}$ and $P:W_{\delta'}^{k,p}(M,E)\rightarrow W_{\delta'}^{k-m,p}(M,E)$ satisfies the hypotheses of \cite[Theorem 6.6 (a)]{LeeFredholm} which is a Structure Theorem for elliptic operators acting on Sobolev spaces. Recall that $Z=\Ker P\cap W_{\delta'}^{k,p}=\Ker P \cap X_\delta^{k,p}$ (see \cite[p. 53]{LeeFredholm}), and that $H$ constructed in the proof of \cite[Theorem 6.6 (a)]{LeeFredholm} is just a projection on $Z$. Since $Z$ is a subset of $X_\delta^{k,p}$ it is obvious that the restriction of $H$ to $X_\delta^{k,p}$ maps $X_\delta^{k,p}$ to itself.

The operator $G$ can also be defined as the restriction of $G:W_{\delta'}^{k,p}(M,E)\rightarrow W_{\delta'}^{k,p}(M,E)$ constructed in \cite[Theorem 6.6 (a)]{LeeFredholm} on $X_\delta^{k,p}$. However, the proof of the fact that $G$ maps $X_\delta^{k,p}$ to itself requires some work. The main complication is that in general we do not have $|\delta'-\frac{n-1}{2}|<R$, otherwise the result would directly follow from Proposition \ref{propEllipticRegularity2}.

To begin with, let us consider the easier case when $p>n$. It is readily seen that if $u\in X_\delta^{k,p}$ then $PGu=u-Hu\in X_\delta^{k,p}$. Since $p>n$, by Proposition \ref{propRellich} we have $X_\delta^{k,p}\subset C_\delta^{k-1,\alpha}$, thus $Gu\in C^{k,\alpha}_\delta\subset C_\delta^{0,0}$ by  \cite[Theorem 6.6 (b)]{LeeFredholm}. Then $Gu\in X_\delta^{k,p}$ by Proposition \ref{propEllipticRegularity2}.

Now suppose that $p<n$. For simplicity we introduce the notation $v=Gu$, where $v\in W_{\delta'}^{k,p}$. As in the proof of Lemma \ref{lemmaEllipticRegularity2}, write $v=v_0+v_\infty$. Here both $v_0$ and $v_\infty$ are in $W_{\delta'}^{k,p}$, $v_\infty$ agrees with $v$ outside of the compact set, and $v_0$ has compact support. Since it is obvious that $v_0\in X_\delta^{k,p}$, in what follows we focus on $v_\infty$. 

First assume that $2R\leq\frac{n-1}{n}$. By \cite[Corollary 6.3 (a)]{LeeFredholm} and Corollary \ref{corParametrix} we have 
\begin{equation*}
v_\infty=\Qtil P v_\infty-\Ttil v_\infty,
\end{equation*}
where $\Qtil P v_\infty \in X_\delta^{0,p}$ and $\Ttil v_\infty \in W_{\delta_1}^{1,p}$ for any $\delta_1$ such that $\delta'\leq \delta_1\leq \delta'+1$ and $|\delta_1+\frac{n-1}{p}-\frac{n-1}{2}|<R$. Assume that initially $\delta'$ was chosen to be sufficiently small. Since $2R<1$, we see that the above requirements are satisfied by
\begin{equation}\label{delta1}
\delta_1=R-\rho-\frac{n-1}{p}+\frac{n-1}{2}
\end{equation}
for some $\rho$ such that $R-\rho=1/N$ for sufficiently big $N>0$. Moreover, it follows by Proposition \ref{propRellich} that $\Ttil v_\infty\in W_{\delta_1}^{0,p_1}$ for any $p_1$ such that $p\leq p_1 \leq \frac{np}{n-p}$. Since $\delta_1+\frac{(n-1)(n-p)}{pn}-\frac{n-1}{2}<-R$, it is obvious that $p_1$ such that $\delta_1+\frac{n-1}{p_1}-\frac{n-1}{2}=-R+\rho$
satisfies both the above condition and the condition $|\delta_1+\frac{n-1}{p_1}-\frac{n-1}{2}|<R$. A straightforward computation shows that 
\begin{equation}\label{p1}
\frac{n-1}{p_1}=-2(R-\rho)+\frac{n-1}{p}.
\end{equation}

We have shown that $\Ttil v_\infty \in W_{\delta_1}^{0,p_1}$, where $\delta_1$ and $p_1$, defined by \eqref{delta1} and \eqref{p1} respectively, satisfy $|\delta_1+\frac{n-1}{p_1}-\frac{n-1}{2}|<R$. Again, by \cite[Corollary 6.3 (a)]{LeeFredholm} we can write 
\begin{equation*}
(\Qtil P-\id)^2v_\infty=\Ttil(\Ttil v_\infty),
\end{equation*}
where $(\Qtil P-\id)^2 v_\infty-v_\infty \in X_\delta ^{0,p}$. Arguing as above, we see that $\Ttil(\Ttil v_\infty) \in W_{\delta_2}^{1,p_1}$, where 
\begin{equation*}
\delta_2=R-\rho-\frac{n-1}{p_1}+\frac{n-1}{2}=3(R-\rho)-\frac{n-1}{p}+\frac{n-1}{2}
\end{equation*}
satisfies $\delta_1<\delta_2<\delta_1+1$. Next, we define $p_2>p_1$ by  
\begin{equation*}
\frac{n-1}{p_2}=-R+\rho-\delta_2+\frac{n-1}{2}=-4(R-\rho)+\frac{n-1}{p}.
\end{equation*}
Note that in the case $p_1<n$ we have $p_2<\frac{np_1}{n-p_1}$. Hence $\Ttil(\Ttil v_\infty) \in W_{\delta_2}^{0,p_2}$ by Proposition \ref{propRellich} and $|\delta_2+\frac{n-1}{p_2}-\frac{n-1}{2}|<R$.

Proceeding in a similar fashion, for $i=1,\dots,l$ we inductively construct an increasing sequence of $\delta_i$, defined by 
\begin{equation*}
\delta_i=R-\rho-\frac{n-1}{p_{i-1}}+\frac{n-1}{2}=(2i-1)(R-\rho)-\frac{n-1}{p}+\frac{n-1}{2}, 
\end{equation*}
and an increasing sequence of $p_i$, defined by
\begin{equation*}
\frac{n-1}{p_i}=-R+\rho-\delta_i+\frac{n-1}{2}=-2i(R-\rho)+\frac{n-1}{p}, 
\end{equation*}
such that $|\delta_i+\frac{n-1}{p_i}-\frac{n-1}{2}|<R$, and 
\begin{equation}\label{bootstrap1}
(\Qtil P-\id)^i v_\infty=\Ttil^i v_\infty\in W_{\delta_i}^{0,p_i}
\end{equation}
for $i=1,\dots,l$. Note also that
\begin{equation}\label{bootstrap2}
(\Qtil P-\id)^i v_\infty+(-1)^{i+1}v_\infty\in X_\delta ^{0,p}.
\end{equation}
Suppose that $l$ is such that $1-NR+\frac{N(n-1)}{p}<2l<\frac{N(n-1)}{p}$. In particular, this implies $-2i(R-\rho)+\frac{n-1}{p}>0$ for $i=1,\dots,l$, hence the above sequences are well-defined. Moreover, 
\begin{equation*}
-R<\frac{2l-1}{N}-\frac{n-1}{p}=\delta_l-\frac{n-1}{2}<-1/N<R,
\end{equation*}
hence $|\delta_l-\frac{n-1}{2}|<R$. Consequently, it follows by \eqref{bootstrap1} and \eqref{bootstrap2} that $v_\infty\in X_{\delta}^{0,p}+X_{\delta_l}^{0,p_l}\subset X_{\delta_0}^{0,p}$ for some $|\delta_0-\frac{n-1}{2}|<R$. Obviously, the same is true for $Gu=v$. Applying Proposition \ref{propEllipticRegularity2} one completes the proof for the case $2R\leq\frac{n-1}{n}$. 

Now assume that $2R>\frac{n-1}{n}$. In this case, there exists $\delta_1$ such that $\delta'<\delta_1\leq \delta'+1$ and 
\begin{equation*}
-R-\frac{(n-1)(n-p)}{np}+\frac{n-1}{2}<\delta_1<R-\frac{n-1}{p}+\frac{n-1}{2}.
\end{equation*}
Then $\Ttil v_\infty \in W_{\delta_1}^{1,p}\subset W_{\delta_1}^{0,p_1}$, where $p_1=\frac{pn}{n-p}$. It is also easy to see that $|\delta_1+\frac{n-1}{p_1}-\frac{n-1}{2}|<R$. We proceed by induction. Suppose that we have constructed $p_i$ and $\delta_i$ such that $\Ttil^i v_\infty\in W_{\delta_i}^{0,p_i}$ with $|\delta_i+\frac{n-1}{p_i}-\frac{n-1}{2}|<R$. If $p_i<n$ then we can choose $\delta_{i+1}$ so that $\delta_i<\delta_{i+1}\leq \delta_i+1$ and 
\begin{equation*}
-R-\frac{(n-1)(n-p_i)}{np_i}+\frac{n-1}{2}<\delta_{i+1}<R-\frac{n-1}{p_i}+\frac{n-1}{2},
\end{equation*} 
which is obviously possible, and set $p_{i+1}=\frac{np_i}{n-p_i}$. Assume that $p_i$ defined by $p_i=\frac{np_{i-1}}{n-p_{i-1}}$ satisfies $p_i<n$ for any $i$. Then the sequence $\{p_i\}$ converges, but the formula $p_i=p_{i-1}\left(1+\frac{p_{i-1}}{n-p_{i-1}}\right)$ implies that the limit must be $0$, which is a contradiction. Consequently, there exists $l$ such that $p_l\geq n$. Let us choose $\delta_l<\delta_{l+1}\leq \delta_l+1$ so that 
\begin{equation*}
-R-\frac{(n-1)(n-p_l)}{np_l}+\frac{n-1}{2}<\delta_{l+1}<R-\frac{n-1}{p_l}+\frac{n-1}{2}.
\end{equation*} 
Then $\Ttil^l v_\infty\in W_{\delta_{l+1}}^{1,p_l}$, where
\begin{equation*}
\delta_{l+1}>-R-\frac{(n-1)(n-p_l)}{np_l}+\frac{n-1}{2}=-R+\frac{n-1}{n}-\frac{n-1}{p_l}+\frac{n-1}{2}\geq-R+\frac{n-1}{2}.
\end{equation*}
We conclude that $\Ttil^l v_\infty \in X_{\delta_{l+1}}^{0,p_l}$ where $|\delta_{l+1}-\frac{n-1}{2}|<R$. The rest of the proof is the same as in the case $2R\leq \frac{n-1}{n}$.

Finally, note that $G$ maps $X_\delta^{k,p}$ into itself in the case $p=n$ as well, which is a consequence of the already established result for $p<n$ and Proposition \ref{propEllipticRegularity2}. Remark also that \eqref{Ker}--\eqref{structure2} are automatically satisfied by restriction, as a corollary of \cite[Theorem 6.6]{LeeFredholm}.
\end{proof}

We are finally able to prove the main result of this section, Theorem \ref{thmFredholmLocalSobolev}. 

\begin{proof}[Proof of Theorem \ref{thmFredholmLocalSobolev}] 
We first prove that $P:X_\delta^{k,p}(M,E)\rightarrow X_\delta^{k-m,p}(M,E)$ is Fredholm provided that $|\delta-\frac{n-1}{2}|<R$. The proof is based on the construction of Theorem \ref{StructureTheorem} and is identical to the respective part in the proof of Theorem C in \cite{LeeFredholm}.

It was already noted in the discussion preceeding Theorem \ref{StructureTheorem} that the kernel of $P: X_\delta^{k,p}(M,E)\rightarrow X_\delta^{k-m,p}(M,E)$ is finite dimensional, since it coincides with the kernel $Z$ of the Fredholm operator $P:W^{m,2}(M,E)\rightarrow W^{0,2}(M,E)$. We have also shown that $Z$ is the same as the $L^2$ kernel of $P$.

Suppose that $f=Pu\in P(X_\delta^{k,p}(M,E))$. Then $(f,v)=(Pu,v)=(u,Pv)=0$ for any $v\in Z$, which means that $f\in Y_\delta^{k-m,p}$. On the other hand, by \eqref{Ker} and \eqref{structure2} any $f\in Y_\delta^{k-m,p}$ can be written as $f=PGf+Hf=PGf$ and is therefore in  $P(X_\delta^{k,p}(M,E))$. We see that the range of $P$ is equal to $Y_\delta ^{k-m,p}$, which is closed. 

Finally, we can represent any $f\in X_\delta^{k-m,p}(M,E)$ as $f=PGf+Hf$, where $PGf\in P(X_\delta^{k,p}(M,E))=Y_\delta^{k-m,p}$ and $Hf\in Z$. Therefore $X_\delta^{k-m,p}(M,E)=Y_\delta^{k-m,p}\oplus Z$, and 
\begin{equation*}
\frac{X_\delta^{k-m,p}(M,E)}{P(X_\delta^{k,p}(M,E))}=\frac{Y_\delta^{k-m,p}\oplus Z}{Y_\delta^{k-m,p}}\cong Z,
\end{equation*} 
which shows both that the cokernel of $P$ is finite dimensional and that the index of $P$ is zero.

It still remains to show that $P:X_\delta^{k,p}(M,E)\rightarrow X_\delta^{k-m,p}(M,E)$ is not Fredholm for all $\delta$ which do not satisfy $|\delta-\frac{n-1}{2}|<R$. Let us first consider the case $\delta\leq -R+\frac{n-1}{2}$. It was shown in the proof of Theorem C in \cite{LeeFredholm} that for such $\delta$ the kernel of $P:C^{k,\alpha}_\delta(M,E)\rightarrow C^{k-m,\alpha}_\delta (M,E)$ is infinite dimensional. Since $C_\delta^{k,\alpha}(M,E)\subset X_\delta^{k,p}(M,E)$, the same is true for $P:X^{k,p}_\delta(M,E)\rightarrow X^{k-m,p}_\delta(M,E)$. We conclude that in this case $P$ is not Fredholm.

Next consider the situation when $\delta>R+\frac{n-1}{2}$. Choose $\delta'$ such that $R+\frac{n-1}{2}<\delta'+\frac{n-1}{p}<\delta$. Then $X_\delta^{k-m,p}(M,E)\subset W_{\delta'}^{0,p}(M,E)$, hence $W^{0,p'}_{-\delta'}(M,E)= (W_{\delta'}^{0,p}(M,E))^*\subset (X^{k-m,p}_\delta(M,E))^*$. Moreover, it is easy to check that $-\delta'<\frac{n-1}{2}-\frac{n-1}{p'}-R$, where $\frac{1}{p'}+\frac{1}{p}=1$, which implies that $P^*=P:W_{-\delta'}^{m,p'}(M,E)\rightarrow W_{-\delta'}^{0,p'}(M,E)$ has infinite dimensional kernel (see the proof of Theorem C in \cite{LeeFredholm}). It is obvious that each element $v\in \Ker P \cap W_{-\delta'}^{m,p'}$ denotes a continuous functional on $X_\delta ^{k-m,p}(M,E)$ by $u\mapsto (u,v)$, and each such functional annihilates $P(X_\delta^{k,p}(M,E))$. We conclude that if $\delta>R+\frac{n-1}{2}$ then  $P:X_\delta^{k,p}(M,E)\rightarrow X_\delta ^{k-m,p}(M,E)$ has infinite dimensional cokernel, thus it is not Fredholm.   


Finally, assume that $\delta=R+\frac{n-1}{2}$. In this case the image of $P:C^{k,\alpha}_{\delta}(M,E)\rightarrow C^{k-m,\alpha}_\delta(M,E)$ is not closed, see the proof of Theorem C in \cite{LeeFredholm}. Consider a sequence of $u_n\in C^{k,\alpha}_\delta(M,E)\subset X_\delta^{k,p}(M,E)$ such that $Pu_n\rightarrow f\in C_\delta ^{k-m,\alpha}$ as $n\rightarrow \infty$ and the equation $Pu=f$ does not admit a solution $u\in C_\delta^{k,\alpha}$. Suppose that $Pu=f$ for some $u\in X_\delta^{k,p}$ and choose $\delta'$ such that $-R+\frac{n-1}{2}<\delta'+\frac{n-1}{p}<\delta=R+\frac{n-1}{2}$. Then $u\in W_{\delta'} ^{k,p}$ where $|\delta'+\frac{n-1}{p}-\frac{n-1}{2}|<R$. Consequently, $u\in C^{k,\alpha}_\delta$ by \cite[Proposition 6.5]{LeeFredholm}, which is a contradiction. Hence $Pu=f$ does not admit a solution $u\in X_\delta ^{k-m,p}$ either. We conclude that if $\delta=R+\frac{n-1}{2}$ then the image of $P:X_\delta^{k,p}(M,E)\rightarrow X_\delta^{k-m,p}(M,E)$ is not closed, hence it is not Fredholm.
\end{proof}

\section{Non-existence of $L^2$-conformal Killing vectors and an isomorphism theorem for the vector Laplacian}\label{secNoL2CKV}
In this appendix, we prove Proposition \ref{propIsomVectLaplacian}. Remark that $\Delta_L$ is a geometric operator, hence according to \cite[Theorem C]{LeeFredholm} and Theorem \ref{thmFredholmLocalSobolev}, we only have to prove the $L^2$-estimate at infinity and that the $L^2$-kernel of $\Delta_L$ is $\{0\}$. The proof of these facts is taken from \cite[Lemma 6.2 and 6.7]{Gicquaud}. The computation of the indicial radius $R = \frac{n+1}{2}$ can be found in \cite[Lemma 6.1]{Gicquaud} or \cite[Proposition G]{LeeFredholm}.

\begin{itemize}
\item \textbf{$L^2$-estimate at infinity:} Since the sectional curvature of $(M, g)$ tends to $-1$ at infinity, there exists a compact subset $K \subset M$ such that $\ric_g \leq - \frac{n-1}{2} g$ on $M \setminus K$. If $\psi \in C^2$ is any 1-form compactly supported in $M \setminus K$, then
\begin{align}
\Delta_L \psi_k & = g^{ij}\left( \nabla_i \nabla_j \psi_k + \nabla_i \nabla_k \psi_j - \frac{2}{n} g_{jk} \nabla_i \nabla^l \psi_l \right)\nonumber\\
		   & = g^{ij} \nabla_i \nabla_j \psi_k + g^{ij}\left( \nabla_k \nabla_i \psi_j - \riemuddd{l}{j}{i}{k} \psi_l\right)
		 	  - \frac{2}{n} \nabla_k \nabla^l\psi_l \nonumber\\
		   & = g^{ij} \nabla_i \nabla_j \psi_k + \ricud{l}{k} \psi_l + \left(1-\frac{2}{n}\right) \nabla_k \left(\nabla^l \psi_l\right) \nonumber\\
\Delta_L \psi_k & = \Delta \psi_k + \left(1 - \frac{2}{n}\right) \nabla_k (\divg \psi) - \ricud{l}{k} \psi_l \label{eqBochnerVectLaplacian}
\end{align}

Hence

\begin{align*}
\left\| \Delta_L \psi \right\|_{L^2} \left\| \psi \right\|_{L^2}
	& \geq - \int_M \psi^k \Delta_L \psi_k\\
	& \geq - \int_M \psi^k \left(g^{ij} \nabla_i \nabla_j \psi_k + \ricud{l}{k} \psi_l + \left(1-\frac{2}{n}\right) \nabla_k \left(\nabla^l \psi_l\right)\right)\\
	& \geq \int_M \left[\left(\nabla_i \psi_j \right) \left(\nabla^i \psi^j \right)+ \left(1-\frac{2}{n}\right) \left(\nabla^k\psi_k\right)^2 - \ricuu{k}{l}\psi_k \psi_l \right]\\
	& \geq \frac{n-1}{2} \left\| \psi \right\|_{L^2}^2,	
\end{align*}
which gives the required estimate:
\[\left\| \Delta_L \psi \right\|_{L^2} \geq \frac{n-1}{2} \left\| \psi \right\|_{L^2}.\]

\item \textbf{$\ker_{L^2} \left(\Delta_L\right) = \{0\}$:} First remark that $\psi \in W^{2, 2}_0$. A similar integration-by-parts argument leads to the observation that if $\psi \in \ker_{L^2} \left(\Delta_L\right)$, then $\psi$ is the ($g$-)dual of a conformal Killing vector field $X$ of $g$. But $X$ is also a conformal Killing vector field for $\gbar$ and a simple argument using elliptic regularity \cite[Proposition 6.5]{LeeFredholm} together with \cite[Lemma 3.6 and 3.7]{LeeFredholm} enables us to deduce that $X \in C^0(\Mbar)$. However, since $X \in L^2(M, g)$ we conclude that $X = 0$ on $\partial M$. The operator $Y \mapsto \overline{\lie} Y$ is (overdetermined) elliptic in the sense of \cite{ADN}, thus $X \in W^{2, p}(\Mbar)$ for any $p > n$. Now let $p \in \partial M$ and select any chart $(\rho, x^1, \ldots, x^{n-1})$ in the neighborhood of $p$. Since $W^{2, p}(\Mbar) \into C^1(\Mbar)$, we see that $X \in C^1$ and $\overline{\lie} X = 0$ up to the boundary. Decomposing the conformal Killing equation into the normal and tangential components, we see that all the partial derivatives of $X$ on $\partial M$ vanish. Further, enlarge $\Mbar$ by allowing $\rho$ to take negative values (so this defines an exterior region of $\Mbar$ in a neighborhood of $p$), extend $\gbar$ to a $C^2$-metric, and define $X$ by zero on this extended region. This gives a vector field $X \in W^{2, p}$ in the ``two-sided" neighborhood of $p$ which vanishes on the exterior side. Finally, it follows by the low-regularity unique continuation result of D. Maxwell \cite[Lemma 7]{Maxwell} that $X$ has to vanish everywhere.
\end{itemize}

As a corollary of the non-existence of $L^2$-conformal Killing vector field, we can also prove that the constant $C_g$ defined as

\begin{equation}\label{eqSobolevConstant}
C_g = \inf_{\psi \in W^{1, 2}_0, \psi \neq 0} \frac{\left(\int_M \left|\lie \psi\right|_g^2 d\mu_g\right)^{\frac{1}{2}}}{\left(\int_M \left|\psi\right|_g^{2^*}d\mu_g\right)^{\frac{1}{2^*}}}
\end{equation}
is strictly positive on an asymptotically hyperbolic manifold. This constant appears in Proposition \ref{propLnNearCMC}.

\begin{prop}\label{propPositivityCg}
Let $(M, g)$ be a $C^{l, \beta}$-asymptotically hyperbolic manifold with $l + \beta \geq 2$. Then the constant $C_g$ is positive.
\end{prop}

\begin{proof}
From the Sobolev embedding (Proposition \ref{propRellich}), we know that there exists a constant $\mu > 0$ such that for all $\psi \in W^{1, 2}_0$, $\mu \left\|\psi\right\|_{L^{2^*}_0} \leq \left\|\psi\right\|_{W^{1,2}_0}$, hence, \[C_g = \inf_{\psi \in W^{1, 2}_0, \psi \neq 0} \frac{\left(\int_M \left|\lie \psi\right|_g^2 d\mu_g\right)^{\frac{1}{2}}}{\left\|\psi\right\|_{W^{1,2}_0}} \frac{\left\|\psi\right\|_{W^{1, 2}_0}} {\left\|\psi\right\|_{L^{2^*}_0}} \geq \mu \inf_{\psi \in W^{1, 2}_0, \psi \neq 0} \frac{\left(\int_M \left|\lie \psi\right|_g^2 d\mu_g\right)^{\frac{1}{2}}}{\left\|\psi\right\|_{W^{1,2}_0}}.\] Thus we only have to prove that there exists a constant $\alpha > 0$ such that \[\int_M \left|\lie \psi\right|_g^2 d\mu_g \geq \alpha \left\|\psi\right\|_{W^{1,2}_0}^2 = \alpha \int_M \left( \left|\nabla \psi\right|_g^2 + \left|\psi\right|_g^2\right) d\mu_g\] for any $\psi \in W^{1, 2}_0$. Note that we can restrict our attention to $C^2$ compactly supported vector fields $\psi$. For such a vector field, we have that

\begin{align*}
\int_M \left|\lie \psi\right|_g^2 d\mu_g
	& = -2 \int_M \left\langle \psi, \Delta_L \psi\right\rangle_g\\
	& = -2 \int_M \left\langle \psi, \Delta \psi + \left(1-\frac{2}{n}\right)\nabla \divg \psi\right\rangle_g d\mu_g -2 \int_M \ric(\psi, \psi) d\mu_g\\
	& = 2 \int_M \left|\nabla\psi\right|_g^2 + \left(1-\frac{2}{n}\right) (\divg \psi)^2 d\mu_g -2 \int_M \ric(\psi, \psi) d\mu_g\\
	& \geq 2 \int_M \left(\left|\nabla\psi\right|_g^2 - \ric(\psi, \psi)\right) d\mu_g.
\end{align*}

Let $\chi$ be some smooth compactly supported function such that $-\ric + \chi^2 g \geq \frac{n-1}{2} g$, then
\begin{align*}
\int_M \left|\lie \psi\right|_g^2 d\mu_g + 2 \int_M \chi^2 |\psi|_g^2 d\mu_g
	& \geq 2 \int_M \left(\left|\nabla\psi\right|_g^2 + (-\ric + \chi^2 g)(\psi, \psi)\right) d\mu_g\\
	& \geq 2 \int_M \left(\left|\nabla\psi\right|_g^2 + \frac{n-1}{2} \left|\psi\right|_g^2\right) d\mu_g\\
	& \geq 2 \left\|\psi\right\|_{W^{1, 2}_0}^2.
\end{align*}

We claim that there exists $\beta > 0$ such that for any $\psi \in W^{1, 2}_0$ \[\int_M \left|\lie \psi\right|_g^2 d\mu_g \geq \beta \int_M \chi^2 |\psi|_g^2 d\mu_g.\] Indeed, if such a constant does not exist, then it is possible to construct a sequence $\psi_k \neq 0$ ($k \geq 1$) such that \[\int_M \left|\lie \psi_k\right|_g^2 d\mu_g \leq \frac{1}{k} \int_M \chi^2 |\psi_k|_g^2 d\mu_g.\] By renormalizing $\psi_k$, we can assume that $\int_M \chi^2 \left|\psi_k\right|_g^2 d\mu_g = 1$. From the previous calculation, we deduce that the sequence $(\psi_k)_k$ is bounded in $W^{1, 2}_0$ so up to extracting a subsequence, we can assume that the sequence $(\psi_k)_k$ converges weakly to some $\psi_\infty \in W^{1, 2}_0$. Remark also that since $\chi \psi_k$ has support included in $\supp \chi$, the Rellich theorem applies, so we can also assume that $(\chi\psi_k)_k$ converges for the $L^2_0$-norm. We have that the $L^2$-limit of $(\chi \psi_k)_k$ coincide with $\chi \psi_\infty$. In particular, this shows that \[\int_M \left|\chi \psi_\infty\right|^2_g d\mu_g = 1,\] so $\psi_\infty \neq 0$. However, \[\int_M \left|\lie \psi_\infty\right|_g^2 d\mu_g = 0,\] so $\psi_\infty$ is the dual of a non-zero $L^2$-conformal Killing vector. This is a contradiction. Thus for some constant $\beta > 0$ and any $\psi \in W^{1, 2}_0$, \[\int_M \left|\lie \psi\right|_g^2 d\mu_g \geq \beta \int_M \chi^2 |\psi|_g^2 d\mu_g.\]

As a consequence, for all $\psi \in W^{1, 2}_0$,
\begin{align*}
\left(1 + \frac{2}{\beta}\right)\int_M \left|\lie \psi\right|_g^2 d\mu_g
	& \geq \int_M \left|\lie \psi\right|_g^2 d\mu_g + 2 \int_M \chi^2 |\psi|_g^2 d\mu_g\\
	& \geq 2 \left\|\psi\right\|_{W^{1, 2}_0}^2.
\end{align*}

This ends the proof of the positivity of $C_g$.
\end{proof}

\bibliographystyle{amsalpha}
\bibliography{biblio}

\end{document}